\documentclass[aps,prb,reprint,amsmath,amssymb,longbibliography,nofootinbib]{revtex4-2}
\usepackage{graphicx}

\graphicspath{{./figs/}}  
\usepackage[usenames]{color}
\usepackage{tikz}
\usepackage{comment}
\usepackage{amsmath,amssymb,amsfonts}	
\usepackage{bm,braket,array}
\usepackage{multirow}
\usepackage[all]{xy}
\usepackage{amsthm}
\usepackage{amscd}
\usepackage{slashed} 
\usepackage{framed}
\usepackage{slashed}

\usepackage[hidelinks]{hyperref}

\makeatletter
\let\PTsavedendpbox\@endpbox
\protected\def\@endpbox{\PTsavedendpbox}
\makeatother

\newtheorem{thm}{Theorem}[section]

\newtheorem{lem}[thm]{Lemma}
\newtheorem{cor}[thm]{Corollary}

\theoremstyle{definition}

\theoremstyle{remark}

\def\re{{\rm Re\,}}
\def\im{{\rm Im\,}}
\def\pf{{\rm Pf\,}}

\def\Tr{{\rm Tr }}
\def\tr{{\rm tr\,}}

\def\dim{{\rm dim\,}}

\def\ker{{\rm Ker\,}}
\def\coker{{\rm Coker\,}}

\def\Mat{{\rm Mat}}

\def\ua{\uparrow}
\def\da{\downarrow}

\newcommand{\s}{\sigma}

\newcommand{\G}{\Gamma}

\newcommand{\C}{\mathbb{C}}

\newcommand{\Q}{\mathbb{Q}}
\newcommand{\R}{\mathbb{R}}
\newcommand{\Z}{\mathbb{Z}}
\newcommand{\T}{\mathbb{T}}

\newcommand{\bx}{{\bm{x}}}
\newcommand{\bk}{{\bm{k}}}
\newcommand{\br}{{\bm{r}}}

\newcommand{\bR}{{\bm{R}}}
\newcommand{\bG}{{\bm{G}}}

\def\widebar{\accentset{{\cc@style\underline{\mskip10mu}}}} 
\def\wideubar{\underaccent{{\cc@style\underline{\mskip10mu}}}} 

\begin{document}
\title{Symmetry-enforced topological parity: revisiting \(2\mathbb Z\) classifications beyond cellwise symmetry actions}
\author{Ken Shiozaki}
\affiliation{Center for Gravitational Physics and Quantum Information, Yukawa Institute for Theoretical Physics, Kyoto University, Kyoto 606-8502, Japan}
\date{\today}

\begin{abstract}
Momentum-dependent symmetry actions can turn the even-integer constraint of a $2\Z$ classification of topological insulators and superconductors into an odd-integer constraint without changing the symmetry algebra.
For internal and order-two crystalline symmetries in up to four spatial dimensions, we derive parity relations between the Chern or winding number of a gapped Hamiltonian and topological invariants of its momentum-dependent antiunitary symmetry matrices.
Whenever the symmetry action enforces odd parity, any symmetry-preserving gapped phase must be topologically nontrivial.
Using cellwise actions, which are represented by momentum-independent matrices in a fixed unit-cell basis, as a reference, we determine how locality constrains the realization of odd topological numbers.
Explicit model constructions and no-go theorems distinguish three cases: odd values can be realized by finite-range symmetry actions, can be realized by exponentially decaying quasilocal actions but not by finite-range actions, or are impossible in finite-dimensional Bloch systems.
\end{abstract}

\maketitle

\tableofcontents

\section{Introduction}
\label{sec:introduction}

The free-fermion classification of topological insulators and superconductors is based on the ten Altland--Zirnbauer (AZ) classes of internal symmetries. 
The classes specify which of time-reversal symmetry (TRS), particle-hole symmetry (PHS), and chiral symmetry are present, together with their algebraic relations~\cite{Altland-Zirnbauer}. For each class and spatial dimension, the strong stable classification is $0$, $\Z$, $2\Z$ or $\Z_2$~\cite{Schnyder-Ryu-Furusaki-Ludwig,Kitaev,RyuSchnyderFurusakiLudwig2010}. 
An entry $2\Z$ restricts the Chern or winding number to even integers, even though the classification group is abstractly $\Z$. 
The simplest example is zero-dimensional class AII: Kramers degeneracy under TRS squaring to $-1$ requires an even number of negative-energy states.

We revisit the even-integer constraint by allowing the symmetry action to depend on momentum while keeping its algebra fixed. 
Topological invariants carried by the symmetry matrices then determine the parity of the integer invariant of a gapped Hamiltonian. 
Schematically, our main results take the form 
\begin{align*}
    &\text{Topological number of Hamiltonian}\\
    &\quad\begin{aligned}[t]
      \equiv{}&\text{Topological number of}\\
              &\text{symmetry operator}\pmod2.
    \end{aligned}
\end{align*}
An appropriate symmetry action can therefore enforce odd values without changing the symmetry algebra, so that a symmetry-preserving gap requires a nontrivial topological phase. 

We consider internal and order-two crystalline symmetries in up to four spatial dimensions, with finitely many single-particle degrees of freedom per unit cell. Equivalently, we work with tight-binding models in which high-energy continuum states are neglected. 
TRS and PHS commute or anticommute with the Hamiltonian, respectively, and their squares are fixed to $\pm1$. 
In condensed-matter systems, an AZ class does not uniquely specify its many-body realization; in particular, physical particle-hole exchange must be distinguished from the redundancy of the Bogoliubov--de Gennes (BdG) representation~\cite{ChiuTeoSchnyderRyu2016}. 
Appendix~\ref{sec:summary-manybody} explains the many-body realizations and their distinctions. 
Treating momentum coordinates fixed by the antiunitary symmetry as adiabatic parameters also extends our results to adiabatic processes in up to three spatial dimensions, with symmetry operators that depend on those parameters.

We define locality in a fixed unit cell and local orbital basis. 
We call a symmetry action {\it cellwise} if it maps degrees of freedom in cell $\bR$ only to those in the image cell $\tau\bR$ under the spatial operation $\tau$. 
Equivalently, its matrix is momentum independent in the periodic Bloch basis. 
Cellwise actions provide the reference for the usual $2\Z$ constraint. 
For an internal symmetry, a cellwise action is onsite when the unit cell is regarded as the unit of local degrees of freedom. 
Our definition also covers actions involving reflections or rotations. We thus distinguish cellwise from the onsite/non-onsite terminology used for internal symmetries in the literature~\cite{ChenLiuWen2011,Seifnashri2024}.

After factoring out the geometric motion $\bR\mapsto\tau\bR$, we call an action finite range if its spread is bounded and quasilocal if its matrix elements decay exponentially with distance. 
These conditions obey 
\begin{align*}
    \text{cellwise}\subseteq\text{finite range}\subseteq\text{quasilocal}.
\end{align*} 
Even an ordinary crystalline symmetry produces momentum-dependent $U(1)$ phases when it maps an orbital to a different unit cell. 
An action that maps the degrees of freedom at each physical site only to those at its geometric image is ultralocal in the terminology of Ref.~\cite{ThorngrenElse2018}. Such an action can be finite range without being cellwise. 
Momentum dependence alone therefore does not establish nonlocality or an anomaly. 
Section~\ref{sec:summary-locality} gives precise definitions and explains their dependence on the basis. 
We treat the locality of the Hamiltonian and that of the symmetry action as separate properties.

Locality imposes further restrictions on whether odd values can be realized. 
Explicit models and no-go theorems distinguish cases that admit finite-range realizations from those that admit quasilocal realizations but no finite-range ones. 
For the latter, we derive constraints by treating the symmetry matrix as a flat Hamiltonian and applying flat-band no-go theorems. 
An isolated, exactly flat band in a finite-range model cannot have a nonzero Chern number~\cite{ChenMazaheriSeidelTang2014,DubailRead2015}. 
More generally, in the AZ classification without crystalline symmetries, gapped finite-range flat Hamiltonians cannot realize strong stable topological phases in two or more dimensions~\cite{Read-2017}. 
The flat-band result does not apply directly to crystalline symmetries, for which we derive separate no-go theorems. 
In some cases, such as four-dimensional class AI, odd values remain forbidden for finite-dimensional Bloch Hamiltonians and symmetry matrices even when arbitrary continuous momentum dependence is allowed.

Symmetry-protected topological (SPT) versions of Lieb--Schultz--Mattis (LSM) theorems identify symmetry actions for which a compatible gapped phase must have nontrivial topology~\cite{YangJiangVishwanathRan2018,ElseThorngren2020,JiangChengQiLu2021,Lu2024MagneticLSM,ChengWang2022Rotation}. 
Our parity constraints are an example: a symmetry action on the full Bloch space, including occupied and unoccupied bands, constrains the integer invariant of a gapped Hamiltonian. 
This perspective is related to the real-space Atiyah--Hirzebruch spectral sequence (AHSS), which describes anomaly cancellation between cells with higher and lower symmetry~\cite{Shiozaki-Xiong-Gomi}. 
We note that SPT--LSM theorems differ from symmetry indicators~\cite{FuKane2007,FangGilbertBernevig2012,PoVishwanathWatanabe2017}, which use the representations of {\rm occupied bands} at high-symmetry momenta.

For chiral classes, we relate the constraints for a cellwise chiral operator $\G$ to those for momentum-dependent $\G$. 
A momentum-dependent chiral operator makes the integer winding number depend on the chosen trivialization of $\G$, so the integer must be treated as a relative quantity~\cite{Thiang}. 
Even though the winding number is a relative quantity, we introduce a $\Z_2$ invariant ${\rm WP}_d[H,\G]$ that is independent of the choice of eigenframes of $\G$ (Sec.~\ref{sec:summary-chiral-general-frame}).
The $\Z_2$ invariant ${\rm WP}_d[H,\G]$ is unchanged by continuous deformations of $(H,\G)$ that preserve the gap and the chiral symmetry algebra, and is constrained by topological invariants of the symmetry matrices.

We also realize models with odd topological invariants as gapped boundaries of finite-range bulks in one higher dimension, where cellwise bulk symmetries induce momentum-dependent boundary actions. These bulks admit symmetry-preserving deformations to an atomic limit; Sec.~\ref{sec:oddmodel-construction} develops the constructions and discusses their relation to anomalous symmetry-protected topological (ASPT) phases.

The paper is organized as follows. Section~\ref{sec:setup-summary} defines the symmetry actions, locality conditions, and invariants, and summarizes the parity constraints with a cellwise chiral operator when present. Section~\ref{sec:summary-integer-invariants} treats momentum-dependent chiral operators, defines the mod 2 invariant, and gives the corresponding parity constraints and locality results. Sections~\ref{sec:1d}--\ref{sec:c2f} prove the parity constraints in one through four dimensions and examine explicit models and the locality of the symmetry actions. These sections also introduce the Wess--Zumino terms, Polyakov--Wiegmann formulas, and a patchwise gluing formula for the Chern--Simons-type action needed in the proofs. 
Section~\ref{sec:oddmodel-construction} constructs finite-range higher-dimensional bulks whose gapped boundaries realize the odd-invariant models and discusses their relation to ASPT phases.
Section~\ref{sec:conclusion} presents our conclusions and open questions. 
Appendix~\ref{sec:summary-manybody} presents the fermionic many-body realizations, including the distinction between physical particle-hole symmetry and Nambu redundancy.
The spin-structure dependence of the Chern--Simons action is discussed in Appendix~\ref{sec:c2f-spin-note}.
The proof of the finite-range no-go theorem for four-dimensional class D is given in Appendix~\ref{app:d4-finite-range}.

\section{\texorpdfstring{\NoCaseChange{Setup and summary of results}}{Setup and summary of results}}
\label{sec:setup-summary}

\subsection{Hamiltonians and symmetry actions}
\label{sec:summary-H-symmetries}

We set the lattice constant to unity and denote the $d$-dimensional Brillouin torus by $T^d=\R^d/(2\pi\Z)^d$. 
Throughout, the unit cell and its local orbitals are fixed. 
For the matrix description, we associate the finitely many orbital degrees of freedom in each unit cell with its center, even when they occupy distinct Wyckoff positions. 
In the associated periodic single-particle basis, or Nambu basis for a BdG system, the Hamiltonian $H_\bk$ is a finite-dimensional Hermitian matrix. 
The Hamiltonian $H_\bk$ is continuous and periodic,
\begin{align}
    H_{\bk+\bG}=H_\bk,\qquad
    \bG=2\pi\bm n,\quad \bm n\in\Z^d,
\end{align}
where $\bG$ is a reciprocal lattice vector.
Symmetry operators are represented by finite-dimensional matrices on the same single-particle Hilbert space. 
We also impose the gap condition
\begin{align}
    0 \notin {\rm Spec} H_\bk, \quad\bk\in T^d. 
\end{align}

Let $\tau$ be an order-two action in real space that reverses $d_\tau$ coordinates:
\begin{align}
    \tau(x_1,\ldots,x_d)
    =(-x_1,\ldots,-x_{d_\tau},x_{d_\tau+1},\ldots,x_d).
\end{align}
The antiunitary operation also complex conjugates the Bloch states, so momentum transforms as
\begin{align}
    \bk\longmapsto -\tau\bk
    =(k_1,\ldots,k_{d_\tau},-k_{d_\tau+1},\ldots,-k_d).
    \label{eq:summary-momentum-action}
\end{align}
Thus, $d_\tau=0$ gives the usual momentum reversal for internal TRS or PHS, whereas $d_\tau=d$ leaves momentum fixed. 
Intermediate values describe combined operations: for example, $d_\tau=1$ combines an antiunitary operation with a reflection. 
In this paper, the AZ class refers to the algebra of the combined antiunitary operators.
Neither the spatial operation nor the internal TRS or PHS is required to be a symmetry separately.

We take the unitary matrices $T_\bk,C_\bk$ to be continuous and periodic on the entire Brillouin zone (BZ) torus, and write TRS and PHS as
\begin{align}
    {\rm TRS}:\quad
    &T_\bk H_\bk^*T_\bk^\dag=H_{-\tau\bk},
    &T_\bk^\top&=\epsilon_T T_{-\tau\bk},\label{eq:summary-TRS}\\
    {\rm PHS}:\quad
    &C_\bk H_\bk^*C_\bk^\dag=-H_{-\tau\bk},
    &C_\bk^\top&=\epsilon_C C_{-\tau\bk},\label{eq:summary-PHS}
\end{align}
where $\epsilon_T,\epsilon_C \in \{\pm 1\}$ specify the squares of the antiunitary actions. 
The matrices $T_\bk,C_\bk$ may depend on momentum. 
We call an action cellwise when its matrix $X_\bk$ is momentum independent in the fixed unit-cell orbital basis. 
For a spatial action $\tau$, a cellwise symmetry maps cell $\bm R$ to cell $\tau\bm R$ without any additional cell displacement. 
Section~\ref{sec:summary-locality} gives the precise definition.

When both TRS and PHS are present, $T$ and $C$ need not commute. Their commutation relation takes the form
\begin{align}
    T_{-\tau\bk}C_\bk^* = z_{\bk} C_{-\tau \bk}T_\bk^*,
\end{align}
where the $U(1)$ phase $z_{\bk}$ cannot in general be removed by redefining the phases of $T_\bk,C_\bk$.
We assume throughout that $z_\bk$ can be trivialized.
Equivalently, the chiral symmetry $\Gamma_\bk \sim T_{-\tau \bk} C_\bk^*$ obtained by combining $T$ and $C$ is required to be an internal $\Z_2$ antisymmetry.
Indeed, $\Gamma_\bk^2= \epsilon_T \epsilon_C z_\bk$. 
When $z_\bk$ can be trivialized, the square of $\Gamma$ is a $U(1)$ phase that can be removed by redefining the phase of $\Gamma_\bk$.
We therefore adopt the convention that $\Gamma_\bk$ squares to $1$.
The chiral symmetry condition is then
\begin{align}
    {\rm Chiral}:\quad
    &\G_\bk H_\bk \G_\bk^\dag=-H_{\bk},\quad 
    \Gamma_\bk^\dag = \Gamma_\bk.
    \label{eq:summary-chiral}
\end{align}
We refer to classes with both TRS and PHS as chiral classes.
With TRS and chiral symmetry as generators, our phase convention gives
\begin{align}
 T_\bk\G_\bk^*
 =\begin{cases}
 \G_{-\tau\bk}T_\bk & ({\rm BDI,CII}),\\
 -\G_{-\tau\bk}T_\bk & ({\rm DIII,CI})
 \end{cases}.
 \label{eq:TRS_G}
\end{align}

For a cellwise chiral operator, choose a momentum-independent basis that diagonalizes $\G$. The gap condition forces the two chiral sectors to have the same dimension $N$, so in this basis
\begin{align}
 \G&=\begin{pmatrix}I_N&0\\0&-I_N\end{pmatrix},\notag\\
 H_\bk&=\begin{pmatrix}0&D_\bk^\dag\\D_\bk&0\end{pmatrix},
 \qquad D_\bk\in{\rm GL}_N(\C).
 \label{eq:summary-cellwise-chiral-block}
\end{align}
We use TRS and chiral symmetry as the independent generators. Equation~\eqref{eq:TRS_G} makes $T_\bk$ block diagonal in BDI and CII and block off-diagonal in DIII and CI. We define the unitary matrices $u_\bk,v_\bk\in U(N)$ by
\begin{align}
 T_\bk&=\begin{pmatrix}v_\bk&0\\0&u_\bk\end{pmatrix}
 \qquad({\rm BDI,CII}),
 \label{eq:summary-cellwise-T-diagonal}
\end{align}
and
\begin{align}
 T_\bk&=\begin{pmatrix}0&\epsilon_T u_{-\tau\bk}\\u_\bk^\top&0\end{pmatrix}
 \qquad({\rm DIII,CI}).
 \label{eq:summary-cellwise-T-offdiagonal}
\end{align}
The off-diagonal form incorporates $T_\bk^\top=\epsilon_T T_{-\tau\bk}$. The TRS conditions therefore become
\begin{align}
 {\rm BDI}, {\rm CII}:\quad&u_\bk D_\bk^*v_\bk^\dag=D_{-\tau\bk},\notag\\
 &u_\bk^\top=\epsilon_T\, u_{-\tau\bk},\quad v_\bk^\top=\epsilon_T\, v_{-\tau\bk},\notag\\
 {\rm DIII}, {\rm CI}:\quad&(D_\bk u_\bk)^\top=\epsilon_T\, D_{-\tau\bk}u_{-\tau\bk}.
 \label{eq:summary-chiral-symmetries}
\end{align}
The same conditions hold after replacing $D_\bk$ by its unitary flattening $q_\bk$, defined in Sec.~\ref{sec:summary-Z-invariants}. These block conventions will be used throughout the cellwise chiral discussion. Section~\ref{sec:summary-chiral-general-frame} extends them to a momentum-dependent chiral operator by introducing a trivialization frame.

The many-body realizations of these single-particle symmetry conditions, for both particle-number-conserving and BdG systems, are given in Appendix~\ref{sec:summary-manybody}.

Table~\ref{tab:AZ} lists the real AZ classes and their symmetry signs. 
The underlying integer invariants belong to the complex classes A and AIII; the antiunitary-symmetry constraints associated with $2\Z$ classifications concern the eight real classes.

\begin{table*}[t]
    \centering
    \caption{Real AZ classes and their zero-dimensional stable classifications. The entries $\pm1$ for TRS and PHS are the signs in Eqs.~\eqref{eq:summary-TRS} and \eqref{eq:summary-PHS}, while 0 denotes the absence of a symmetry. The entry 1 for Chiral denotes its presence.
    The $R_n$ column gives representative descriptions of the classifying spaces. 
    The notation $2\Z$ means that the map to the classification without antiunitary symmetry (the complex class) is multiplication by two.}
    \label{tab:AZ}
    \small
    \setlength{\tabcolsep}{5pt}
    \begin{tabular}{ccccccc}
        \hline
        AZ & TRS & PHS & Chiral & $n\bmod8$ & $R_n$ & $\pi_0(R_n)$ \\
        \hline
        AI   & $+1$ & $0$ & $0$ & $0$ & $O(N+M)/(O(N)\times O(M))$ & $\Z$ \\
        BDI  & $+1$ & $+1$ & $1$ & $1$ & $O(N)$ & $\Z_2$ \\
        D    & $0$ & $+1$ & $0$ & $2$ & $O(2N)/U(N)$ & $\Z_2$ \\
        DIII & $-1$ & $+1$ & $1$ & $3$ & $U(2N)/Sp(N)$ & $0$ \\
        AII  & $-1$ & $0$ & $0$ & $4$ & $Sp(N+M)/(Sp(N)\times Sp(M))$ & $2\Z$ \\
        CII  & $-1$ & $-1$ & $1$ & $5$ & $Sp(N)$ & $0$ \\
        C    & $0$ & $-1$ & $0$ & $6$ & $Sp(N)/U(N)$ & $0$ \\
        CI   & $+1$ & $-1$ & $1$ & $7$ & $U(N)/O(N)$ & $0$ \\
        \hline
    \end{tabular}
\end{table*}

For cellwise symmetry actions, the strong stable classification depends on the numbers of reversed and fixed momentum coordinates and is given by
\begin{align}
    \pi_0\!\left(R_{n-d+2d_\tau}\right),
    \label{eq:summary-strong-group}
\end{align}
with the index taken modulo eight~\cite{Teo-Kane,Shiozaki-Sato}. 
The strong classification omits weak components associated with lower-dimensional subtori and is therefore not the full classification on the BZ torus. 
We focus on the cases
\begin{align}
    n-d+2d_\tau\equiv4\pmod8,
    \label{eq:summary-even-cases}
\end{align}
which have a $2\Z$ classification, and determine the parity of the integer invariant for a fixed momentum-dependent symmetry matrix.

\subsection{Definitions of locality}
\label{sec:summary-locality}

For a fixed choice of unit cell and local orbitals, let $\bm R$ and $a$ label the cells and orbitals, respectively, and denote the real-space basis by $\ket{\bm R,a}$. 
Our convention for the periodic Bloch basis is
\begin{align}
 \ket{\bk,a}=\sum_{\bm R\in\Z^d}e^{i\bk\cdot\bm R}\ket{\bm R,a},
 \label{eq:summary-locality-bloch}
\end{align}
even if there are multiple Wyckoff positions in each unit cell.
We factor out the geometric cell motion $\bm R\mapsto\tau\bm R$ and represent any additional motion and mixing by $X_\bk$. For example, writing the single-particle antiunitary action of TRS or PHS as $\mathcal X$, a cellwise action is represented by a constant matrix $X_0$ as
\begin{align}
 \mathcal X\ket{\bm R,a}
 &=\sum_b\ket{\tau\bm R,b}(X_0)_{ba},\notag\\
 \mathcal X\ket{\bk,a}
 &=\sum_b\ket{-\tau\bk,b}(X_0)_{ba}.
 \label{eq:summary-cellwise-action}
\end{align}
A cellwise action thus mixes orbitals only within the geometrically mapped cell, with no additional cell displacement. Equivalently, $X_\bk=X_0$ in the fixed Bloch basis. For an internal action, $\tau={\rm id}$, and cellwise corresponds to onsite if the unit cell is taken as the unit of local degrees of freedom. For a general action, we measure its spread by the distance $\|\bm R'-\tau\bm R\|$ between the image cell $\bm R'$ and the geometric image $\tau\bm R$.

Define the Fourier coefficients of a symmetry matrix by
\begin{align}
 \tilde X(\br)
 &=\int_{T^d}\frac{d^d k}{(2\pi)^d}\,
       e^{i\bk\cdot\br}X_\bk,\qquad \br\in\Z^d,\notag\\
 X_\bk&=\sum_{\br\in\Z^d}\tilde X(\br)e^{-i\bk\cdot\br}.
 \label{eq:summary-fourier}
\end{align}
The Fourier series converges absolutely for the finite-range and quasilocal actions defined below. A cellwise action has $\tilde X(\br)=X_0\delta_{\br,0}$. We call an action finite range if $\tilde X(\br)$ has finite support. Equivalently, $X_\bk$ is a finite Laurent polynomial in $e^{\pm ik_\mu}$. In real space, the image of an orbital in cell $\bm R$ is then confined to finitely many cells around $\tau\bm R$. We call the action quasilocal if there exist $C,a>0$ such that
\begin{align}
 \|\tilde X(\br)\|\le C e^{-a\sum_\mu|r_\mu|}.
 \label{eq:summary-quasilocal}
\end{align}
The three locality conditions obey
\begin{align}
 \text{cellwise}\ \subseteq\ \text{finite range}\
 \subseteq\ \text{quasilocal}.
 \label{eq:summary-locality-hierarchy}
\end{align}
Quasilocal therefore includes finite range. We use the label ``quasilocal only'' for cases that admit no finite-range realization.

An ordinary crystalline symmetry maps orbitals according to their physical positions. In the fixed assignment of orbitals to unit cells, this can require additional cell displacements relative to the geometric cell motion, producing Bloch $U(1)$ phases.
Thus, when orbitals at different Wyckoff positions are represented as internal degrees of freedom at the cell center, an ultralocal action on the physical sites can become a finite-range, non-cellwise matrix action. 
Whether an action is cellwise can depend on the unit cell and orbital basis. 
We do not assume that a change of unit cell can make every finite-range action cellwise.

Quasilocality is related to holomorphic extension to complex momentum. If $X_\bk$ admits a periodic holomorphic extension to the strip $|\operatorname{Im}k_\mu|<a$, then, for every $0<a'<a$, we have~\cite{Kohn1959,PanatiPisante2013}
\begin{align}
 \|\tilde X(\br)\|
 =O\!\left(e^{-a'\sum_\mu|r_\mu|}\right).
\end{align}
To obtain the exponential decay, shift the integration contour in Eq.~\eqref{eq:summary-fourier} by $+i a'\operatorname{sgn}r_\mu$ in each component; the decay factor comes from $e^{i\bk\cdot\br}$. 
Conversely, Eq.~\eqref{eq:summary-quasilocal} ensures that the Fourier series converges locally uniformly in a strip, defining a holomorphic extension. 
Smoothness alone does not imply exponential decay.

The hopping range of the Hamiltonian and the locality of the symmetry action are separate properties. 
A flattened Hamiltonian or a doubled Hermitian Hamiltonian constructed from a symmetry matrix is generally quasilocal but need not have finite range. 
Momentum-dependent basis changes can also alter locality, which we therefore evaluate in the original unit-cell and real-space orbital basis.

\subsection{\texorpdfstring{$\Z$}{Z} topological invariants}
\label{sec:summary-Z-invariants}
In even dimension $d=2m$, let $P_\bk$ be the projector onto the negative-energy bands. 
We define the Chern number by
\begin{align}
    {\rm ch}_m[H]
    =\frac{1}{m!}\left(\frac{i}{2\pi}\right)^m
       \int_{T^{2m}}\Tr\!\left[(P_\bk dP_\bk dP_\bk)^m\right]
       \in\Z.
    \label{eq:summary-chern}
\end{align}
The notation ${\rm ch}_m$ denotes the integral of the Chern character.
Equivalently, using the Berry connection $A$ for the occupied states of $H_\bk$ and its curvature $F = dA + A^2$, we can write
\begin{align}
    {\rm ch}_m[H] = \frac{1}{m!}\left(\frac{i}{2\pi}\right)^m
       \int_{T^{2m}}\tr(F^m).
\end{align}
In $d=0$, ${\rm ch}_0[H]$ counts the negative-energy states. In class AII, $TH^*T^\dag=H$ and $T^\top=-T$ imply Kramers pairing, so
\begin{align}
    {\rm ch}_0[H]:=\Tr P\in2\Z\qquad(\text{zero-dimensional AII}).
    \label{eq:summary-ch0}
\end{align}

The Chern number does not depend on the choice of connection within a given topological sector; the Berry connection is one possible choice.
When TRS is present, a time-reversal-symmetric connection is useful for calculations; we introduce it in Sec.~\ref{sec:c2f-tr-connection}.

In odd dimension $d=2m+1$, the $\Z$-valued invariant in the presence of chiral symmetry is the winding number. In the cellwise block form~\eqref{eq:summary-cellwise-chiral-block}, set $q_\bk=D_\bk(D_\bk^\dag D_\bk)^{-1/2}$. This is the unitary lower-left block of the flattened Hamiltonian. Momentum-dependent chiral operators and the choice of a trivialization frame are treated separately in Sec.~\ref{sec:summary-integer-invariants}. For a periodic map
\begin{align}
    q: T^{2m+1} \to U(N),
\end{align}
the winding number is defined by
\begin{align}
W_{2m+1}[q]
    ={}&\frac{(-1)^m m!}{(2m+1)!}
       \left(\frac{1}{2\pi i}\right)^{m+1}\notag\\
       &\times\int_{T^{2m+1}}\tr\!\left[(q^{-1}dq)^{2m+1}\right]
       \in\Z.
    \label{eq:summary-winding}
\end{align}
The cases needed here are
\begin{align}
\begin{aligned}
    W_1[q]&=\frac{1}{2\pi i}\oint\tr(q^{-1}dq),\\
    W_3[q]&=\frac{1}{24\pi^2}\int_{T^3}\tr[(q^{-1}dq)^3].
    \end{aligned}
\end{align}
The winding number also extends to maps $q_\bk$ with values in the invertible matrices $GL_N(\C)$. In the following, however, we use unitary representatives when evaluating winding numbers.

\subsection{Invariants of symmetry matrices}
\label{sec:summary-symmetry-invariants}

In this subsection, we summarize the topological invariants of symmetry matrices that enter the parity constraints.
Numerical functionals are written with square brackets, as in $\nu[X;\gamma]$. 
The first argument is the matrix, bundle, or other object being evaluated. A semicolon separates that argument from auxiliary data such as a path, surface, or spin structure.
The evaluation domain is omitted when it is clear from context.
Cohomology classes are written as $w_j(E),c_j(E)$. Their numerical evaluations are written with an explicit pairing with a fundamental class, such as $\langle w_j(E),[\Sigma]\rangle$.

Symmetry matrices have topological invariants of their own, distinct from those of the Hamiltonian. For the TRS block $u_\bk$ in classes DIII and CI, defined in Eq.~\eqref{eq:summary-cellwise-T-offdiagonal}, we use the winding number in Eq.~\eqref{eq:summary-winding}. A unitary matrix $X_\bk$ satisfying $X_\bk^\top=\pm X_{-\tau\bk}$ obeys the same condition as the off-diagonal block in class CI or DIII with $u_\bk=1$ in Eq.~\eqref{eq:summary-chiral-symmetries}. Its strong stable classification is
\begin{align}
    X_\bk^\top=X_{-\tau\bk}:&\quad
       \pi_0(R_{7-d+2d_\tau}),\notag\\
    X_\bk^\top=-X_{-\tau\bk}:&\quad
       \pi_0(R_{3-d+2d_\tau}).
    \label{eq:summary-symmetry-groups}
\end{align}
We retain the Brillouin torus and evaluate the invariants entering the parity constraints, rather than replacing momentum space by a $d$-dimensional sphere that retains only the strong component.

First, if $X_\bk^\top=-X_{-\tau\bk}$, then $X$ is antisymmetric at fixed points $P,Q$ and has a Pfaffian there. For an oriented path $\gamma_{PQ}$ from $P$ to $Q$, define the $\Z_2$ invariant~\cite{Qi-Hughes-Zhang-2010} by
\begin{align}
\begin{gathered}
    (-1)^{\nu[X;\gamma_{PQ}]}
    :=\frac{\pf[X_Q]}{\pf[X_P]}
      \exp\!\left[-\frac12\int_{\gamma_{PQ}}d\log\det X_\bk\right],\\
    \nu[X;\gamma_{PQ}]\in \{0,1\}.
    \end{gathered}
    \label{eq:summary-pfaffian}
\end{align}
The Pfaffian identity ensures that the right-hand side is $\pm1$, since the square of the Pfaffian equals the determinant. In one-dimensional CII, $\gamma_{0\pi}$ is the positively oriented path from $k=0$ to $\pi$. In two-dimensional C, we set
\begin{align}
\begin{gathered}
    \G=(0,0),\qquad X=(\pi,0),\\
    Y=(0,\pi),\qquad M=(\pi,\pi),
    \end{gathered}
\end{align}
and take $\gamma_{\G X},\gamma_{YM}$ along increasing $k_x$ on $k_y=0,\pi$, respectively.
The invariant $\nu$ is additive under direct sums:
\begin{align}
    \nu[X \oplus Y;\gamma_{PQ}] = \nu[X;\gamma_{PQ}] + \nu[Y;\gamma_{PQ}] \pmod 2.
\end{align}
Under a basis transformation $X_\bk \mapsto V^\dag_{-\tau \bk} X_\bk V_\bk^*$, the invariant shifts by a winding number:
\begin{align}
    \nu[V^\dag_{-\tau} X V^*;\gamma_{PQ}] 
    &= \nu[X;\gamma_{PQ}] + W_1[V] \pmod 2.
\end{align}
Here, $W_1[V]$ is the one-dimensional winding number along $\gamma_{QP}$ and the image $-\tau(\gamma_{PQ})$ under $-\tau$.

For a momentum-fixing symmetry $X_\bk^\top=X_\bk$, a local Takagi factorization of the symmetric unitary matrix and its transition functions,
\begin{align}
    X=U_iU_i^\top,\qquad U_j=U_iS_{ij},\qquad S_{ij}\in O(N),
    \label{eq:summary-real-bundle}
\end{align}
define a real bundle. We denote the associated principal $O(N)$ bundle by $\mathcal P_X$ and use $w_j(\mathcal P_X)$ for the Stiefel--Whitney classes of the associated real vector bundle.
The second Stiefel--Whitney (SW) class obeys the Whitney sum formula
\begin{align}
 w_2(E\oplus F)
 =w_2(E)+w_2(F)+w_1(E)\cup w_1(F).
 \label{eq:summary-SW-whitney}
\end{align}
Direct sums therefore require the cross term involving the first SW classes.

In three dimensions, we define the $\Z_2$ invariant $\mu[X]$ for $X_\bk^\top=-X_{(k_x,-k_y,-k_z)}$ in Eq.~\eqref{eq:alt62-nu} by introducing an $\R/2\Z$-valued refinement of the Wess--Zumino (WZ) term.
For a symmetric unitary matrix $X_\bk^\top=X_\bk$, we choose an auxiliary spin structure $\mathfrak s$ and define the CI invariant $\zeta[X;\mathfrak s]$ using a spin Chern--Simons (CS) term in Eqs.~\eqref{eq:alt64-nu-definition_1} and \eqref{eq:alt64-nu-definition_2}~\cite{ShiozakiCI2026}.
For three-dimensional BDI and four-dimensional D, we take the difference of two $\zeta$ invariants evaluated using the same $\mathfrak s$. Whenever a gapped Hamiltonian exists, the difference is independent of $\mathfrak s$.
Both invariants are additive,
\begin{align}
    \mu[X\oplus Y] &= \mu[X]+\mu[Y]\pmod 2,\notag\\
    \zeta[X\oplus Y;\mathfrak s] &= \zeta[X;\mathfrak s]+\zeta[Y;\mathfrak s]\pmod 2,
\end{align}
and a basis transformation $X_\bk \mapsto V^\dag_{-\tau \bk} X_\bk V_\bk^*$ shifts them by a winding number:
\begin{align}
    \mu[V^\dag_{-\tau} X V^*] &= \mu[X]+W_3[V]\pmod 2,\notag\\
    \zeta[V^\dag_{-\tau} X V^*;\mathfrak s] &= \zeta[X;\mathfrak s]+W_3[V]\pmod 2.
\end{align}
Sections~\ref{sec:3d-CII} and \ref{sec:3d-BDI} construct the $\Z_2$ invariants $\mu[X], \zeta[X;\mathfrak s]$, respectively.

\subsection{Overview of the results}

\begin{table*}[tp]
\centering
\caption{Parity constraints and realization of odd invariants for finite-dimensional Bloch Hamiltonians and symmetry matrices.
$d$ is the spatial dimension, AZ denotes the Altland--Zirnbauer class, $d_\tau$ counts the real-space coordinates reversed by the antiunitary symmetry, and $-\tau$ is its momentum-space action.
``No (finite dim.)'' excludes odd values in this finite-dimensional setting, regardless of locality. 
The other entries specify the locality of symmetry operators realizing odd values: finite range means that a finite-range realization exists, whereas quasilocal only means that an exponentially decaying realization exists but no finite-range realization does. 
For chiral classes, $\G$ is cellwise and is written as $\operatorname{diag}(I,-I)$. Section numbers link to the corresponding discussions.}
\label{tab:result}
\small
\setlength{\tabcolsep}{4pt}
\renewcommand{\arraystretch}{1.35}

\begin{tabular*}{\textwidth}{@{}ccc@{\hspace{12pt}\extracolsep{\fill}}llcc@{}}
\hline
$d$ & AZ & $d_\tau$ & Constraint ($\bmod2$) & Symmetry & Odd values & Sec. \\
\hline
0 & AII & 0 & ${\rm ch}_0[H]\equiv0$
& $\begin{aligned}[t]TH^*T^\dag&=H,\\T^\top&=-T\end{aligned}$
& No (finite dim.) & \ref{sec:summary-Z-invariants} \\[2pt]
\hline
1 & CII & 0 & $W_1[q]\equiv\nu[u;\gamma_{0\pi}]-\nu[v;\gamma_{0\pi}]$
& $\begin{aligned}[t]uq^*v^\dag&=q_{-\tau},\\u^\top&=-u_{-\tau},\\v^\top&=-v_{-\tau}\end{aligned}$
& finite range & \ref{sec:1d-CII} \\[2pt]
1 & DIII & 1 & $W_1[q]\equiv W_1[u]$
& $(qu)^\top=-qu$
& finite range & \ref{sec:1d-DIII} \\[2pt]
\hline
2 & C & 0 & ${\rm ch}_1[H]\equiv\nu[C;\gamma_{\G X}]-\nu[C;\gamma_{YM}]$
& $\begin{aligned}[t]CH^*C^\dag&=-H_{-\tau},\\C^\top&=-C_{-\tau}\end{aligned}$
& quasilocal only & \ref{sec:2d-C} \\[2pt]
2 & AII & 1 & ${\rm ch}_1[H]\equiv0$
& $\begin{aligned}[t]TH^*T^\dag&=H_{-\tau},\\T^\top&=-T_{-\tau}\end{aligned}$
& No (finite dim.) & \ref{sec:2d-AII} \\[2pt]
2 & D & 2 & ${\rm ch}_1[H]\equiv\langle w_2(\mathcal P_C),[T^2]\rangle$
& $\begin{aligned}[t]CH^*C^\dag&=-H,\\C^\top&=C\end{aligned}$
& finite range & \ref{sec:2d_classD} \\[2pt]
\hline
3 & CI & 0 & $W_3[q]\equiv W_3[u]$
& $(qu)^\top=(qu)_{-\tau}$
& quasilocal only & \ref{sec:3d-CI} \\[2pt]
3 & CII & 1 & $W_3[q]\equiv\mu[u]-\mu[v]$
& $\begin{aligned}[t]uq^*v^\dag&=q_{-\tau},\\u^\top&=-u_{-\tau},\\v^\top&=-v_{-\tau}\end{aligned}$
& quasilocal only & \ref{sec:3d-CII} \\[2pt]
3 & DIII & 2 & $W_3[q]\equiv W_3[u]$
& $(qu)^\top=-(qu)_{-\tau}$
& quasilocal only & \ref{sec:3d-DIII} \\[2pt]
3 & BDI & 3 & $W_3[q]\equiv\zeta[u;\mathfrak s]-\zeta[v;\mathfrak s]$
& $\begin{aligned}[t]uq^*v^\dag&=q,\\u^\top&=u,\\v^\top&=v\end{aligned}$
& finite range & \ref{sec:3d-BDI} \\[2pt]
\hline
4 & AI & 0 & ${\rm ch}_2[H]\equiv0$
& $\begin{aligned}[t]TH^*T^\dag&=H_{-\tau},\\T^\top&=T_{-\tau}\end{aligned}$
& No (finite dim.) & \ref{sec:c2f-ai40} \\[2pt]
4 & C & 1 & ${\rm ch}_2[H]\equiv\mu[C_0]-\mu[C_\pi]$
& $\begin{aligned}[t]CH^*C^\dag&=-H_{-\tau},\\C^\top&=-C_{-\tau}\end{aligned}$
& quasilocal only & \ref{sec:c2f-c41} \\[2pt]
4 & AII & 2 & ${\rm ch}_2[H]\equiv0$
& $\begin{aligned}[t]TH^*T^\dag&=H_{-\tau},\\T^\top&=-T_{-\tau}\end{aligned}$
& No (finite dim.) & \ref{sec:c2f-aii42} \\[2pt]
4 & D & 3 & ${\rm ch}_2[H]\equiv\zeta[C_0;\mathfrak s]-\zeta[C_\pi;\mathfrak s]$
& $\begin{aligned}[t]CH^*C^\dag&=-H_{-\tau},\\C^\top&=C_{-\tau}\end{aligned}$
& quasilocal only & \ref{sec:c2f-d43} \\[2pt]
4 & AI & 4 & ${\rm ch}_2[H]\equiv0$
& $\begin{aligned}[t]TH^*T^\dag&=H,\\T^\top&=T\end{aligned}$
& No (finite dim.) & \ref{sec:c2f-ai44} \\[2pt]
\hline
\end{tabular*}
\end{table*}

Table~\ref{tab:result} summarizes the parity constraints proved in this paper for all 15 cases with $d\le4$ satisfying Eq.~\eqref{eq:summary-even-cases}. 
In chiral classes, we take $\G$ to be cellwise and use the off-diagonal block $q$ defined in Sec.~\ref{sec:summary-Z-invariants}. The matrices $u,v$ are defined by the TRS block forms in Sec.~\ref{sec:summary-H-symmetries} and satisfy Eq.~\eqref{eq:summary-chiral-symmetries}. Section~\ref{sec:summary-integer-invariants} treats momentum-dependent chiral operators and presents the corresponding results in a separate table.
Table~\ref{tab:result} assumes that a finite-dimensional gapped $H$ compatible with the specified symmetries exists. All congruences are taken $\bmod2$. 
The invariants are defined in Secs.~\ref{sec:summary-Z-invariants} and \ref{sec:summary-symmetry-invariants}. In the symmetry column, we suppress the momentum subscript $\bk$ and write $X_{-\tau}(\bk):=X_{-\tau\bk}$. For four-dimensional C and D, the abbreviation $C_s$ is defined by
\begin{align}
 C_s(k_1,k_2,k_3):=C_{(k_1,k_2,k_3,s)},\qquad s=0,\pi.
 \label{eq:summary-boundary-C}
\end{align}

In the ``Odd values'' column, ``No (finite dim.)'' means that odd values cannot be realized by finite-dimensional Bloch Hamiltonians and symmetry matrices, regardless of locality. The entry ``finite range'' means that odd values can be realized by a finite-range symmetry action. 
The entry ``quasilocal only'' means that a finite-range realization is impossible but an exponentially decaying action can realize odd values. 
This column concerns the realization of symmetry actions, not the hopping range of the Hamiltonian. 

The final column refers to the corresponding section.
Our finite-dimensional matrix description excludes boundaries with infinitely many degrees of freedom in the normal direction at each fixed boundary momentum, such as the surface of a higher-dimensional half-space. 
The table therefore does not determine whether such boundaries can realize odd values. 

\section{\texorpdfstring{\NoCaseChange{Chiral classes and a mod 2 invariant}}{Chiral classes and a mod 2 invariant}}
\label{sec:summary-integer-invariants}

When the chiral operator depends on momentum, interpreting the winding number requires some care. We begin with a simple model that illustrates how the choice of a trivialization frame is related to the boundary condition, and then give the general formulation. Finally, we express the parity constraints in the original orbital basis and summarize them in Table~\ref{tab:result-wp-full-T}.

\subsection{An example with a momentum-dependent chiral operator}
\label{sec:summary-nononsite-chiral-example}

Consider a one-dimensional system with two spin states $\ua,\da$ per unit cell. We denote the real-space basis by $\ket{j,\s}$, with $j\in\Z$ and $\s\in\{\ua,\da\}$. Define the momentum basis by
\begin{align}
 \ket{k,\s}=\sum_j\ket{j,\s}e^{ikj},
 \label{eq:summary-nononsite-chiral-bloch}
\end{align}
and, for $m>0$, consider the Hamiltonian and chiral operator
\begin{align}
 H_k=m\sigma_z,\qquad
 \G_k=\begin{pmatrix}0&e^{-ik}\\e^{ik}&0\end{pmatrix}.
 \label{eq:summary-nononsite-chiral-model}
\end{align}
The operators satisfy $\G_k^2=I$ and $\{\G_k,H_k\}=0$, and the Hamiltonian has energies $\pm m$. The Hamiltonian $H$ contains only a mass term within each unit cell, whereas $\G$ exchanges spins in neighboring cells and is therefore finite range but non-cellwise:
\begin{align}
 \hat\G\ket{j,\ua}=\ket{j-1,\da},\qquad
 \hat\G\ket{j,\da}=\ket{j+1,\ua}.
 \label{eq:summary-nononsite-chiral-action}
\end{align}
Restricting the original degrees of freedom to $j=1,\ldots,L$ leaves $H$ without boundary zero modes. However, $\G$ maps $\ket{1,\ua}$ and $\ket{L,\da}$ outside the retained region, so simply truncating $\G$ destroys its unitarity.

To calculate the winding number, choose eigenvectors of $\G_k$ with eigenvalues $+1,-1$ as
\begin{align}
\begin{gathered}
 u_k=\frac1{\sqrt2}\begin{pmatrix}1\\e^{ik}\end{pmatrix},\qquad
 v_k=\frac1{\sqrt2}\begin{pmatrix}e^{-ik}\\-1\end{pmatrix},\\
 V_k=(u_k,v_k),
 \end{gathered}
 \label{eq:summary-nononsite-chiral-frame}
\end{align}
respectively. The matrix $V_k$ is unitary, and
\begin{align}
 \widetilde H_k:=V_k^\dag H_kV_k
 &=m\begin{pmatrix}0&e^{-ik}\\e^{ik}&0\end{pmatrix},\notag\\
 \widetilde\G_k:=V_k^\dag\G_kV_k
 &=\begin{pmatrix}1&0\\0&-1\end{pmatrix}.
 \label{eq:summary-nononsite-chiral-transformed}
\end{align}
The lower-left block of $\widetilde H_k$ is $D_k=me^{ik}$. After flattening, $q_k=e^{ik}$ gives $W_1[q]=1$.

The boundary condition corresponding to this winding number truncates the Wannier basis defined by $u_k,v_k$, rather than the original spin basis. To make the distinction explicit, define
\begin{align}
\ket{j,A}
 &:=\int_{-\pi}^{\pi}\frac{dk}{2\pi}
 \bigl(\ket{k,\ua},\ket{k,\da}\bigr)u_ke^{-ikj}\notag\\
 &=\frac{\ket{j,\ua}+\ket{j-1,\da}}{\sqrt2},\notag\\
 \ket{j,B}
 &:=\int_{-\pi}^{\pi}\frac{dk}{2\pi}
 \bigl(\ket{k,\ua},\ket{k,\da}\bigr)v_ke^{-ikj}\notag\\
 &=\frac{\ket{j+1,\ua}-\ket{j,\da}}{\sqrt2}.
 \label{eq:summary-nononsite-chiral-wannier}
\end{align}
The Hamiltonian acts as $\hat H\ket{j,A}=m\ket{j-1,B}$ and $\hat H\ket{j,B}=m\ket{j+1,A}$. Retaining the Wannier orbitals with $j=1,\ldots,L$ leaves the zero modes $\ket{1,A}$ and $\ket{L,B}$. Truncation in the Wannier basis preserves the unitarity of $\G$ and its anticommutation with $H$. The two zero modes have chiral eigenvalues $+1,-1$, respectively. The chiral index---the number of $+1$ zero modes minus the number of $-1$ zero modes---is therefore $W_1[q]=1$ at the left boundary and $-1$ at the right boundary.

The basis that trivializes $\G_k$ is not unique. For example, we may instead choose
\begin{align}
\begin{gathered}
 u'_k=u_ke^{-ik}
 =\frac1{\sqrt2}\begin{pmatrix}e^{-ik}\\1\end{pmatrix},\\
 v'_k=v_k,\qquad V'_k=(u'_k,v'_k).
 \end{gathered}
 \label{eq:summary-nononsite-chiral-shifted-frame}
\end{align}
In this basis,
\begin{align}
 \widetilde H'_k:=V_k'^\dag H_kV'_k
 &=m\begin{pmatrix}0&1\\1&0\end{pmatrix},\notag\\
 \widetilde\G'_k:=V_k'^\dag\G_kV'_k
 &=\begin{pmatrix}1&0\\0&-1\end{pmatrix},
 \label{eq:summary-nononsite-chiral-shifted-model}
\end{align}
so flattening gives $q'_k=1$ and hence $W_1[q']=0$. The relation $u'_k=u_ke^{-ik}$ between the two bases corresponds to a lattice translation of the Wannier orbitals,
\begin{align}
 \ket{j,A}'=\ket{j+1,A},\qquad
 \ket{j,B}'=\ket{j,B}.
 \label{eq:summary-nononsite-chiral-wannier-shift}
\end{align}
Since $H$ couples $\ket{j,A}'$ to $\ket{j,B}'$ for each $j$, restricting this Wannier basis to $j=1,\ldots,L$ produces no zero modes.

The two descriptions differ only in the choice of basis; the infinite-system Hamiltonian and chiral operator are unchanged. Truncating the two Wannier bases at the same unit-cell labels nevertheless retains different subspaces and can therefore produce different boundary spectra. Describing the same physical boundary in the two bases requires transforming the truncation projector as well.

An integer winding number therefore requires a choice of basis that trivializes $\G$, and its relation to boundary states refers to the open boundary condition defined in that basis. Any discussion of a Hamiltonian together with a momentum-dependent chiral operator must account for the freedom to choose the trivialization frame.

\subsection{General chiral operators and trivialization frames}
\label{sec:summary-chiral-general-frame}

We now formulate the freedom to choose eigenframes of $\G$ while keeping the Bloch basis associated with the original real-space orbitals fixed.
Let $H_\bk,\G_\bk$ be smooth and periodic, with 
\begin{align}
    \G_\bk^\dag=\G_\bk, \quad \G_\bk^2=I, \quad \{\G_\bk,H_\bk\}=0. 
    \label{eq:chiral_sym}
\end{align}
We fix $\G$ and denote its eigenbundles with eigenvalues $\pm1$ by $E_\pm$.
A gapped $H$ defines a map $H|_{E_+}:E_+\to E_-$. The map and its homotopy class under gap-preserving deformations are independent of local basis choices.
To assign an integer to each homotopy class, however, one must choose a reference class to represent zero~\cite{Thiang}.
We consider winding numbers in dimensions $d=1,3$ and first establish the existence of a trivialization frame.

\begin{lem}
\label{lem:summary-chiral-global-frame}
On a torus $T^d$ of arbitrary dimension, the existence of a gapped $H$ satisfying the conditions in \eqref{eq:chiral_sym} implies that $\tr\G_\bk=0$, and $E_\pm$ are isomorphic complex vector bundles of the same rank $N$.
All their positive-degree Chern classes vanish. In particular, for $d\le4$, both bundles are trivial and admit smooth, periodic global orthonormal frames.
\end{lem}
\begin{proof}
Define the flattened Hamiltonian by
\begin{align}
 Q_\bk:=H_\bk(H_\bk^2)^{-1/2},\qquad Q_\bk^2=I.
 \label{eq:summary-chiral-flattening}
\end{align}
The flattened Hamiltonian anticommutes with $\G$ and therefore exchanges the $+1$ and $-1$ eigenspaces of $\G$. The eigenspaces consequently have the same dimension, so $\tr\G_\bk=0$.
If ${\cal U}_{\alpha,\bk}$ is a local orthonormal frame of $E_+$, then $Q_\bk{\cal U}_{\alpha,\bk}$ is a local orthonormal frame of $E_-$.
On overlaps,
\begin{align*}
 {\cal U}_{\beta,\bk}={\cal U}_{\alpha,\bk}g_{\alpha\beta,\bk}
 \quad\Longrightarrow\quad
 Q_\bk{\cal U}_{\beta,\bk}
 =(Q_\bk{\cal U}_{\alpha,\bk})g_{\alpha\beta,\bk},
\end{align*}
so the transition functions $g_{\alpha\beta,\bk}$ are identical and $E_+\simeq E_-$.
Since the full single-particle bundle $E_+\oplus E_-$ is trivial, the total Chern classes satisfy $c(E_+)^2=c(E_+)c(E_-)=1$.
The groups $H^{2j}(T^d;\Z)$ are torsion free. Comparing the positive-degree components in increasing degree therefore gives $c_j(E_\pm)=0$ for $j\ge1$.
Finally, for $d\le4$, the vanishing Chern classes imply that both bundles admit smooth, periodic global orthonormal frames~\cite{MonacoRoussigne2023}\footnote{For $d \geq 5$, a complex vector bundle with vanishing Chern classes need not be trivial.
In five dimensions, $\pi_4[U(2)]=\Z_2$ allows a nontrivial rank-two complex vector bundle whose Chern classes vanish.}
.
\end{proof}

The contrapositive of the lemma gives a criterion for gaplessness.
\begin{cor}
If either eigenbundle $E_+$ or $E_-$ of $\G$ has a nonzero Chern number, every Hamiltonian $H$ that anticommutes with $\G$ must be gapless.
\end{cor}

Choose global frames ${\cal U}_\pm$ for $E_\pm$ and assemble them into the unitary matrix $V_\bk = ({\cal U}_+,{\cal U}_-)$. In this basis,
\begin{align}
 V_\bk^\dag\G_\bk V_\bk
 &=\begin{pmatrix}I_N&0\\0&-I_N\end{pmatrix},\notag\\
 V_\bk^\dag H_\bk V_\bk
 &=\begin{pmatrix}0&D_\bk^\dag\\D_\bk&0\end{pmatrix},
 \qquad D_\bk\in{\rm GL}_N(\C),
 \label{eq:summary-chiral-block}
\end{align}
where
$D_\bk = ({\cal U}_-)^\dag H_\bk {\cal U}_+$.
Following our winding-number convention, we take the flattened lower-left block to be
\begin{align}
 q_\bk:=D_\bk(D_\bk^\dag D_\bk)^{-1/2}\in U(N).
\end{align}
The matrix $q$ represents the flattened Hamiltonian $Q:E_+\to E_-$ in the chosen frames and defines the winding number $W_d[q]$.

The winding number depends on the frames ${\cal U}_\pm$.
Indeed, a different frame
\begin{align}
V'=({\cal U}_+ a, {\cal U}_- b)
\label{eq:chiral_basis_change}
\end{align}
gives $q'=b^\dag q a$ and hence
\begin{align}
 W_d[q']=W_d[q]+W_d[a]-W_d[b].
 \label{eq:summary-chiral-winding-frame}
\end{align}
The map $Q: E_+ \to E_-$ remains basis independent; the frame change shifts the reference used to assign an integer to its homotopy class.
The $\Z$ classification is therefore relative to a reference Hamiltonian $H_0$ and is characterized by the winding-number difference $W_d[q] - W_d[q_0]$~\cite{Thiang}.

Nevertheless, for a fixed Bloch basis associated with the original real-space orbitals, the following winding-number parity is independent of the eigenframes of $\G$:
\begin{align}
 {\rm WP}_d[H,\G]:=\bigl(W_d[q]+W_d[V]\bigr)\bmod2
 \in\Z_2.
 \label{eq:summary-chiral-winding-parity}
\end{align}
The transformation~\eqref{eq:chiral_basis_change} gives $W_d[V']=W_d[V]+W_d[a]+W_d[b]$. Together with Eq.~\eqref{eq:summary-chiral-winding-frame}, this identity shows that $W_d[q]+W_d[V]$ changes by $2W_d[a]$.
Thus, ${\rm WP}_d[H,\Gamma]$ is a bulk $\Z_2$ invariant independent of the choice of boundary.

The parity invariant obstructs a continuous deformation of both $H$ and $\G$ to momentum-independent operators that preserves the gap and their anticommutation relation. The allowed deformations include changes in $\G$. To establish the obstruction, choose ${\cal U}_-=Q{\cal U}_+$ and set
\begin{align}
 V=({\cal U}_+,Q{\cal U}_+).
 \label{eq:summary-chiral-joint-frame}
\end{align}
In this basis,
\begin{align}
 V^\dag\G V&=\begin{pmatrix}I_N&0\\0&-I_N\end{pmatrix},\notag\\
 V^\dag QV&=\begin{pmatrix}0&I_N\\I_N&0\end{pmatrix},
 \label{eq:summary-chiral-joint-trivialization}
\end{align}
so both operators become constant simultaneously. We then have $q=I_N$ and
\begin{align}
 {\rm WP}_d[H,\G]\equiv W_d[V]\pmod2.
 \label{eq:summary-chiral-joint-parity}
\end{align}
Any other basis yielding the same constant pair has the form $V'=V\operatorname{diag}(a,a)$, so $W_d[V']-W_d[V]=2W_d[a]$. Thus, ${\rm WP}_d$ is the parity of the winding number of a basis transformation that makes both operators constant.
If ${\rm WP}_d=1$, the two operators cannot be continuously deformed to become momentum independent simultaneously.

The parity is invariant under continuous deformations and changes of eigenframes {\it within a fixed original orbital basis}. It need not be invariant under a momentum-dependent change of the full basis. In fact, a basis transformation ${\cal W}: T^d \to U(2N)$ gives
\begin{align}
 \begin{split}
    &{\rm WP}_d[{\cal W}^\dag H{\cal W},{\cal W}^\dag \G{\cal W}]\\
    &\quad\equiv {\rm WP}_d[H,\G]+W_d[{\cal W}]\pmod2,
 \end{split}
\end{align}
where $W_d[{\cal W}]$ can in general be any integer.

Finally, express TRS in the basis of Eq.~\eqref{eq:summary-chiral-block}:
\begin{align}
 T^{(V)}_\bk=V_{-\tau\bk}^\dag T_\bk V_\bk^*.
\end{align}
The transformed matrices obey the same algebraic relation~\eqref{eq:TRS_G}. Thus, $T^{(V)}_\bk$ has the block forms in Eqs.~\eqref{eq:summary-cellwise-T-diagonal} and \eqref{eq:summary-cellwise-T-offdiagonal}, with $u_\bk,v_\bk$ now defined in the chosen frame $V_\bk$. The matrices $D_\bk$ and $q_\bk$ obey the same conditions~\eqref{eq:summary-chiral-symmetries} as in the cellwise case.

In the class-by-class calculations below, we use the transformed representation~\eqref{eq:summary-chiral-symmetries} to relate the winding number of $q_\bk$ to the topological invariants of the symmetry matrices $u_\bk,v_\bk$.
We then obtain a relation between ${\rm WP}_d[H,\Gamma]$ in the original basis and the topological invariant of the time-reversal operator $T_\bk$.
Neither flattening nor a momentum-dependent basis transformation $V$ generally preserves cellwise or finite-range locality. We therefore evaluate locality in the original unit-cell and real-space orbital basis.

\subsection{Parity constraints in the original orbital basis}
\label{sec:chiral-parity-results}

For the constraints in this subsection, we assume that a finite-dimensional gapped Hamiltonian compatible with the specified symmetries exists. Table~\ref{tab:result-wp-full-T} gives the resulting necessary conditions.

\begin{table*}[tp]
\centering
\caption{Parity constraints for chiral classes with momentum-dependent $\G$, expressed in terms of ${\rm WP}_d[H,\G]$ and the full time-reversal matrix $T$ in the original orbital basis.
Momentum subscripts are suppressed as in Table~\ref{tab:result}.
The symmetry column specifies the remaining algebraic relations.
The locality column concerns realizations of $\mathrm{WP}_d=1$ by the original symmetry operators $T$ and $\G$, using the fixed cell and orbital basis. 
Section numbers link to the underlying formulas and examples corresponding to each cellwise case.}   
\label{tab:result-wp-full-T}
\small
\setlength{\tabcolsep}{4pt}
\renewcommand{\arraystretch}{1.35}
\begin{tabular*}{\textwidth}{@{}ccc@{\hspace{12pt}\extracolsep{\fill}}llcc@{}}
\hline
$d$ & AZ & $d_\tau$ & Constraint ($\bmod2$) & Symmetry & Locality for ${\rm WP}=1$ & Sec. \\
\hline
1 & CII & 0 & ${\rm WP}_1[H,\G]\equiv\nu[T;\gamma_{0\pi}]$
& $\begin{aligned}[t]T^\top&=-T_{-\tau},\\T\G^*&=\G_{-\tau}T\end{aligned}$
& finite range & \ref{sec:1d-CII} \\[2pt]
1 & DIII & 1 & ${\rm WP}_1[H,\G]\equiv\tfrac12 W_1[T]$
& $\begin{aligned}[t]T^\top&=-T,\\T\G^*&=-\G T\end{aligned}$
& finite range & \ref{sec:1d-DIII} \\[2pt]
\hline
3 & CI & 0 & ${\rm WP}_3[H,\G]\equiv-\tfrac12 W_3[T]$
& $\begin{aligned}[t]T^\top&=T_{-\tau},\\T\G^*&=-\G_{-\tau}T\end{aligned}$
& quasilocal only & \ref{sec:3d-CI} \\[2pt]
3 & CII & 1 & ${\rm WP}_3[H,\G]\equiv\mu[T]$
& $\begin{aligned}[t]T^\top&=-T_{-\tau},\\T\G^*&=\G_{-\tau}T\end{aligned}$
& quasilocal only & \ref{sec:3d-CII} \\[2pt]
3 & DIII & 2 & ${\rm WP}_3[H,\G]\equiv-\tfrac12 W_3[T]$
& $\begin{aligned}[t]T^\top&=-T_{-\tau},\\T\G^*&=-\G_{-\tau}T\end{aligned}$
& quasilocal only & \ref{sec:3d-DIII} \\[2pt]
3 & BDI & 3 & ${\rm WP}_3[H,\G]\equiv\zeta[T;\mathfrak s]$
& $\begin{aligned}[t]T^\top&=T,\\T\G^*&=\G T\end{aligned}$
& finite range & \ref{sec:3d-BDI} \\[2pt]
\hline
\end{tabular*}
\end{table*}

Table~\ref{tab:result-wp-full-T} expresses the six chiral cases in the Bloch basis associated with the original real-space orbitals. Diagonalizing $\G_\bk$ with the frame $V_\bk$ of Sec.~\ref{sec:summary-chiral-general-frame} and setting
\begin{align}
 T^{(V)}_\bk:=V_{-\tau\bk}^\dag T_\bk V_\bk^*
 \label{eq:summary-variable-chiral-T}
\end{align}
allows us to use the expressions involving $u,v$ in Table~\ref{tab:result}. The direct-sum and basis-transformation laws in Sec.~\ref{sec:summary-symmetry-invariants} then express the results in terms of the original $T$.

In CII and BDI, $T^{(V)}=\operatorname{diag}(v,u)$. For example, in one-dimensional CII, the direct-sum and basis-transformation laws of $\nu$, together with the constraint in Table~\ref{tab:result}, give
\begin{align}
 \nu[T;\gamma_{0\pi}]+W_1[V]
 &\equiv\nu[T^{(V)};\gamma_{0\pi}]\notag\\
 &\equiv\nu[v;\gamma_{0\pi}]+\nu[u;\gamma_{0\pi}]\notag\\
 &\equiv W_1[q]\pmod2.
\end{align}
The definition in Eq.~\eqref{eq:summary-chiral-winding-parity} therefore yields ${\rm WP}_1[H,\G]\equiv\nu[T;\gamma_{0\pi}]$. The direct-sum and basis-transformation laws of $\mu$ and $\zeta$ give the three-dimensional CII and BDI results in the same way. Since the BDI equality holds for every $\mathfrak s$, the existence of a gapped $H$ compatible with both $T,\G$ implies that $\zeta[T;\mathfrak s]$, like ${\rm WP}_3[H,\G]$, is independent of the spin structure.

In DIII and CI, $T^{(V)}$ has an off-diagonal block form. In one-dimensional DIII and three-dimensional CI and DIII, respectively,
\begin{align}
 W_1[T^{(V)}]&=W_1[T]-2W_1[V]=2W_1[u],\notag\\
 W_3[T^{(V)}]&=W_3[T]+2W_3[V]=-2W_3[u].
\end{align}
Combining these identities with $W_d[q]\equiv W_d[u]\pmod2$ in Table~\ref{tab:result} gives the constraints in Table~\ref{tab:result-wp-full-T}. The same identities ensure that the half-winding numbers $\tfrac12W_d[T]$ appearing in the table are integers.

The locality entries in Table~\ref{tab:result-wp-full-T} refer to the original matrices $T_\bk,\G_\bk$. No finite-range assumption is made for the diagonalizing frame $V_\bk$. In three-dimensional CI and DIII, applying $W_3[T]=0$ for finite-range unitary matrices~\cite{Read-2017} directly excludes ${\rm WP}_3=1$; in three-dimensional CII, the same conclusion follows from $\mu[T]=0$ in Theorem~\ref{thm:CII-finite-range}. The finite-range odd-invariant models in one-dimensional CII and DIII and three-dimensional BDI remain available as a subfamily with cellwise $\G$. Thus, allowing momentum-dependent $\G$ does not change the locality distinctions in the table.

\section{\texorpdfstring{\NoCaseChange{One dimension}}{One dimension}}
\label{sec:1d}
In one dimension, the conventional $2\Z$ classification occurs for class CII with internal symmetry ($d_\tau=0$) and class DIII with an antiunitary reflection ($d_\tau=1$).
Following Sec.~\ref{sec:summary-chiral-general-frame}, choose a trivialization of the chiral operator and let $q_\bk$ denote the unitary matrix obtained by flattening the off-diagonal block of the Hamiltonian.

\subsection{\texorpdfstring{$(d,d_\tau)=(1,0)$}{(d,d\_tau)=(1,0)}, class CII}
\label{sec:1d-CII}
The symmetry conditions are
\begin{align}
    u_k q_k^* v_k^\dag = q_{-k},\qquad
    u_k^\top = -u_{-k},\qquad
    v_k^\top = -v_{-k}.
\end{align}
The $\Z_2$ invariants $\nu[u;\gamma_{0\pi}],\nu[v;\gamma_{0\pi}]$ of the symmetry matrices $u_k,v_k$, defined in Eq.~\eqref{eq:summary-pfaffian}, determine the winding-number parity through
\begin{align}
    W_1[q] \equiv \nu[u;\gamma_{0\pi}]-\nu[v;\gamma_{0\pi}] \pmod 2.
    \label{eq:CII-winding-parity}
\end{align}

\begin{proof}
Applying the symmetry condition gives
\begin{align}
    \int_{-\pi}^0 d\log\det q_k
    &=-\int_0^\pi d\log\det q_{-k} \nonumber\\
    &=\int_0^\pi d\log\det q_k
      -\int_0^\pi d\log\det(u_kv_k^\dag).
\end{align}
Hence,
\begin{align}
    (-1)^{W_1[q]}
    =\frac{\det q_\pi}{\det q_0}
      \exp\left[-\frac{1}{2}\int_0^\pi
      d\log\det(u_kv_k^\dag)\right].
\end{align}
At the fixed points $k=0,\pi$, we have $u_k=q_kv_kq_k^\top$, so
\begin{align}
    \pf[u_k]=\det q_k\,\pf[v_k].
\end{align}
Substituting the Pfaffian identity and the definition of $\nu[\,\cdot\,;\gamma_{0\pi}]$ yields
\begin{align}
    (-1)^{W_1[q]-\nu[u;\gamma_{0\pi}]+\nu[v;\gamma_{0\pi}]}
    =\frac{\det q_\pi}{\det q_0}
      \frac{\pf[u_0]}{\pf[u_\pi]}
      \frac{\pf[v_\pi]}{\pf[v_0]}
    =1.
\end{align}
This proves Eq.~\eqref{eq:CII-winding-parity}.
\end{proof}

\subsubsection{Example and SPT--LSM theorem}
\label{sec:1d-CII-SPT-LSM}
An example with an odd winding number is
\begin{align}
    u_k=i \sigma_y e^{ik\sigma_z},\qquad
    v_k=i \sigma_y,\qquad
    q_k=e^{-ik(1-\sigma_z)/2}.
    \label{eq:1DCII_model}
\end{align}
For this choice, $\nu[u;\gamma_{0\pi}]=1$, $\nu[v;\gamma_{0\pi}]=0$, and $W_1[q]=-1$. Every entry of $u_k,v_k,q_k$ is a finite Laurent polynomial in $e^{\pm ik}$, so both the Hamiltonian and the symmetry action admit finite-range realizations.

Equation~\eqref{eq:CII-winding-parity} can be interpreted as an SPT--LSM-type constraint: the difference between the symmetry actions in the two chiral subspaces requires an odd winding number.
In Fig.~\ref{fig:1d-SPT-LSM}[a], $A$ and $B$ are complex-fermion orbitals with $\G=+1,-1$, respectively. The action $v_k=i\sigma_y$ pairs the orbitals $A_{1,j},A_{2,j}$ within the same unit cell into a Kramers pair. In contrast,
\begin{align}
 u_k=\begin{pmatrix}0&e^{-ik}\\-e^{ik}&0\end{pmatrix}
 \label{eq:CII-realspace-pairing}
\end{align}
pairs the orbitals $B_{1,j},B_{2,j+1}$ in adjacent cells into a Kramers pair. With the unit cells fixed as in the figure, time reversal in the $B$ subspace is therefore not cellwise.
Time reversal exchanges the intracell bond $A_{1,j}$--$B_{1,j}$ with the intercell bond $A_{2,j}$--$B_{2,j+1}$ shown in the figure, so intracell couplings alone cannot open a gap. More explicitly, if $q$ is momentum independent, the symmetry conditions at $k=0,\pi$ force $q=0$ because $u_\pi=-u_0$.

The model $q_k=\operatorname{diag}(1,e^{-ik})$ combines an intracell dimer with the dimer limit of a nontrivial Su--Schrieffer--Heeger (SSH) model~\cite{SuSchriefferHeeger1979} in a direct sum. Keeping only the cells $j=1,\ldots,L$ shown in the figure leaves zero modes $B_{2,1}$ at the left end and $A_{2,L}$ at the right end. The cut preserves the cellwise chiral condition but breaks time reversal in the open system because the symmetry action crosses the cut. These SSH end states refer to the specified termination; the bulk parity constraint is a separate statement.

\begin{figure*}[tb]
\centering
\includegraphics[width=0.9\textwidth]{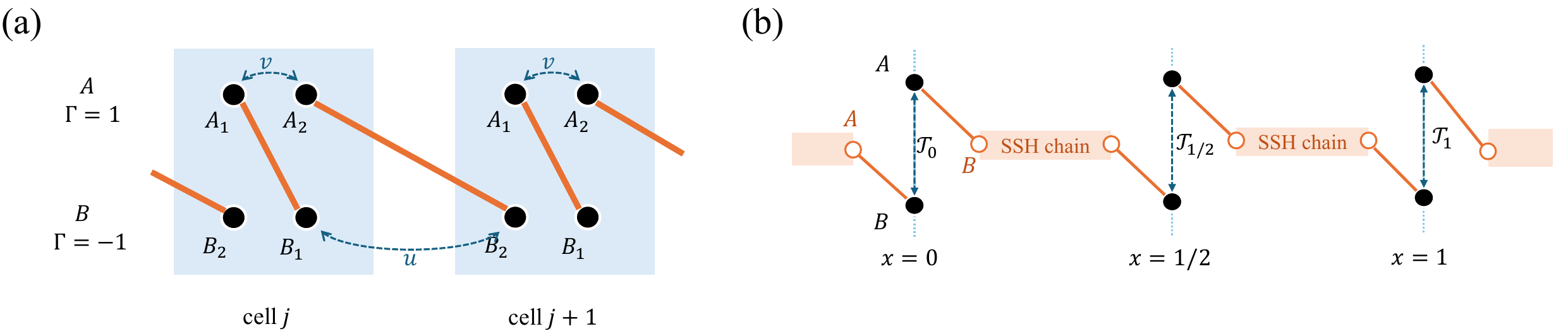}
\caption{Real-space interpretation of the one-dimensional SPT--LSM constraints in terms of complex fermions.
$A$ and $B$ denote the $\Gamma=+1$ and $\Gamma=-1$ subspaces, respectively.
Dashed arrows indicate Kramers pairing, and orange bonds indicate hopping.
[a] Class CII. Blue shading marks unit cells. Time reversal pairs $A_{1,j}$ with $A_{2,j}$ through $v$, and $B_{1,j}$ with $B_{2,j+1}$ through $u$.
The intracell and intercell dimers realize $q_k=\operatorname{diag}(1,e^{-ik})$, the direct sum of a trivial dimerized chain and a topological SSH chain, with $W_1[q]=-1$.
[b] Class DIII. At each mirror center $x=s$, the $A,B$ orbitals shown as filled circles form a Kramers pair under the antiunitary mirror operation $\mathcal T_s$.
The antiunitary mirror symmetry and the chiral condition pin an isolated pair to zero energy.
SSH chains on the intervening intervals supply end states (open circles) that form a second Kramers pair at each center.
The orange bonds couple states in opposite chiral subspaces, gapping all four states while preserving the symmetries, as in Eq.~\eqref{eq:DIII-mirror-SSH-mass}.
Orange shading marks the gapped interiors of the chains. }

\label{fig:1d-SPT-LSM}
\end{figure*}

\subsection{\texorpdfstring{$(d,d_\tau)=(1,1)$}{(d,d\_tau)=(1,1)}, class DIII}
\label{sec:1d-DIII}
The symmetry condition is
\begin{align}
    (q_ku_k)^\top=-q_ku_k.
\end{align}
The Pfaffian $\pf[q_ku_k]$ is therefore a nonzero complex function on the circle, with integer winding number. Using $\pf[q_ku_k]^2=\det q_k\det u_k$, we obtain
\begin{align}
    W_1[q]+W_1[u]
    =2\frac{1}{2\pi i}\oint d\log\pf[q_ku_k]
    \in 2\Z.
\end{align}
Hence,
\begin{align}
    W_1[q]\equiv W_1[u]\pmod 2.
    \label{eq:DIII-winding-parity}
\end{align}

An example with an odd winding number is
\begin{align}
    q_k=\begin{pmatrix}e^{ik}&0\\0&1\end{pmatrix},\qquad
    u_k=q_k^\dag i\sigma_y = \begin{pmatrix}0&e^{-ik}\\-1&0\end{pmatrix}.
    \label{eq:1DDIII_model}
\end{align}
The matrices satisfy $q_ku_k=i\sigma_y$, with $W_1[q]=1$ and $W_1[u]=-1$.
The symmetry action in this example also admits a finite-range realization.

\subsubsection{Real-space AHSS and an SPT--LSM-type constraint}
\label{sec:1d-DIII-SPT-LSM}
The odd-parity constraint in Eq.~\eqref{eq:DIII-winding-parity} can be understood through the real-space AHSS as the cancellation of zero modes at reflection centers by SSH end states~\cite{Shiozaki-Xiong-Gomi}. We use the number-conserving realization of Appendix~\ref{sec:summary-manybody}, taking $A,B$ in Fig.~\ref{fig:1d-SPT-LSM}[b] to be complex-fermion orbitals with $\G=+1,-1$.
With lattice translations included, the reflection centers in a unit cell lie at $x=0,1/2$. Placing the $A_2,B_1$ orbitals of Eq.~\eqref{eq:1DDIII_model} at $x=j$ and the $A_1,B_2$ orbitals at $x=j-1/2$ gives a Kramers pair of $A,B$ orbitals at each center. Reflection then maps each orbital to its geometric image. The momentum dependence of the symmetry matrix records the change in unit cell under reflection; it does not make the geometric action nonlocal.
Write the annihilation operators for the two orbitals at a center $s$ as $\hat a_{s,A},\hat a_{s,B}$. With an appropriate choice of phases, the antiunitary reflection that fixes the center acts as
\begin{align}
 \widehat{\mathcal T}_s\hat a_{s,A}\widehat{\mathcal T}_s^{-1}
 &=\hat a_{s,B},\qquad
 \widehat{\mathcal T}_s\hat a_{s,B}\widehat{\mathcal T}_s^{-1}
 =-\hat a_{s,A}.
 \label{eq:DIII-mirror-complex-pair}
\end{align}
For one Kramers pair, TRS restricts the single-particle Hamiltonian to $h_s=\epsilon_s I_2$, and the chiral condition then forces $\epsilon_s=0$. A single pair therefore cannot be gapped at zero energy while preserving both conditions.

Place a nontrivial SSH chain in $0<x<1/2$ and its reflection and translation images in the remaining intervals. Since reflection does not fix a generic point, a class-AIII chain with chiral symmetry suffices within each interval. At each center, the end states of the chains on the left and right form a second Kramers pair with opposite $\G$ eigenvalues. Denote their annihilation operators by $\hat b_{s,A},\hat b_{s,B}$ and let them transform as in Eq.~\eqref{eq:DIII-mirror-complex-pair}. The coupling
\begin{align}
 \widehat H_s
 &=m\bigl(\hat a_{s,A}^\dag\hat b_{s,B}
          -\hat a_{s,B}^\dag\hat b_{s,A}+{\rm h.c.}\bigr),
 \qquad m\in\R\setminus\{0\}
 \label{eq:DIII-mirror-SSH-mass}
\end{align}
preserves both the antiunitary reflection and the chiral condition. Its single-particle eigenvalues are $\pm|m|$, each doubly degenerate, so all four zero modes are gapped. The orange bonds in the figure represent these couplings.
For $N$ Kramers pairs, the off-diagonal block in the chiral basis is an $N\times N$ antisymmetric matrix. An odd number of pairs must therefore leave zero modes, while an even number can be gapped by taking direct sums of the coupling above. The obstruction at each center is thus the parity of the number of Kramers pairs.

Counting the end states yields the first differential of the AHSS. Let $n$ be the class-AIII winding number on a 1-cell; the SSH chain generates this classification. The chains and their reflection images together leave $n\bmod2$ Kramers pairs at each center, giving
\begin{align}
 d^1_{1,-1}:E^1_{1,-1}\cong\Z
 &\longrightarrow E^1_{0,-1}\cong\Z_2\oplus\Z_2,\notag\\
 n&\longmapsto(n\bmod2,n\bmod2).
 \label{eq:DIII-realspace-AHSS}
\end{align}
The two components on the right describe the obstructions at $x=0,1/2$. With no independent zero modes at either center, the gluing condition is $d^1_{1,-1}(n)=0$, so $n$ must be even. If one pair is supplied at each center, as in the figure, cancellation by the chain end states instead requires $d^1_{1,-1}(n)=(1,1)$, forcing $n$ to be odd.
Requiring nontrivial SSH chains to open a gap for the specified local degrees of freedom and symmetry is an SPT--LSM-type constraint.

\section{\texorpdfstring{\NoCaseChange{Two dimensions}}{Two dimensions}}
\label{sec:2d}
The relevant symmetry classes are C ($d_\tau=0$), AII ($d_\tau=1$), and D ($d_\tau=2$).
We write the momentum as $\bk=(k_x,k_y)\in T^2$.
Denote the occupied projector by $P$ and the occupied bundle by $E=\operatorname{im}P$. In a local orthonormal frame $\Phi$, we have $P=\Phi\Phi^\dag$, the Berry connection is $A=\Phi^\dag d\Phi$, and the curvature is $F(A)=dA+A^2$. In the calculations below, $\Tr$ denotes the trace over all bands, whereas $\tr$ denotes the trace of a matrix expressed in an occupied or unoccupied frame.
The projector and frame expressions for the curvature trace are related by $\tr F(A)=\Tr[P\,dP\,dP]$.

\subsection{\texorpdfstring{$(d,d_\tau)=(2,0)$}{(d,d\_tau)=(2,0)}, class C}
\label{sec:2d-C}
Consider the symmetry
\begin{align}
    C_\bk H_\bk^* C_\bk^\dag = -H_{-\bk},\qquad
    C_\bk^\top = -C_{-\bk}.
    \label{eq:classC-symmetry}
\end{align}
Here, we allow $C_\bk$ to depend on momentum.
For a Hamiltonian gapped at zero energy, denote the occupied and unoccupied projectors by $P_\bk^{(-)},P_\bk^{(+)}=1-P_\bk^{(-)}$, respectively. The symmetry implies
\begin{align}
    P_{-\bk}^{(-)}=C_\bk P_\bk^{(+)*}C_\bk^\dag.
\end{align}
Throughout this subsection, the superscripts $(-),(+)$ distinguish occupied and unoccupied projectors, frames, and connections.

Take the half Brillouin zone $D=[-\pi,\pi]\times[0,\pi]$ and identify $k_x=\pm\pi$.
Both $\oint_{k_y=0}$ and $\oint_{k_y=\pi}$ are taken in the positive $k_x$ direction, so $\int_{\partial D}=\oint_{k_y=0}-\oint_{k_y=\pi}$.
Denote the pullback of differential forms under $-\tau\bk=-\bk$ by $(-\tau)^*$, and define the $\C$-valued 1-form
\begin{align}
    B_\bk^{(\pm)}:=\Tr[P_\bk^{(\pm)*}C_\bk^\dag dC_\bk].
\end{align}
Using PHS together with $P^{(+)}=1-P^{(-)}$ gives
\begin{align}
 \begin{aligned}
  &(-\tau)^*\Tr[P^{(-)}dP^{(-)}dP^{(-)}]\\
  &\qquad=\Tr[P^{(-)}dP^{(-)}dP^{(-)}]+dB^{(+)}.
 \end{aligned}
\end{align}
The Chern number is consequently
\begin{align}
    {\rm ch}_1[H]
    &=2\frac{i}{2\pi}\int_D\Tr[P^{(-)}dP^{(-)}dP^{(-)}]\notag\\
    &\quad+\frac{i}{2\pi}\int_{\partial D}B^{(+)}\notag\\
    &\equiv\frac{i}{2\pi}
      \left(\oint_{k_y=0}-\oint_{k_y=\pi}\right)\notag\\
    &\qquad\bigl(2\tr A^{(-)}+B^{(+)}\bigr)\pmod 2.
    \label{eq:classC-halfBZ}
\end{align}
Here, $A^{(\pm)}=\Phi^{(\pm)\dag}d\Phi^{(\pm)}$ are the Berry connections of orthonormal Bloch frames.
The occupied bundle on the cylinder $D$ is trivial as a complex bundle, so a frame defined throughout $D$ makes the last equality exact. Independent choices of frame on the two boundary circles change the right-hand side by an even integer. We therefore evaluate the boundary expression $\bmod 2$ below.

Denote the four fixed points by $\G=(0,0)$, $X=(\pi,0)$, $Y=(0,\pi)$, and $M=(\pi,\pi)$. On the circle $k_y=0$, write $k=(k_x,0)$. Choose a smooth periodic occupied frame $\Phi_k^{(-)}$ and define the unoccupied frame by $\Phi_k^{(+)}=-C_{-k}\Phi_{-k}^{(-)*}$. The two frames obey
\begin{align}
 \begin{aligned}
  (-\tau)^*\tr A^{(-)}&=-\tr A^{(+)}+B^{(+)},\\
  (-\tau)^*B^{(+)}&=B^{(-)}.
 \end{aligned}
\end{align}
The pullback identities reduce the circle integral to the independent half path $\G\to X$:
\begin{align}
    &\oint_{k_y=0}\left(\tr A^{(-)}+\frac12B^{(+)}\right)\notag\\
    &\quad=\int_{\G\to X}\left(\tr A^{(-)}+\tr A^{(+)}
      -\frac12d\log\det C\right).
\end{align}
Define the frame for all bands by ${\cal U}_k=(\Phi_k^{(-)},\Phi_k^{(+)})$. Then
\begin{align}
    \exp\int_{\G\to X}(\tr A^{(-)}+\tr A^{(+)})
    =\frac{\det{\cal U}_X}{\det{\cal U}_{\G}}.
\end{align}
At the fixed points $k=0,\pi$, let $N$ denote the number of occupied bands. The symmetry matrix and its Pfaffian can be expressed in terms of the full frame as
\begin{align}
    &C_k={\cal U}_kJ_N{\cal U}_k^\top,\qquad
    J_N=\begin{pmatrix}0&1_N\\-1_N&0\end{pmatrix},\notag\\
    &\pf[C_k]=(\det{\cal U}_k)\pf[J_N].
\end{align}
The constant factor $\pf[J_N]$ cancels in the ratio of endpoint Pfaffians, yielding
\begin{align}
    &\exp\oint_{k_y=0}\left(\tr A^{(-)}+\frac12B^{(+)}\right)\notag\\
    &\quad=\frac{\pf[C_X]}{\pf[C_{\G}]}
      \exp\left[-\frac12\int_{\G\to X}d\log\det C\right]\notag\\
    &\quad=(-1)^{\nu[C;\gamma_{\G X}]}.
\end{align}
Repeating the calculation at $k_y=\pi$ gives $(-1)^{\nu[C;\gamma_{YM}]}$. Substituting the two boundary contributions into Eq.~\eqref{eq:classC-halfBZ} yields
\begin{align}
    {\rm ch}_1[H]\equiv\nu[C;\gamma_{\G X}]-\nu[C;\gamma_{YM}]\pmod 2.
    \label{eq:classC-parity}
\end{align}

\subsubsection{Interpretation as time reversal exchanging the sectors}
\label{sec:classC-sector-example}

To interpret Eq.~\eqref{eq:classC-parity}, consider the flattened Hamiltonian
\begin{align}
    Q_\bk:=H_\bk(H_\bk^2)^{-1/2}.
\end{align}
Since $Q_\bk^2=I$ and $[Q_\bk,H_\bk]=0$, $Q_\bk$ defines a momentum-dependent $\Z_2$ symmetry.
The antiunitary action $C_\bk K$ squares to $-1$ and exchanges the occupied sector $Q=-1$ and the unoccupied sector $Q=+1$.
The occupied and unoccupied sectors have Chern numbers ${\rm ch}_1[H]$ and $-{\rm ch}_1[H]$, respectively.
In the fixed orbital basis, $C_\bk$ is the sewing matrix for the antiunitary action.
The right-hand side of Eq.~\eqref{eq:classC-parity} is the sewing-matrix formula obtained from the difference in time-reversal polarization between $k_y=0,\pi$. Since the Berry curvature of the full system vanishes, this formula gives the Kane--Mele $\Z_2$ invariant of the full bundle with the antiunitary action specified above~\cite{FuKane2006}.
Equation~\eqref{eq:classC-parity} thus relates the Kane--Mele invariant to the parity of the sector Chern number when time reversal exchanges the $\Z_2$ symmetry sectors.

\subsubsection{Locality}
The right-hand side of Eq.~\eqref{eq:classC-parity} is also the two-dimensional class-DIII strong $\Z_2$ invariant of the auxiliary Hamiltonian~\cite{Qi-Hughes-Zhang-2010}
\begin{align}
    \widetilde H_\bk:=\begin{pmatrix}0&C_\bk\\C_\bk^\dag&0\end{pmatrix}.
\end{align}
The auxiliary Hamiltonian satisfies $\widetilde H_\bk^2=1$. Let $\rho_\mu$ denote the Pauli matrices acting on the auxiliary blocks. The operators $i\rho_yK$ and $\rho_xK$ represent TRS squaring to $-1$ and PHS squaring to $+1$, respectively.
If $C_\bk$ is a finite Laurent polynomial, $\widetilde H_\bk$ is an exactly flat finite-range Hamiltonian.
A no-go theorem states that the two-dimensional class-DIII strong invariant is trivial for any translation-symmetric finite-range flat-band model with finitely many orbitals~\cite{Read-2017}.
A nonzero right-hand side of Eq.~\eqref{eq:classC-parity} therefore rules out a finite-range realization of $C_\bk$.

\subsubsection{Example}
\label{sec:2dclassC-model}
A simultaneous construction of $H_\bk$ and $C_\bk$ uses standard spectral flattening~\cite{KobayashiInamuraShiozaki2026}.
Start with a gapped class-D Hamiltonian satisfying
\begin{align}
    \tau_xH_\bk^*\tau_x=-H_{-\bk}.
\end{align}
The matrix $\tau_x$ exchanges the particle and hole blocks.
The flattened Hamiltonian $Q_\bk=H_\bk(H_\bk^2)^{-1/2}$ obeys $[H_\bk,Q_\bk]=0$ and $\tau_xQ_\bk^*\tau_x=-Q_{-\bk}$.
Defining
\begin{align}
    C_\bk:=i\tau_xQ_\bk^*=-iQ_{-\bk}\tau_x
    \label{eq:classC-from-D}
\end{align}
gives a unitary $C_\bk$ satisfying
\begin{align}
    &C_\bk^\top=iQ_\bk\tau_x=-C_{-\bk},\notag\\
    &C_\bk H_\bk^*C_\bk^\dag
    =\tau_xQ_\bk^*H_\bk^*Q_\bk^*\tau_x=-H_{-\bk}.
\end{align}
The original Hamiltonian $H_\bk$ therefore satisfies the class-C conditions in Eq.~\eqref{eq:classC-symmetry}.
With this phase convention,
\begin{align*}
 \tau_xC_\bk^*\tau_x&=C_{-\bk},\\
 g_\bk=C_{-\bk}\tau_x&=-iQ_\bk,\qquad g_\bk^2=-I,\\
 \tau_xg_\bk^*\tau_x&=g_{-\bk},
\end{align*}
so the Nambu compatibility condition in Eq.~\eqref{eq:summary-bdg-trs-compatibility} of Appendix~\ref{sec:summary-bdg-symmetries} holds as well. The phase of $C_\bk$ follows the convention in Eq.~\eqref{eq:oddmodel-CD}; the resulting $g_\bk$ is the non-onsite $\Z_4^F$ action of Ref.~\cite{KobayashiInamuraShiozaki2026}.

If the starting Hamiltonian $H_\bk$ is finite range and gapped, its flattening $Q_\bk$ is real analytic. The resulting $C_\bk$ therefore has exponentially decaying real-space matrix elements. When ${\rm ch}_1[H]$ is odd, $C_\bk$ gives an exponentially decaying symmetry action for which finite range is impossible.

\subsection{\texorpdfstring{$(d,d_\tau)=(2,1)$}{(d,d\_tau)=(2,1)}, class AII}
\label{sec:2d-AII}
Consider an antiunitary symmetry accompanied by the reflection $x \mapsto -x$,
\begin{align}
    T_\bk H_\bk^*T_\bk^\dag=H_{k_x,-k_y},\qquad
    T_\bk^\top=-T_{k_x,-k_y}.
    \label{eq:AII-mirror-symmetry}
\end{align}
This symmetry forces the Chern number to be even, irrespective of the momentum dependence of $T_\bk$.

To prove the even-parity constraint, set $-\tau(k_x,k_y)=(k_x,-k_y)$ and denote the occupied projector by $P_\bk$. The symmetry implies
\begin{align}
    P_{-\tau\bk}=T_\bk P_\bk^*T_\bk^\dag,
\end{align}
and the trace of the curvature transforms as
\begin{align}
    (-\tau)^*\Tr[P_\bk dP_\bk dP_\bk]
    =-\Tr[P_\bk dP_\bk dP_\bk]
      +d\Tr[P_\bk^*T_\bk^\dag dT_\bk].
\end{align}
Here, $(-\tau)^*$ denotes the pullback of differential forms. Choose the half Brillouin zone $D=[-\pi,\pi]\times[0,\pi]$ with $k_x=\pm\pi$ identified. The occupied complex bundle on $D$ is trivial and therefore admits a global orthonormal frame $\Phi_\bk$. Set $A_\bk=\Phi_\bk^\dag d\Phi_\bk$. Combining the orientation reversal under $-\tau$ with Stokes' theorem gives
\begin{align}
    {\rm ch}_1[H]
    &=2\frac{i}{2\pi}\int_D\Tr[P_\bk dP_\bk dP_\bk]\notag\\
    &\quad-\frac{i}{2\pi}\int_D d\Tr[P_\bk^*T_\bk^\dag dT_\bk]
      \nonumber\\
    &=\frac{i}{2\pi}
      \left(\oint_{k_y=0}-\oint_{k_y=\pi}\right)\notag\\
    &\qquad\left(2\tr A_\bk-\Tr[P_\bk^*T_\bk^\dag dT_\bk]\right).
    \label{eq:AII-mirror-boundary}
\end{align}
Both circle integrals are taken in the direction of increasing $k_x$.

On each fixed line $k_y=0,\pi$, the frames $\Phi_\bk$ and $T_\bk\Phi_\bk^*$ span the same occupied space. They are therefore related by a periodic unitary matrix $w_\bk$:
\begin{align}
    T_\bk\Phi_\bk^*=\Phi_\bk w_\bk,\qquad
    w_\bk=\Phi_\bk^\dag T_\bk\Phi_\bk^*,\qquad
    w_\bk^\top=-w_\bk.
\end{align}
Along the fixed lines, the transformation law for the Berry connection gives
\begin{align}
    -\tr A_\bk+\Tr[P_\bk^*T_\bk^\dag dT_\bk]
    &=\tr[(T_\bk\Phi_\bk^*)^\dag d(T_\bk\Phi_\bk^*)]
      \nonumber\\
    &=\tr A_\bk+d\log\det w_\bk.
\end{align}
Hence, each boundary term in Eq.~\eqref{eq:AII-mirror-boundary} is
\begin{align}
    &\frac{i}{2\pi}\oint
    \left(2\tr A_\bk-\Tr[P_\bk^*T_\bk^\dag dT_\bk]\right)\notag\\
    &\quad=\frac{1}{2\pi i}\oint d\log\det w_\bk \nonumber\\
    &\quad=2\frac{1}{2\pi i}\oint d\log\pf[w_\bk]
    \in 2\Z.
\end{align}
Here, we used $\det w_\bk=\pf[w_\bk]^2$. We therefore obtain
\begin{align}
    {\rm ch}_1[H]\equiv0\pmod 2.
    \label{eq:AII-mirror-even}
\end{align}
A gauge transformation between the two frames does not force the boundary integral to vanish. The even-parity constraint instead follows from the even winding number of the determinant of an antisymmetric sewing matrix.

\subsubsection{Relation to the SPT--LSM theorem}
\label{sec:AII-mirror-SPT-LSM}
The even-parity constraint in Eq.~\eqref{eq:AII-mirror-even} is consistent with the real-space AHSS description in terms of topological phases on cells and anomalies on their boundaries~\cite{Shiozaki-Xiong-Gomi}.
For the real-space argument, consider number-conserving free fermions without imposing lattice translations.
The antiunitary reflection $T:(x,y)\mapsto(-x,y)$ acts on the reflection line $x=0$ as time reversal squaring to $-1$.
A single helical pair on the reflection line carries the same anomaly as the edge of a two-dimensional class-AII $\Z_2$ topological insulator~\cite{KaneMele2005Z2}.

Place a Chern insulator with Chern number $n$ in the half-plane $x>0$ and its symmetry image in $x<0$.
Reflection and time reversal each reverse the sign of the Chern number, so both half-planes have the same Chern number $n$.
The chiral edge modes from opposite sides propagate in opposite directions along the reflection line, forming $|n|$ helical pairs.
For this cell decomposition, the first differential of the AHSS is therefore
\begin{align}
 d^1_{2,-2}:E^1_{2,-2}\cong\Z
 &\longrightarrow E^1_{1,-2}\cong\Z_2,\notag\\
 n&\longmapsto n\bmod2.
 \label{eq:AII-mirror-AHSS}
\end{align}
The right-hand side represents the parity of the number of helical pairs induced on the reflection line.
Without an independent anomaly on the reflection line, the condition for symmetry-preserving gluing of the two half-planes is $d^1_{2,-2}(n)=0$, so $n$ is even.

If a single helical pair is supplied on the reflection line, odd $n$ makes the total number of pairs even after the contributions from both half-planes are included.
Coupling counterpropagating modes from different pairs gaps the reflection line while preserving the symmetry (Fig.~\ref{fig:AII-mirror-anomaly-cancellation}[a]).
If these degrees of freedom admit a fully gapped symmetric system, its Chern number must therefore be odd. The requirement of a nontrivial bulk phase to cancel the boundary anomaly is an SPT--LSM-type constraint.

However, a single helical pair cannot be realized as an independent one-dimensional local lattice model with cellwise time-reversal symmetry.
Such a pair can be realized at the boundary of a higher-dimensional system. 
For example, place a two-dimensional class-AII topological insulator on the reflection plane $x=0$ within the three-dimensional half-space $z\leq 0$. 
Its boundary supplies a single helical pair on the reflection line of the surface $z=0$.
Supplying anomalous boundary degrees of freedom in this way takes the construction outside the setting of finite-dimensional Bloch Hamiltonians and symmetry matrices used in this paper.
The conditional odd-parity constraint from anomaly cancellation therefore does not contradict the even-parity constraint in Eq.~\eqref{eq:AII-mirror-even}.

The single-particle obstruction does not by itself exclude many-body realizations with non-onsite $U(1)$ symmetry.
It has been shown that the quantum spin-Hall edge admits a one-dimensional lattice realization with finite-dimensional local Hilbert spaces of fermions and Ising spins. 
In this model, both $U(1)$ and time reversal act non-onsitely~\cite{Metlitski2019QSHBoundary}.
A many-body edge of this type may supply the required reflection-line anomaly and allow a many-body analogue of the odd-Chern phase described above, beyond the single-particle framework.
Whether coupling this edge to the two Chern-insulator regions permits symmetric gapping remains to be established.

\begin{figure*}[t]
\centering
\includegraphics[width=0.85\textwidth]{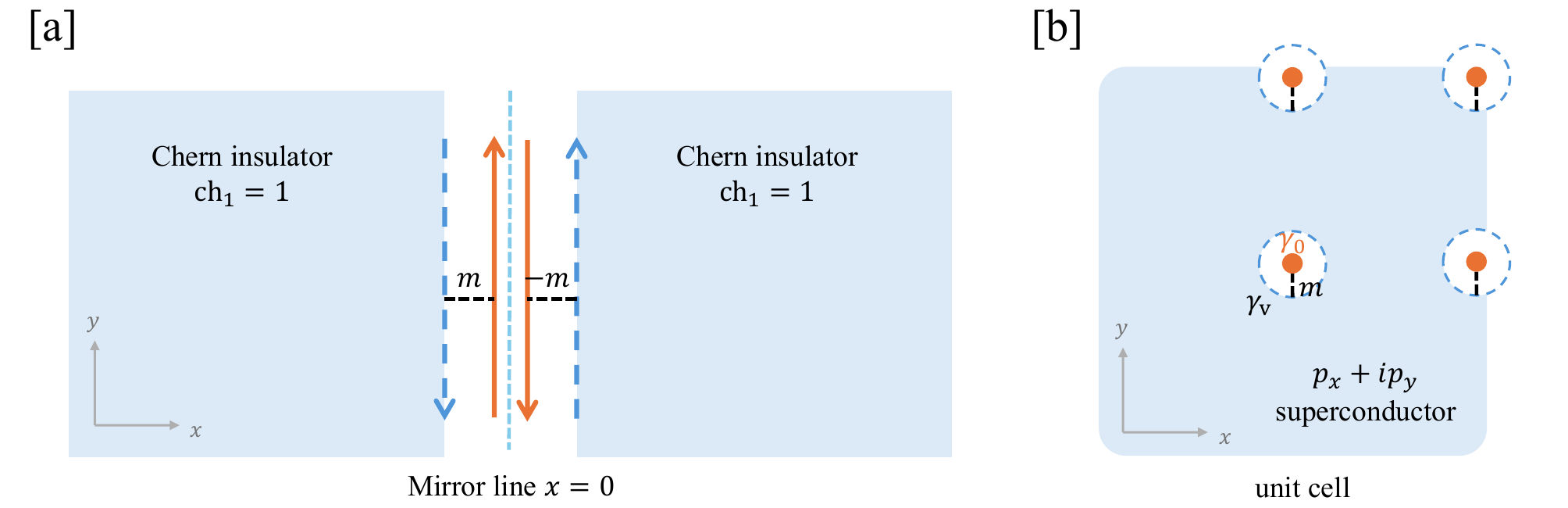}
\caption{Real-space SPT--LSM mechanisms in two dimensions.
[a] Class AII. Orange solid arrows denote the original helical pair on the mirror line $x=0$.
The mirror-related Chern insulators on the two sides each have ${\rm ch}_1=1$ and supply the counterpropagating edge modes shown by blue dashed arrows.
Black dashed segments mark couplings $m$ and $-m$ between counterpropagating modes. Taken together, the couplings preserve the symmetry and gap all four modes.
The transverse separation of the modes is schematic.
[b] Class D in the superconducting description.
A $p_x+ip_y$ superconductor with $\widehat C_2^2=1$ has vortices at the rotation centers.
Blue dashed circles mark these vortices, each binding a Majorana zero mode $\gamma_{\rm v}$.
Orange dots denote additional microscopic Majorana modes $\gamma_0$ at the same centers.
Black dashed bonds represent the symmetry-allowed couplings $im\gamma_0\gamma_{\rm v}$ that gap each pair of zero modes.
The shaded region indicates a unit cell.}
\label{fig:AII-mirror-anomaly-cancellation}
\end{figure*}

\subsection{\texorpdfstring{$(d,d_\tau)=(2,2)$}{(d,d\_tau)=(2,2)}, class D}
\label{sec:2d_classD}
Consider PHS that fixes momentum,
\begin{align}
C_\bk H_\bk^* C_\bk^\dag = -H_\bk,\quad C_\bk^\top=C_\bk.
\label{eq:D_pointwise_symmetry}
\end{align}
The symmetry combines a twofold rotation $\tau(x,y) = (-x,-y)$ with a particle-hole transformation. We allow its matrix $C_\bk$ to depend on momentum.
A gap at zero energy implies equal numbers of positive- and negative-energy states. We therefore write the matrix size of the Hamiltonian as $2n$.

Let $\{O_i\}_i$ be a good cover of $T^2$ and write $O_{ij}=O_i\cap O_j$, and so on.
On each patch, choose a local Takagi factorization
\begin{align}
C_\bk=U_{i,\bk}U_{i,\bk}^\top,\quad U_{i,\bk}\in U(2n),\quad \bk\in O_i.
\end{align}
The transition functions $S_{ij}$ on overlaps are defined by
\begin{align}
U_{j,\bk}=U_{i,\bk}S_{ij,\bk},\quad S_{ij,\bk}\in O(2n),\quad \bk\in O_{ij}
\end{align}
and obey the cocycle condition $S_{ij}S_{jl}=S_{il}$. The Takagi factors $U_i$ are defined patchwise; no global factor is assumed. Denote the resulting principal $O(2n)$ bundle by $\mathcal P_C$. We use the SW classes of $\mathcal P_C$ to mean the classes of the real vector bundle associated with its standard representation.
In the local Takagi basis, define $\widetilde H_{i,\bk}:=U_{i,\bk}^\dag H_\bk U_{i,\bk}$. The PHS condition in Eq.~\eqref{eq:D_pointwise_symmetry} gives
\begin{align}
\widetilde H_{i,\bk}^*=-\widetilde H_{i,\bk}.
\end{align}

Let $\Phi_{i,\bk}$ be a local orthonormal Bloch frame for the negative-energy states, and define transition functions by
\begin{align}
\Phi_{j,\bk}=\Phi_{i,\bk}V_{ij,\bk},\quad V_{ij,\bk}\in U(n),\quad \bk\in O_{ij}.
\end{align}
The transition functions define the occupied bundle $E$. Denote its principal $U(n)$ bundle of orthonormal frames by $\mathcal P_H$, so $c_1(\mathcal P_H)=c_1(E)$. We reserve $\mathcal P_C,\mathcal P_H$ for principal bundles and $P$ for the projector.
Expressing the occupied frame in the Takagi basis as $\widetilde\Phi_{i,\bk}:=U_{i,\bk}^\dag\Phi_{i,\bk}$ gives
\begin{align}
S_{ij,\bk}\widetilde\Phi_{j,\bk}=\widetilde\Phi_{i,\bk}V_{ij,\bk},\quad \bk\in O_{ij}.
\label{eq:D_tr_func}
\end{align}
Because the Hamiltonian in the Takagi basis is purely imaginary, $\widetilde\Phi_i^*$ is an orthonormal frame for the positive-energy states and is orthogonal to $\widetilde\Phi_i$. Consequently,
\begin{align}
Q(\widetilde\Phi_i):=\sqrt{2}\,(\re\widetilde\Phi_i,\im\widetilde\Phi_i)\in O(2n)
\end{align}
is a real orthogonal matrix. Taking the real and imaginary parts of Eq.~\eqref{eq:D_tr_func} gives
\begin{align}
&S_{ij}Q(\widetilde\Phi_j)=Q(\widetilde\Phi_i)R(V_{ij}),\notag\\
&R(V):=\begin{pmatrix}
\re V&\im V\\
-\im V&\re V
\end{pmatrix}\in SO(2n).
\end{align}
Here, $Q(\widetilde\Phi_i)$ denotes the real orthogonal matrix built from the local frame, rather than a flattened Hamiltonian. The realification map $R:U(n)\to SO(2n)$ is an injective group homomorphism.
The transition-function relation identifies the real vector bundle associated with $\mathcal P_C$ with the realification $E_\R$ of the occupied bundle $E$.
Let the principal $SO(2n)$ bundle with transition functions $R(V_{ij})$ be denoted by $(\mathcal P_H)_\R$. This is the oriented orthonormal frame bundle of $E_\R$, and hence
\begin{align}
w_1(\mathcal P_C)=0,\quad w_2(\mathcal P_C)=w_2((\mathcal P_H)_\R).
\end{align}
Using the standard relation $c_1(\mathcal P_H)\equiv w_2((\mathcal P_H)_\R)\pmod 2$ for the realification of a complex vector bundle~\cite{Milnor-Stasheff}, we obtain
\begin{align}
c_1(\mathcal P_H)\equiv w_2(\mathcal P_C)\pmod 2.
\end{align}
Evaluating the congruence on the fundamental class of $T^2$ gives
\begin{align}
{\rm ch}_1[H]\equiv\langle w_2(\mathcal P_C),[T^2]\rangle\pmod 2.
\label{eq:D_chern_SW}
\end{align}

\subsubsection{Interpretation as time reversal exchanging the sectors}
\label{sec:classD-sector-example}

Equation~\eqref{eq:D_chern_SW} can also be interpreted by regarding the flattened Hamiltonian $Q_\bk:=H_\bk(H_\bk^2)^{-1/2}$ as a momentum-dependent $\Z_2$ symmetry.
The antiunitary action $C_\bk K$ squares to $+1$ and exchanges the occupied sector $Q=-1$ and the unoccupied sector $Q=+1$.
The antiunitary action fixes momentum, as does the product $C_2T$ of a twofold rotation and time reversal.
The occupied and unoccupied sectors have Chern numbers ${\rm ch}_1[H]$ and $-{\rm ch}_1[H]$, respectively.
Equation~\eqref{eq:D_chern_SW} thus relates the Chern numbers of the $\Z_2$ symmetry sectors to the second SW class defined by $C_2T$ symmetry~\cite{AhnParkYang2019}.

\subsubsection{Example}
When analyzing the invariant of the symmetry matrix, we must also impose the necessary condition $w_1(\mathcal P_C)=0$ for a gapped model to exist. For example, for
\begin{align}
C_\bk=\operatorname{diag}(e^{ik_x},e^{ik_y}),
\end{align}
the classes are $w_1(\mathcal P_C)=x+y$ and $\langle w_2(\mathcal P_C),[T^2]\rangle=1$, where $x,y\in H^1(T^2;\Z_2)$ generate cohomology in the $k_x,k_y$ directions, respectively. Since $w_1(\mathcal P_C)\ne0$, no fully gapped model can exist.
The obstruction to a gap can also be checked directly. Write a Hermitian Hamiltonian as
\begin{align}
H_\bk=\begin{pmatrix}a_\bk&b_\bk\\b_\bk^*&-a_\bk\end{pmatrix},\quad a_\bk\in\R,\quad b_\bk\in\C.
\end{align}
The symmetry condition gives
\begin{align}
a_\bk=0,\quad b_\bk^2=-|b_\bk|^2e^{i(k_x-k_y)}.
\end{align}
If $b_\bk$ were nonzero everywhere, the phase winding along $k_x$ would be even on the left-hand side and odd on the right-hand side. The two sides could not agree.

Now place one complex-fermion degree of freedom at each of the four points $(0,0)$, $(1/2,0)$, $(0,1/2)$, and $(1/2,1/2)$ in a unit cell. The corresponding symmetry matrix is
\begin{align}
C_\bk=\operatorname{diag}(1,e^{ik_x},e^{ik_y},e^{i(k_x+k_y)}).
\label{eq:D_four_band_C}
\end{align}
The phase factors record the cell displacements of the orbitals under the twofold rotation. Although the action is not cellwise in the fixed Bloch basis, it can be realized geometrically by mapping each orbital to its rotation image. The four real line bundles have first SW classes $0,x,y,x+y$, respectively.
The Whitney sum formula in Eq.~\eqref{eq:summary-SW-whitney} gives
\begin{align*}
 w_1(\mathcal P_C)&=x+y+(x+y)=0,\\
 w_2(\mathcal P_C)&=x\cup y,
\end{align*}
so $\langle w_2(\mathcal P_C),[T^2]\rangle=1$.
We discuss the real-space SPT--LSM interpretation of this four-orbital configuration in Sec.~\ref{sec:2d-D-SPT-LSM}.

For $t>0$ and $t_2\ne0$, the symmetry in Eq.~\eqref{eq:D_four_band_C} admits the fully gapped model~\cite{AffleckRahmaniPikulin2017}\footnote{Given a Takagi factorization $C_\bk=U_\bk U_\bk^\top$, it suffices to choose $\widetilde H_\bk=U_\bk^\dag H_\bk U_\bk$ to be purely imaginary. With $S_x=1_2\otimes\sigma_z$ and $S_y=\sigma_z\otimes1_2$, the transformed Hamiltonian must also obey $\widetilde H_{\bk+2\pi\hat x}=S_x\widetilde H_\bk S_x$ and $\widetilde H_{\bk+2\pi\hat y}=S_y\widetilde H_\bk S_y$. In Eq.~\eqref{eq:D_four_band_H}, half-angle functions satisfying these boundary conditions multiply three mutually anticommuting purely imaginary Dirac matrices.}
\begin{align}
H_\bk&=U_\bk\Bigl[
t\sin\frac{k_x}{2}\,1_2\otimes\sigma_y
+t\sin\frac{k_y}{2}\,\sigma_y\otimes\sigma_z
\notag\\
&\hspace{5em}
+2t_2\cos\frac{k_x}{2}\cos\frac{k_y}{2}\,\sigma_y\otimes\sigma_x
\Bigr]U_\bk^\dag,
\label{eq:D_four_band_H}\\
U_\bk&=\operatorname{diag}(1,e^{ik_x/2},e^{ik_y/2},e^{i(k_x+k_y)/2}).
\end{align}
The Hamiltonian is $2\pi$ periodic in both directions and has finite-range hopping. Its minimum absolute energy is $\min(t,2|t_2|)>0$, and its Chern number is ${\rm ch}_1[H]=\operatorname{sgn}t_2$. At $t_2=0$, the gap closes and the Chern number changes by $2$, consistent with Eq.~\eqref{eq:D_chern_SW}.

\subsubsection{Relation to the SPT--LSM theorem}
\label{sec:2d-D-SPT-LSM}
The odd Chern number required by the complex-fermion configuration in Eq.~\eqref{eq:D_four_band_C} can be understood as an SPT--LSM constraint: zero modes bound to $\pi$ flux in a Chern insulator cancel the degrees of freedom at the rotation centers.
To explain the cancellation, we use the equivalent $p_x+ip_y$-wave superconducting phase~\cite{ChengWang2022Rotation,Shiozaki-Xiong-Gomi,ElseThorngren2020}.

If the twofold-rotation PHS is realized as a symmetry of a BdG Hamiltonian, the prescription in Appendix~\ref{sec:summary-bdg-symmetries} yields a unitary $C_2$ rotation symmetry $C_{2\bk} = C_{-\bk} \Sigma_x$. The relation $C_{2-\bk}C_{2\bk} = 1$ corresponds to $\widehat C_2^2=1$ for the physical rotation $\widehat C_2$ acting on fermion operators.
The superconductor is even under $C_2$ rotation and is incompatible with $p_x+ip_y$-wave pairing.

The same rotation symmetry requires vortices in the $p_x+ip_y$ phase. To see this, consider spinless fermions with the transformation law $\widehat C_2\hat c(\bx)\widehat C_2^{-1}=\hat c(-\bx)$ and the gap term $\Delta(\bx)\hat c^\dag(\bx)(\partial_x+i\partial_y)\hat c^\dag(\bx)+{\rm h.c.}$.
Rotation reverses the sign of $\partial_x+i\partial_y$, so symmetry imposes $\Delta(-\bx)=-\Delta(\bx)$ at every $C_2$ rotation center and requires a $\pi$ flux there.
A vortex Majorana mode $\gamma_{\rm v}$ therefore appears at each rotation center as a bound state formed from the bulk $p_x+ip_y$-wave superconducting degrees of freedom. The vortices form a Majorana vortex lattice.

Rotation symmetry pins each zero mode and prevents an invertible phase. Adding a Majorana degree of freedom $\gamma_0$ at a rotation center, however, allows the coupling $im\gamma_0\gamma_{\rm v}$, which preserves $\hat C_2$ and gaps the two zero modes (Fig.~\ref{fig:AII-mirror-anomaly-cancellation}[b]).

The real-space AHSS describes this zero-mode cancellation through the second differential~\cite{Shiozaki-Xiong-Gomi}
\begin{align}
 d^2_{2,-2}:E^2_{2,-2}\cong\Z
 &\longrightarrow E^2_{0,-1}\cong\Z_2^{\times 4},\notag\\
 n&\longmapsto (n,n,n,n).
 \label{eq:2d-D-AHSS-vortex}
\end{align}
The input is the Chern number of the two-dimensional phase; the output records the parities of the Majorana degrees of freedom at the rotation centers.
The SPT--LSM theorem takes the form of the preimage $(d^2_{2,-2})^{-1}(1,1,1,1) = 2\Z+1$, which forces the Chern number to be odd.

\section{\texorpdfstring{\NoCaseChange{Three dimensions}}{Three dimensions}}
\label{sec:3d}
The relevant symmetry classes are CI ($d_\tau=0$), CII ($d_\tau=1$), DIII ($d_\tau=2$), and BDI ($d_\tau=3$).
Following Sec.~\ref{sec:summary-chiral-general-frame}, we choose a trivialization of the chiral operator and write the flattened unitary off-diagonal block of the Hamiltonian as $q_\bk$.

\subsection{Wess--Zumino term and Polyakov--Wiegmann formula}
\label{sec:wz-pw}
With the winding-number normalization in Eq.~\eqref{eq:summary-winding}, the three-form to be integrated is
\begin{align}
    \omega_3(q):=\frac1{24\pi^2}\tr[(q^{-1}dq)^3],\qquad
    W_3[q]=\int_{T^3}\omega_3(q).
    \label{eq:coh-w3}
\end{align}
For a unitary matrix, $q^{-1}=q^\dag$. Under transposition, complex conjugation, and the momentum action $-\tau$, we have
\begin{align}
    \omega_3(q^\top)&=-\omega_3(q),\qquad
    \omega_3(q^*)=\omega_3(q),\notag\\
    \omega_3(q\circ(-\tau))&=(-\tau)^*\omega_3(q).
    \label{eq:coh-w3-transform}
\end{align}
Here $(-\tau)^*$ denotes the pullback of differential forms. Subscripts will indicate composition with a momentum action, as in $q_{-\tau}=q\circ(-\tau)$. The momentum action $-\tau$ in Eq.~\eqref{eq:summary-momentum-action} reverses $3-d_\tau$ coordinates, so integration gives
\begin{align}
    W_3[q\circ(-\tau)]=(-1)^{3-d_\tau}W_3[q].
\end{align}

For a unitary matrix $q$ on a closed oriented two-dimensional manifold $\Sigma$, we define the WZ term~\cite{WittenBosonization1984} using an extension to three dimensions:
\begin{align}
    {\rm WZ}[q;\Sigma]:=\int_M\omega_3(\widetilde q)\in\R/\Z,
    \qquad \partial M=\Sigma,\quad \widetilde q|_\Sigma=q.
    \label{eq:c2f-WZ}
\end{align}
Here $M$ is a compact oriented three-dimensional manifold, and $\widetilde q$ is an extension of $q$.\footnote{The extension exists because $\Omega^{SO}_2(U(N))=0$. Even when $\Sigma=T^2$, the one-dimensional winding numbers in Eq.~\eqref{eq:coh-m1} need not vanish.} Gluing two extensions gives a closed three-dimensional manifold with integer $W_3$. The WZ term is therefore independent of the extension $\bmod1$.
We evaluate ${\rm WZ}[q;\Sigma]$ on the restriction of $q$ to $\Sigma$ and omit $\Sigma$ when the integration surface is clear from context.

For a boundary $\partial M=\bigsqcup_a\Sigma_a$, orient each component $\Sigma_a$ by the induced orientation. Then
\begin{align}
    \int_M\omega_3(q)\equiv\sum_a{\rm WZ}[q;\Sigma_a]\pmod1.
    \label{eq:coh-wz-boundary}
\end{align}
Let $q_t=q|_{T^2\times\{t\}}$, with $t\in S^1$, and orient $T^3$ by $dk_x\wedge dk_y\wedge dt$. Then
\begin{align}
    W_3[q]=\oint_{S^1}d\,{\rm WZ}[q_t].
    \label{eq:coh-wz-winding}
\end{align}

For the Polyakov--Wiegmann (PW) term~\cite{PolyakovWiegmann1984}, we define the two-form and its integral separately:
\begin{align}
    {\rm pw}(q_1,q_2)
       &:=-\frac1{8\pi^2}\tr[(q_1^{-1}dq_1)(dq_2q_2^{-1})],\notag\\
    {\rm PW}[q_1,q_2;\Sigma]&:=\int_\Sigma{\rm pw}(q_1,q_2)\pmod1.
    \label{eq:coh-pw-definition}
\end{align}
When the integer part matters, we use the real-valued integral $\int_\Sigma{\rm pw}$ rather than ${\rm PW}$. The corresponding local identity is
\begin{align}
    \omega_3(q_1q_2)
       =\omega_3(q_1)+\omega_3(q_2)+d\,{\rm pw}(q_1,q_2).
    \label{eq:alt62-local-PW}
\end{align}
The PW two-form satisfies
\begin{align}
    &{\rm pw}(q_1,q_2)+{\rm pw}(q_1q_2,q_3)\notag\\
    &\quad={\rm pw}(q_2,q_3)+{\rm pw}(q_1,q_2q_3),\notag\\
    &{\rm pw}(q_1^*,q_2^*)={\rm pw}(q_1,q_2)
       =-{\rm pw}(q_2^\top,q_1^\top).
    \label{eq:coh-pw-identities}
\end{align}

We orient $\Sigma=T^2$ by $dk_x\wedge dk_y$ and define the determinant winding numbers by
\begin{align}
    m_\mu[q]:=\frac1{2\pi i}\oint_{S^1_\mu}d\log\det q\in\Z,
    \qquad \mu=x,y.
    \label{eq:coh-m1}
\end{align}
The circle $S^1_\mu$ is oriented in the direction of increasing $k_\mu$; the other coordinates are held fixed in the integral.
The PW formula for $U(N)$~\cite{GawedzkiWaldorf2009,MonacoTauber2017} includes a correction involving the winding numbers:
\begin{align}
    &{\rm WZ}[q_1q_2]\notag\\
    &\quad\equiv{\rm WZ}[q_1]+{\rm WZ}[q_2]+{\rm PW}[q_1,q_2]\notag\\
    &\qquad-\frac12\{m_x[q_1]m_y[q_2]-m_y[q_1]m_x[q_2]\}\pmod1.
    \label{eq:alt62-global-PW}
\end{align}

\subsection{\texorpdfstring{$(d,d_\tau)=(3,0)$}{(d,d\_tau)=(3,0)}, class CI}
\label{sec:3d-CI}
The symmetry condition is
\begin{align}
    (q_\bk u_\bk)^\top=q_{-\bk}u_{-\bk}.
    \label{eq:CI-symmetry}
\end{align}
Set $X_\bk:=q_\bk u_\bk$. We will prove $W_3[X]\in2\Z$, which implies
\begin{align}
    W_3[q]\equiv W_3[u]\pmod2.
    \label{eq:CI-winding-parity}
\end{align}

Orient $T^3$ by $dk_x\wedge dk_y\wedge dk_z$ and take the half Brillouin zone $V=T^2\times[0,\pi]$. Transposition and momentum inversion both reverse the sign of the three-dimensional winding number. By Eq.~\eqref{eq:CI-symmetry}, the two halves therefore contribute equally:
\begin{align}
    W_3[X]=2\int_V\omega_3(X).
\end{align}
With both boundary tori $T^2$ oriented by $dk_x\wedge dk_y$, the definition of the WZ term gives
\begin{align}
 \begin{split}
  \frac12 W_3[X]
  &\equiv {\rm WZ}[X;\{k_z=\pi\}]\\
  &\quad-{\rm WZ}[X;\{k_z=0\}]\pmod1.
 \end{split}
 \label{eq:CI-WZ-boundary}
\end{align}
On both boundaries, the matrix satisfies $X_{k_x,k_y}^\top=X_{-k_x,-k_y}$. We now show that the WZ term vanishes for any two-dimensional matrix family with this property.

\begin{lem}
\label{lem:CI-factorization}
Suppose that a periodic unitary matrix $X:T^2\to U(N)$ satisfies $X_\bk^\top=X_{-\bk}$. Then there is a periodic unitary matrix $g:T^2\to U(N)$ such that
\begin{align}
    X_\bk=g_{-\bk}g_\bk^\top,
    \qquad\text{equivalently}\qquad
    X_\bk g_\bk^*=g_{-\bk}.
    \label{eq:CI-factorization}
\end{align}
Moreover, $g$ can be chosen to be homotopic to the constant map $1_N$.
\end{lem}

\begin{proof}
At each inversion-fixed point $P=-P$, choose a Takagi factorization $X_P=g_Pg_P^\top$. On $k_y=0$, choose a path $g_{k_x,0}$, $k_x\in[0,\pi]$, connecting $g_\Gamma$ and $g_X$. Define the other half of the circle by
\begin{align}
    g_{-k_x,0}=X_{k_x,0}g_{k_x,0}^*,
    \qquad k_x\in[0,\pi].
\end{align}
The Takagi factorizations at the fixed points ensure that the two halves agree at their endpoints, giving a periodic $g$ on $k_y=0$. Repeat the construction on $k_y=\pi$.

For any integer $m$, we may modify either boundary loop by
\begin{align}
    g_{k_x,k_y}\longmapsto
    g_{k_x,k_y}\operatorname{diag}(e^{imk_x},1,\ldots,1),
    \qquad k_y=0,\pi.
\end{align}
The modification preserves Eq.~\eqref{eq:CI-factorization} and shifts the winding number $m_x[g]$ in the $k_x$ direction by $m$. At the fixed points, the added right factor is real orthogonal, so the Takagi factorizations are preserved as well. We can thus match the winding numbers of the two boundary loops. Since $\pi_1(U(N))\cong\Z$ is classified by the determinant winding number, the loops can then be continuously interpolated on the cylinder $T^1\times[0,\pi]$. Defining the other half of the torus by Eq.~\eqref{eq:CI-factorization} gives a continuous factorization over all of $T^2$.

Finally, let $m_x,m_y$ denote the determinant winding numbers of the resulting $g$ in the $k_x,k_y$ directions, and make the replacement
\begin{align}
    g_\bk\longmapsto
    g_\bk\operatorname{diag}
    \bigl(e^{-i(m_xk_x+m_yk_y)},1,\ldots,1\bigr).
\end{align}
The replacement sets both winding numbers to zero while preserving the factorization. These two integers determine the homotopy class of a map from $T^2$ to $U(N)$, so the resulting $g$ is homotopic to $1_N$.
\end{proof}

Let $g$ be the factor from the lemma and write $\iota(k_x,k_y)=(-k_x,-k_y)$. A contraction of $g$ to a constant also provides compatible extensions for the WZ terms of $g\circ\iota$, $g^\top$, and $X=(g\circ\iota)g^\top$. The PW formula then gives
\begin{align}
 \begin{split}
  {\rm WZ}[X]
  &\equiv {\rm WZ}[g\circ\iota]+{\rm WZ}[g^\top]\\
  &\quad+{\rm PW}[g\circ\iota,g^\top]\pmod1.
 \end{split}
 \label{eq:CI-WZ-PW}
\end{align}
Since $\iota$ preserves the orientation of $T^2$ and transposition reverses the sign of the WZ term,
\begin{align}
    {\rm WZ}[g\circ\iota]
    \equiv {\rm WZ}[g],\qquad
    {\rm WZ}[g^\top]
    \equiv-{\rm WZ}[g]\pmod1.
\end{align}
Thus the first two terms in Eq.~\eqref{eq:CI-WZ-PW} cancel.

For the cross term, setting $A=g^\dag dg$ gives
\begin{align}
    {\rm PW}[g\circ\iota,g^\top]
    \equiv-\frac{1}{8\pi^2}\int_{T^2}\tr[(\iota^*A)A^\top]\pmod1.
\end{align}
Here $\iota^*$ denotes the pullback of differential forms.
Cyclicity of the trace and anticommutativity of the wedge product give
\begin{align}
    \iota^*\tr[(\iota^*A)A^\top]
    =\tr[A(\iota^*A)^\top]
    =-\tr[(\iota^*A)A^\top].
\end{align}
The two-form changes sign under pullback by $\iota$, whereas the orientation is preserved. Its integral must therefore vanish, giving
\begin{align}
    {\rm WZ}[X]\equiv0\pmod1.
\end{align}
The vanishing of the boundary WZ terms in Eq.~\eqref{eq:CI-WZ-boundary} implies $W_3[X]\in2\Z$. Additivity of the winding number on the closed torus $T^3$ then gives
\begin{align}
    W_3[q]=W_3[X]-W_3[u]
    \equiv W_3[u]\pmod2.
\end{align}
This is the parity relation in Eq.~\eqref{eq:CI-winding-parity}.

Setting $q=u^\dag$ gives a model satisfying the symmetry. A unitary matrix $u$ with $W_3[u]\ne0$, however, cannot have finite range. Such a matrix would define
\begin{align}
    \widetilde H_\bk=\begin{pmatrix}0&u_\bk\\u_\bk^\dag&0\end{pmatrix},
\end{align}
which would be an exactly flat, finite-range class AIII model with a nonzero three-dimensional winding number, contradicting the no-go theorem for finite-range unitary matrices~\cite{Read-2017}.

\subsection{\texorpdfstring{$(d,d_\tau)=(3,1)$}{(d,d\_tau)=(3,1)}, class CII}
\label{sec:3d-CII}
Let $-\tau\bk=(k_x,-k_y,-k_z)$, and suppose that smooth periodic unitary matrices $q,u,v$ satisfy
\begin{align}
    u_\bk q_\bk^*v_\bk^\dag=q_{-\tau\bk},\qquad
    u_\bk^\top=-u_{-\tau\bk},\qquad
    v_\bk^\top=-v_{-\tau\bk}.
    \label{eq:alt62-symmetry}
\end{align}
On the fixed lines, $u,v$ are skew-symmetric, so their size must be even, $2n$. Since $-\tau$ preserves orientation in three dimensions and transposition reverses the sign of $W_3$,
\begin{align}
    W_3[u]=W_3[v]=0.
\end{align}
These vanishing winding numbers do not determine the parity of $W_3[q]$. We will instead lift the boundary WZ term to $\R/2\Z$ to derive a constraint that depends only on $u,v$.

\subsubsection{Mod 2 WZ term}
On the boundary planes $k_z=0,\pi$, write $\rho(k_x,k_y)=(k_x,-k_y)$.
\begin{lem}
\label{lem:alt62-factorization}
If $w:T^2\to U(2n)$ satisfies $w_\bk^\top=-w_{\rho\bk}$, there is a periodic unitary matrix $Q$ such that
\begin{align}
    w_\bk=Q_{\rho\bk}JQ_\bk^\top,
    \qquad J=\bigoplus_{j=1}^n i\sigma_y.
    \label{eq:alt62-factorization}
\end{align}
\end{lem}

\begin{proof}
On the fixed circles $k_y=0,\pi$, the skew-symmetric unitary matrix can be factored as $w=QJQ^\top$. The group $Sp(n)$ acting from the right is connected, so $Q$ can be chosen continuously and periodically on each fixed circle. The identity $\det w=(\det Q)^2$ gives $m_x[w]=2m_x[Q]$ on these circles. Because $w$ is defined on the cylinder $S^1\times[0,\pi]$, the winding number $m_x[Q]$ agrees on its two boundaries. We can therefore interpolate the boundary values of $Q$ inside the cylinder and extend the factor to the other half by
\begin{align}
    Q_{\rho\bk}=w_\bk Q_\bk^*J^{-1},
    \qquad 0\leq k_y\leq\pi.
\end{align}
The resulting continuous factorization satisfies Eq.~\eqref{eq:alt62-factorization}.
\end{proof}

To work mod 2, we retain the real-valued integral $\int_{T^2}{\rm pw}$ of the defining PW two-form, rather than reducing it to a value in $\R/\Z$. Substituting the boundary factorization gives
\begin{align}
    &{\rm WZ}[w]\notag\\
    &\quad\equiv {\rm WZ}[Q_\rho]+{\rm WZ}[Q^\top]
       +\int_{T^2}{\rm pw}(Q_\rho,JQ^\top)\notag\\
    &\qquad-\frac12\{m_x[Q_\rho]m_y[Q^\top]
       -m_y[Q_\rho]m_x[Q^\top]\}\notag\\
    &\quad\equiv-2{\rm WZ}[Q]+\int_{T^2}{\rm pw}(Q_\rho,JQ^\top)\notag\\
    &\qquad-m_x[Q]m_y[Q]\pmod1.
\end{align}
Here $Q_\rho=Q\circ\rho$.
Changing the real representative of ${\rm WZ}[Q]$ by an integer changes $-2{\rm WZ}[Q]$ by an even integer. The expression above therefore defines a value in $\R/2\Z$, lifting the WZ term to period $2$. We define
\begin{align}
    &{\rm WZ}_\rho[w]\notag\\
    &\quad:=-2{\rm WZ}[Q]+\int_{T^2}{\rm pw}(Q_\rho J,Q^\top)\notag\\
    &\qquad-m_x[Q]m_y[Q]
        \quad\in\R/2\Z.
    \label{eq:alt62-refined-WZ}
\end{align}
\footnote{A similar mod 2 refinement of the WZ term under the different symmetry $J q_\bk^* J^\dag = q_{-\bk},  J^\top = -J$ was introduced in Ref.~\cite{GawedzkiSquareRoot2017}.}
By construction,
\begin{align}
    {\rm WZ}_\rho[w]\equiv{\rm WZ}[w]\pmod1.
\end{align}

\begin{lem}
\label{lem:alt62-frame-independence}
The quantity defined in Eq.~\eqref{eq:alt62-refined-WZ} is independent of the choice of factor $Q$.
\end{lem}

\begin{proof}
Write an alternative factor as $Q'=QV$. The change of factor obeys
\begin{align}
    V_\rho JV^\top=J,
    \qquad V_\rho=JV^*J^{-1}.
    \label{eq:alt62-frame-change}
\end{align}
On the fixed circles, $\det V=1$, which implies $m_x[V]=0$ and $m_y[V]\in2\Z$.
We can further show that
\begin{align}
    {\rm WZ}[V]\equiv0\pmod1.
    \label{eq:alt62-gauge-WZ-zero}
\end{align}
Equation~\eqref{eq:alt62-frame-change} restricts ${\rm WZ}[V]$ to $0$ or $1/2$. We may therefore evaluate the WZ term after a continuous deformation.
The boundary loops of the cylinder $[-\pi,\pi] \times [0,\pi]$ lie in $Sp(n)$ and can be contracted to constants because $\pi_1(Sp(n))=0$.
A continuous deformation can then remove the $k_x$ dependence of $V$. For the resulting matrix, which depends only on $k_y$, we have ${\rm WZ}[V]=0$, proving Eq.~\eqref{eq:alt62-gauge-WZ-zero}.

Under this change of factor, the first term in Eq.~\eqref{eq:alt62-refined-WZ} becomes
\begin{align}
    &-2{\rm WZ}[QV]\notag\\
    &\quad\equiv-2{\rm WZ}[Q]-2{\rm WZ}[V]
       -2\int_{T^2}{\rm pw}(Q,V)\notag\\
    &\qquad+m_x[Q]m_y[V]-m_y[Q]m_x[V]\pmod2.
\end{align}
The PW identities~\eqref{eq:coh-pw-identities}, together with the orientation reversal under $\rho$, give the following equality of real integrals:
\begin{align}
    &\int_{T^2}{\rm pw}(Q_\rho V_\rho J,V^\top Q^\top)\notag\\
    &\quad=\int_{T^2}{\rm pw}(Q_\rho J,Q^\top)
       +2\int_{T^2}{\rm pw}(Q,V).
\end{align}
\footnote{Since $J$ is constant, Eq.~\eqref{eq:coh-pw-identities} gives
\begin{align*}
    &\int_{T^2}{\rm pw}(Q_\rho V_\rho J,V^\top Q^\top)\\
    &\quad=\int_{T^2}\bigl\{{\rm pw}(V_\rho J,V^\top)
       -{\rm pw}(V^\top,Q^\top)\\
    &\qquad\qquad+{\rm pw}(Q_\rho J,Q^\top)
       -{\rm pw}(Q_\rho,V_\rho J)\bigr\}\\
    &\quad=\int_{T^2}{\rm pw}(Q_\rho J,Q^\top)
         +2\int_{T^2}{\rm pw}(Q,V).
\end{align*}
In the last step, the first term vanishes because $V_\rho J=JV^*$. The remaining terms combine by the transposition identity and the reversal of the integration orientation under $\rho$.}
The winding-number correction is
\begin{align}
    &-m_x[QV]m_y[QV]\notag\\
    &\quad=-(m_x[Q]+m_x[V])(m_y[Q]+m_y[V]).
\end{align}
Because $m_x[V]=0$, $m_y[V]\in2\Z$, and ${\rm WZ}[V]\equiv0\pmod1$, the quantity ${\rm WZ}_\rho[w]$ is independent of the factor $Q$.
\end{proof}

\subsubsection{Parity of the three-dimensional winding number}
Orient the fundamental region $V=T^2\times[0,\pi]$ of momentum space by $dk_x\wedge dk_y\wedge dk_z$, and orient $\Sigma_z=\{k_z=z\}$ by $dk_x\wedge dk_y$. With these conventions, $\partial V=\Sigma_\pi-\Sigma_0$.
For $w$ satisfying $w^\top_{\bk} = - w_{-\tau\bk}$, define the $\Z_2$ invariant $\mu[w] \in \{0,1\}$ by
\begin{align}
 \begin{split}
  \mu[w]
  &:=\int_V\omega_3(w)-{\rm WZ}_\rho[w;\Sigma_\pi]\\
  &\quad+{\rm WZ}_\rho[w;\Sigma_0]\pmod2.
 \end{split}
 \label{eq:alt62-nu}
\end{align}
Equation~\eqref{eq:coh-wz-boundary} and ${\rm WZ}_\rho\equiv{\rm WZ}\pmod1$ ensure that the defining expression is an integer $\bmod2$.

For $w_1,w_2$ satisfying the same symmetry, $\mu$ is additive under direct sums:
\begin{align}
    \mu[w_1\oplus w_2]\equiv\mu[w_1]+\mu[w_2]\pmod2.
    \label{eq:CII-mu-direct-sum}
\end{align}
To prove additivity, factor $w_i=(Q_i)_\rho J_iQ_i^\top$ on each boundary and set $a_i=m_x[Q_i]$ and $b_i=m_y[Q_i]$. In Eq.~\eqref{eq:alt62-global-PW} applied to $(Q_1\oplus I)(I\oplus Q_2)$, the PW term vanishes. The difference between $-2{\rm WZ}$ evaluated on the direct sum and the sum of its values on the two factors is therefore $a_1b_2-b_1a_2$ $\bmod2$. The PW integral in Eq.~\eqref{eq:alt62-refined-WZ} is additive under direct sums. Including the final winding-number term yields
\begin{align}
    &{\rm WZ}_\rho[w_1\oplus w_2]
       -\sum_{i=1}^2{\rm WZ}_\rho[w_i]\notag\\
    &\quad\equiv(a_1b_2-b_1a_2)-(a_1b_2+a_2b_1)\notag\\
    &\quad=-2a_2b_1\equiv0\pmod2.
\end{align}
The volume integral is also additive because $\omega_3(w_1\oplus w_2)=\omega_3(w_1)+\omega_3(w_2)$. Hence Eq.~\eqref{eq:CII-mu-direct-sum} follows.

\begin{thm}
If $q,u,v:T^3\to U(2n)$ satisfy Eq.~\eqref{eq:alt62-symmetry}, then
\begin{align}
    W_3[q]\equiv\mu[u]-\mu[v]\pmod2.
    \label{eq:CII-3d-winding-parity}
\end{align}    
\end{thm}

\begin{proof}
Writing the symmetry as $u q^*=q_{-\tau} v$ and using Eq.~\eqref{eq:alt62-local-PW}, we obtain
\begin{align}
    \omega_3(q_{-\tau})
    &=\omega_3(q)+\omega_3(u)-\omega_3(v)\notag\\
    &\quad+d\,{\rm pw}(u,q^*)-d\,{\rm pw}(q_{-\tau},v).
\end{align}
The symmetry under the momentum action $-\tau$ gives
\begin{align}
    W_3[q]
    &=\int_V\{\omega_3(q)+\omega_3(q_{-\tau})\}
    \notag\\
    &\equiv \mu[u]-\mu[v]\notag\\
    &\quad+\Bigl[2{\rm WZ}[q]+{\rm WZ}_\rho[u]-{\rm WZ}_\rho[v]
    \notag\\[-2pt]
    &\qquad
       +\int_{T^2}\{{\rm pw}(u,q^*)-{\rm pw}(q_\rho,v)\}\Bigr]_{k_z=0}^{k_z=\pi}
       \pmod2.
    \label{eq:alt62-half-BZ}
\end{align}

It remains to evaluate the boundary term. On each boundary, factor $v=B_\rho JB^\top$ and set $A=qB$. The symmetry in Eq.~\eqref{eq:alt62-symmetry} then gives $u=A_\rho JA^\top$.
The PW identities
\begin{align}
    &\int_{T^2}\{{\rm pw}(a,b)+{\rm pw}(ab,c)\}\notag\\
    &\quad=\int_{T^2}\{{\rm pw}(b,c)+{\rm pw}(a,bc)\},\notag\\
    &\int_{T^2}{\rm pw}(a^*,b^*)=\int_{T^2}{\rm pw}(a,b)
       =-\int_{T^2}{\rm pw}(b^\top,a^\top)
\end{align}
and $\int_{T^2}{\rm pw}(a_\rho,b_\rho)=-\int_{T^2}{\rm pw}(a,b)$ give
\begin{align}
    &\int_{T^2}\{{\rm pw}(u,q^*)-{\rm pw}(q_\rho,v)\}\notag\\
    &\quad=-2\int_{T^2}{\rm pw}(A,B^\dag)
       -\int_{T^2}{\rm pw}(A_\rho J,A^\top)\notag\\
    &\qquad+\int_{T^2}{\rm pw}(B_\rho J,B^\top).
    \label{eq:alt62-boundary-PW}
\end{align}
Applying Eq.~\eqref{eq:alt62-global-PW} to $q=AB^\dag$ and substituting Eqs.~\eqref{eq:alt62-refined-WZ} and \eqref{eq:alt62-boundary-PW}, we obtain
\begin{align}
    &2{\rm WZ}[q]+{\rm WZ}_\rho[u]-{\rm WZ}_\rho[v]\notag\\
    &\quad+\int_{T^2}\{{\rm pw}(u,q^*)-{\rm pw}(q_\rho,v)\}\notag\\
    &\quad\equiv
       -(m_x[A]m_y[B]-m_y[A]m_x[B])\notag\\
    &\qquad-m_x[A]m_y[A]+m_x[B]m_y[B]\notag\\
    &\quad\equiv m_x[q]m_y[q]\pmod2.
\end{align}
The one-dimensional winding numbers $m_x[q],m_y[q]$ agree at $k_z=0,\pi$. The boundary difference in Eq.~\eqref{eq:alt62-half-BZ} therefore vanishes, proving the stated parity relation.
\end{proof}

The invariant is related to previous work on gapless modes at defects~\cite{Teo-Kane}. 
If $(k_y,k_z)$ are interpreted as momenta and $k_x$ as the parameter of a circle surrounding a defect, then $w$ corresponds to the off-diagonal block of a point-defect Hamiltonian in two-dimensional class DIII. 
When a global factorization $w_\bk=Q_{-\tau\bk}JQ_\bk^\top$ exists, substituting $(q,u,v)=(Q,w,J)$ into Eq.~\eqref{eq:CII-3d-winding-parity} and using $\mu[J]=0$ gives
\begin{align*}
    \mu[w]\equiv W_3[Q]\pmod2.
\end{align*}
The resulting winding-number formula agrees with Ref.~\cite{Teo-Kane}.
The definition in Eq.~\eqref{eq:alt62-nu} requires only the boundary factorizations established in Lemma~\ref{lem:alt62-factorization}; a global factorization over all three dimensions is not assumed.

\subsubsection{Locality}
We next ask whether the symmetry matrices $w=u,v$ can have finite range.
The antiunitary symmetry with momentum action $-\tau\bk=(k_x,-k_y,-k_z)$ acts in real space by the reflection $\tau\bx=(-x,y,z)$.
Below, we extract two-dimensional data associated with the reflection plane and apply the finite-range no-go theorem~\cite{Read-2017} to the bundles defined by their images, showing that finite-range symmetry matrices cannot realize a nonzero value of the invariant.

\begin{thm}
\label{thm:CII-finite-range}
Let $-\tau\bk=(k_x,-k_y,-k_z)$. Suppose that a unitary matrix $w:T^3\to U(2n)$ satisfies $w_\bk^\top=-w_{-\tau\bk}$ and has finite range in all directions, meaning that it can be written using a finite set $F\subset\Z^3$ as
\begin{align}
    w_\bk=\sum_{\bR\in F}w_\bR e^{i\bk\cdot\bR}.
    \label{eq:CII-finite-range}
\end{align}
Then
\begin{align}
    \mu[w]=0.
    \label{eq:CII-finite-mu-zero}
\end{align}
Consequently, if both $u,v$ satisfying Eq.~\eqref{eq:alt62-symmetry} have finite range, then
\begin{align}
    W_3[q]\in2\Z.
    \label{eq:CII-finite-nogo}
\end{align}
\end{thm}

\begin{proof}
Applying Eq.~\eqref{eq:CII-3d-winding-parity} to $q=g_{-\tau}$ gives
\begin{align}
    \mu[g w g_{-\tau}^\top]-\mu[w]\equiv W_3[g]\pmod2.
    \label{eq:CII-read-congruence}
\end{align}
Replacing the matrix by $e^{2i\ell k_x}w$ leaves $\mu$ unchanged. We can therefore make all Fourier exponents in $k_x$ nonnegative.
Next, enlarge the unit cell in the $x$ direction by a sufficiently large odd factor, keeping the reflection center fixed.
The enlarged unit cell describes the same real-space operator and leaves the strong invariant $\mu$ unchanged. With $\mu[W]=\mu[w]$, the matrix can be written as
\begin{align}
    W(X,p)=A(p)+XB(p),\qquad X=e^{iK},\quad p=(k_y,k_z).
    \label{eq:CII-read-linear}
\end{align}
The intracell change of basis depends only on $k_x$, so the winding-number correction in Eq.~\eqref{eq:CII-read-congruence} also vanishes.

Unitarity and the transposition symmetry give
\begin{align}
    AA^\dag+BB^\dag&=A^\dag A+B^\dag B=I,\notag\\
    AB^\dag&=A^\dag B=0,\notag\\
    A_p^\top&=-A_{-p},\qquad B_p^\top=-B_{-p}.
\end{align}
Thus $P_A=AA^\dag$ and $P_B=BB^\dag$ are complementary orthogonal projectors,
and their images $E_A=\operatorname{Ran}A$ and $E_B=\operatorname{Ran}B$ define bundles of constant rank.
The bundle $E_A$ carries a time-reversal action squaring to $-1$, which sends $v\in(E_A)_p$ to $-A_{-p}\bar v\in(E_A)_{-p}$.
The bundle $E_B$ carries an analogous time-reversal action.
The no-go theorem~\cite{Read-2017} then forces the Kane--Mele invariant of each bundle to vanish.\footnote{One can also establish the vanishing invariant by constructing a class AII Hamiltonian directly.
For $M=A,B$, consider the matrices
\begin{align*}
    \widetilde H_M&=-\begin{pmatrix}&M\\M^\dag&\end{pmatrix},
    \qquad J_0=\begin{pmatrix}0&I\\-I&0\end{pmatrix}.
\end{align*}
The Hamiltonian satisfies $J_0\widetilde H_M(p)^*J_0^\dag=\widetilde H_M(-p)$. Zero eigenvalues may be present, but the projector $P_-=(\widetilde H_M^2-\widetilde H_M)/2$ onto the isolated $-1$ band has finite range. Read's theorem therefore applies to the gapped flat Hamiltonian $I-2P_-$.
The occupied bundle of the flat Hamiltonian consists of states $(v,M_p^\dag v)^\top/\sqrt2$ with $v \in \operatorname{Ran} M$.}
The two bundles therefore admit smooth global orthonormal frames $U_p,V_p$ that respect time reversal:
\begin{align}
    P_AU_p=U_p,\qquad U_p^\dag U_p&=I,\qquad
    A_pU_{-p}^*=-U_pJ_A,\notag\\
    P_BV_p=V_p,\qquad V_p^\dag V_p&=I,\qquad
    B_pV_{-p}^*=-V_pJ_B.
\end{align}
Here $J_A,J_B$ are the standard skew-symmetric matrices $i\sigma_y\otimes I$ with sizes matching the respective ranks; rank-zero components are omitted.
The two image bundles are orthogonal complements, so $Q(p)=(U_p,V_p)$ is unitary and gives the factorization
\begin{align}
    W(X,p)&=Q(p)D(X)Q(-p)^\top,\notag\\
    D(X)&=\operatorname{diag}(-J_A,-XJ_B).
\end{align}
The factor $Q$ depends only on $p$, so $W_3[Q]=0$. The factor $D$ depends only on $K$, so Eq.~\eqref{eq:alt62-nu} gives $\mu[D]=0$.
Equation~\eqref{eq:CII-read-congruence} then gives $\mu[w]=\mu[W]=0$.
\end{proof}

\subsubsection{Example}
If $W_3[q]$ is odd, at least one of $u,v$ cannot have finite range.
Exponentially decaying symmetry matrices can still realize an odd winding number.
For any globally defined unitary matrix $q:T^3\to U(2)$, set
\begin{align}
    v_\bk=J:=i\sigma_y,\qquad
    u_\bk=q_{-\tau\bk}Jq_\bk^\top.
    \label{eq:CII-model-sewing}
\end{align}
Then $u_\bk^\top=-u_{-\tau\bk}$, $v_\bk^\top=-v_{-\tau\bk}$, and $u_\bk q_\bk^*v_\bk^\dag=q_{-\tau\bk}$ hold.
A matrix $q$ with odd winding number therefore gives symmetry matrices with $\mu[u]-\mu[v]=1$.
To obtain $q$ with winding number $1$, take the standard $2 \times 2$ Dirac model
\begin{align}
    D_\bk&=d_0(\bk)1_2+i\sum_{j=1}^3d_j(\bk)\sigma_j,\notag\\
    d_0(\bk)&=-2+\cos k_x+\cos k_y+\cos k_z,\notag\\
    d_j(\bk)&=\sin k_j\quad(j=1,2,3),
\end{align}
and flatten the model.

\subsection{\texorpdfstring{$(d,d_\tau)=(3,2)$}{(d,d\_tau)=(3,2)}, class DIII}
\label{sec:3d-DIII}
Write $-\tau\bk=(k_x,k_y,-k_z)$. The symmetry condition is
\begin{align}
    (q_\bk u_\bk)^\top=-q_{-\tau\bk}u_{-\tau\bk}.
    \label{eq:alt63-symmetry}
\end{align}
With $X_\bk:=q_\bk u_\bk$, the symmetry condition reads $X_\bk^\top=-X_{-\tau\bk}$. We will prove $W_3[X]\in2\Z$ and hence show that $W_3[q]$ and $W_3[u]$ have the same parity.

Orient $T^3$ by $dk_x\wedge dk_y\wedge dk_z$ and take the half Brillouin zone $V=T^2\times[0,\pi]$. Transposition and momentum reversal by $-\tau$ both reverse the sign of the three-dimensional winding number, whereas multiplication by the constant $-1$ does not. Equation~\eqref{eq:alt63-symmetry} therefore makes the contributions from the two halves equal:
\begin{align}
    W_3[X]=2\int_V\omega_3(X).
\end{align}
With both boundary tori $T^2$ oriented by $dk_x\wedge dk_y$, we obtain
\begin{align}
 \begin{split}
  \frac12 W_3[X]
  &\equiv {\rm WZ}[X;\{k_z=\pi\}]\\
  &\quad-{\rm WZ}[X;\{k_z=0\}]\pmod1.
 \end{split}
 \label{eq:alt63-WZ-boundary}
\end{align}
On the fixed planes $k_z=0,\pi$, the condition $X_\bk^\top=-X_\bk$ forces the matrix size to be even, $2N$. The following lemma shows that the right-hand side of Eq.~\eqref{eq:alt63-WZ-boundary} vanishes.

\begin{lem}
\label{lem:alt63-skew-WZ}
If a smooth periodic unitary matrix $X:T^2\to U(2N)$ satisfies $X^\top=-X$, then
\begin{align}
    {\rm WZ}[X]\equiv0\pmod1.
    \label{eq:alt63-skew-WZ}
\end{align}
\end{lem}

\begin{proof}
Set $J:=\bigoplus_{j=1}^N i\sigma_y$. The pointwise factorization $X=QJQ^\top$ determines $Q$ up to right multiplication by $Sp(N)$.
Since $Sp(N)$ is connected and simply connected, there is no obstruction to choosing $Q$ periodically on $T^2$.
As in Sec.~\ref{sec:wz-pw}, extend $Q$ to a compact oriented three-dimensional manifold $M$ with $\partial M=T^2$ and set $\widetilde X=\widetilde QJ\widetilde Q^\top$. This extension also satisfies $\widetilde X=-\widetilde X^\top$, so Eq.~\eqref{eq:coh-w3-transform} gives $\omega_3(\widetilde X)=-\omega_3(\widetilde X)=0$. The defining integral of the WZ term therefore vanishes.
\end{proof}

Applying the lemma on both fixed planes reduces Eq.~\eqref{eq:alt63-WZ-boundary} to $W_3[X]\in2\Z$. Additivity of the winding number under multiplication then gives
\begin{align}
    W_3[q]=W_3[X]-W_3[u]\equiv W_3[u]\pmod2.
    \label{eq:alt63-winding-parity}
\end{align}

Thus, an odd $W_3[q]$ requires an odd $W_3[u]$, and the no-go theorem~\cite{Read-2017} rules out a finite-range realization of $u$.
The construction in Sec.~\ref{sec:oddmodel-chiral} gives an exponentially decaying model: choose $q$ by flattening a Dirac model with odd winding number and set $u=q^\dagger J$, $J=i\sigma_y$.

\subsection{\texorpdfstring{$(d,d_\tau)=(3,3)$}{(d,d\_tau)=(3,3)}, class BDI}
\label{sec:3d-BDI}
Consider the momentum-preserving BDI symmetry
\begin{align}
u_\bk q_\bk^*v_\bk^\dag=q_\bk,\qquad
u_\bk^\top= u_\bk, v_\bk^\top=v_\bk.
\label{eq:alt64-symmetry}
\end{align}
Assume that $q_\bk,u_\bk,v_\bk\in U(N)$ are smooth and periodic on all of $T^3$. We first define a $\Z_2$ invariant of the symmetry matrices $u,v$ that determines the parity of $W_3[q]$, then construct a model with an odd winding number.

\subsubsection{Spin Chern--Simons term and \texorpdfstring{$\Z_2$}{Z2} invariant}
For a family of symmetric unitary matrices $X:T^3\to U(N), X^\top =X$, let $\mathcal P_X$ be the real bundle defined by local Takagi factorizations $X=UU^\top$, with real Berry connection $A_X=\tfrac12[U^\dag dU+(U^\dag dU)^*]$. Choose a common auxiliary metric and fix a spin structure $\mathfrak s$. The class CI strong index associated with $\pi_3(R_{-1})=\Z_2$ is
\begin{align}
\zeta[X;\mathfrak s]
&:=\xi[A_X;\mathfrak s]-N\xi_{\rm triv}[\mathfrak s]\in \{0,1\}
\pmod 2,\label{eq:alt64-nu-definition_1}\\
\xi[A_X;\mathfrak s]
&=\frac{\eta_{iD_{A_X}}(0)+\dim\ker(iD_{A_X})}{2},
\label{eq:alt64-nu-definition_2}
\end{align}
as defined in Ref.~\cite{ShiozakiCI2026}.
Here $iD_{A_X}$ is the Dirac operator coupled to the real bundle $\mathcal P_X$, $\xi[A_X;\mathfrak s]$ is its reduced eta invariant, and $\xi_{\rm triv}$ is the reduced eta invariant for the trivial real line bundle.
The difference of reduced eta invariants is metric independent, and $\zeta[X;\mathfrak s]$ is quantized to $\Z_2$ on $T^3$.
In general, $\zeta[X;\mathfrak s]$ depends on the spin structure $\mathfrak s$.
The difference on the right-hand side of Eq.~\eqref{eq:alt64-nu-definition_1}, multiplied by $2\pi$, gives the spin CS term of the real bundle~\cite{DijkgraafWitten1990,JenquinSpinCS2005}:
\begin{align}
    {\rm CS}_{\rm spin}[A;\mathfrak s]
    =2\pi (\xi[A;\mathfrak s]-N\xi_{\rm triv}[\mathfrak s]) \quad \pmod{4\pi}. 
\end{align}
The spin CS term is defined for a general real Berry connection. In class CI, PHS quantizes its value to $0,2\pi$.

The properties of $\zeta$ needed below follow from Ref.~\cite{ShiozakiCI2026} and the definition just given.

For a fixed spin structure $\mathfrak s$, direct sums of symmetric unitary matrices $X,Y$ obey
\begin{align}
 \zeta[X\oplus Y;\mathfrak s]
 \equiv\zeta[X;\mathfrak s]+\zeta[Y;\mathfrak s]\pmod2.
 \label{eq:alt64-zeta-additivity}
\end{align}
A direct sum of local Takagi factors gives a direct sum of the real Berry connections and of the coupled Dirac operators. Additivity of the $\eta$ invariant, the kernel dimension, and the trivial-bundle subtraction then proves the direct-sum law from Eqs.~\eqref{eq:alt64-nu-definition_1} and \eqref{eq:alt64-nu-definition_2}.

To study changes of spin structure, let $\ell_j$ be the generator of $H^1(T^3;\Z_2)$ corresponding to the $k_j$ circle, and let $\mathfrak s_{\rm NS}$ denote the spin structure that is antiperiodic in every direction. The shift $\mathfrak s\mapsto\mathfrak s+\ell_j$ exchanges periodic and antiperiodic boundary conditions in the $j$th direction. With a flat metric, the nonzero eigenvalues of the trivial Dirac operator occur in positive and negative pairs. Zero modes occur only when all directions are periodic, so
\begin{align}
 \xi_{\rm triv}\!\left[\mathfrak s_{\rm NS}
       +\sum_{j=1}^3\epsilon_j\ell_j\right]
 &=\epsilon_1\epsilon_2\epsilon_3,
 \qquad \epsilon_j\in\{0,1\}.
 \label{eq:alt64-trivial-xi-spin}
\end{align}

For any $\ell\in H^1(T^3;\Z_2)$, write
$a=w_1(\mathcal P_X)$. The change in the invariant is
\begin{align}
 &\zeta[X;\mathfrak s+\ell]-\zeta[X;\mathfrak s]\notag\\
 &\quad\equiv\langle w_2(\mathcal P_X)\cup\ell,[T^3]\rangle\notag\\
 &\qquad-\xi_{\rm triv}[\mathfrak s+\ell+a]
          +\xi_{\rm triv}[\mathfrak s+\ell]\notag\\
 &\qquad+\xi_{\rm triv}[\mathfrak s+a]
          -\xi_{\rm triv}[\mathfrak s]\pmod2.
 \label{eq:alt64-zeta-spin-change}
\end{align}
The four trivial-bundle terms can be evaluated using Eq.~\eqref{eq:alt64-trivial-xi-spin}.
These four terms cancel when $w_1(\mathcal P_X)=0$, but must be retained in general.

Multiplication by a phase factor changes the invariant as follows. For $n_j\in\Z$ and $\ell=\sum_j n_j\ell_j$,
\begin{align}
 \zeta[e^{i\sum_j n_jk_j}X;\mathfrak s]
 &\equiv\zeta[X;\mathfrak s+\ell]\notag\\
 &\quad+N\bigl(\xi_{\rm triv}[\mathfrak s+\ell]
                     -\xi_{\rm triv}[\mathfrak s]\bigr)\pmod2.
 \label{eq:alt64-zeta-phase-shift}
\end{align}
The phase factor replaces the local Takagi factor $U$ by $e^{i\sum_j n_jk_j/2}U$. The real Berry connection remains unchanged; only the gluing around the $j$th direction acquires the factor $(-1)^{n_j}$. For the coupled Dirac operator, this twist is equivalent to the spin-structure shift $\mathfrak s\mapsto\mathfrak s+\ell$. The trivial-bundle subtraction in Eq.~\eqref{eq:alt64-nu-definition_1} produces the correction in Eq.~\eqref{eq:alt64-zeta-phase-shift}. A phase factor can therefore be accounted for solely by a spin-structure shift only when the correction vanishes or cancels.

The invariant also vanishes under a simple condition on momentum dependence: if $X$ is independent of $k_j$ and $\mathfrak s$ is antiperiodic in the $j$th direction, then
\begin{align}
 \zeta[X;\mathfrak s]=0\pmod2.
 \label{eq:alt64-zeta-two-variable}
\end{align}
The argument for fully antiperiodic boundary conditions in the cited work on the CI invariant also applies when only the $j$th direction is antiperiodic. With a flat product metric, the eigenvalues in the $j$th direction are nonzero half-integers. The full product Dirac spectrum is consequently paired between positive and negative energies, with no zero modes. The trivial-bundle contribution vanishes as well, by Eq.~\eqref{eq:alt64-trivial-xi-spin}.

Finally, for any smooth $g:T^3\to U(N)$, the invariant satisfies
\begin{align}
\zeta[gXg^\top;\mathfrak s]-\zeta[X;\mathfrak s]
\equiv W_3[g]\pmod 2.
\label{eq:alt64-nu-transformation}
\end{align}
To derive this transformation law, identify the real bundles, interpolate between their real connections, and apply the Atiyah--Patodi--Singer index theorem~\cite{AtiyahPatodiSinger1975I} on $T^3\times[0,1]$. The symmetry condition~\eqref{eq:alt64-symmetry} can be written as $u_\bk=q_\bk v_\bk q_\bk^\top$. Substituting $X=v$, $g=q$ into Eq.~\eqref{eq:alt64-nu-transformation} then gives the parity constraint
\begin{align}
W_3[q]\equiv\zeta[u;\mathfrak s]-\zeta[v;\mathfrak s]\pmod 2.
\label{eq:alt64-parity}
\end{align}
Whenever a gapped $q$ exists, the symmetry matrices $u,v$ therefore determine the parity of its winding number. Their invariant difference equals $W_3[q]\bmod2$, so the difference is independent of the spin structure. In particular, $u=v=1_N$ implies $q:T^3\to O(N)$ and $W_3[q]\in2\Z$.

\subsubsection{Realization with finite-range symmetry actions}
Choose the symmetry matrices
\begin{align}
u_\bk&=\operatorname{diag}(1,e^{ik_1},e^{ik_2},e^{i(k_1+k_2)}),\notag\\
v_\bk&=e^{ik_3}u_\bk.
\label{eq:alt64-torus-symmetries}
\end{align}
For the spin structure $\mathfrak s_{\rm NS}$ that is antiperiodic in all directions, the real-line-bundle calculation in Ref.~\cite{ShiozakiCI2026} gives
\begin{align}
&\zeta[e^{i(n_1k_1+n_2k_2+n_3k_3)};\mathfrak s_{\rm NS}]\notag\\
&\quad\equiv n_1n_2n_3\pmod2,
\qquad (n_1,n_2,n_3)\in\Z^3.
\end{align}
Additivity under direct sums, Eq.~\eqref{eq:alt64-zeta-additivity}, then gives
\begin{align}
\zeta[u;\mathfrak s_{\rm NS}]=0,\qquad
\zeta[v;\mathfrak s_{\rm NS}]=1.
\end{align}
Every gapped model with these symmetry matrices must therefore have an odd winding number.

To construct a model $q_\bk$ with these symmetries, denote the two-dimensional Hamiltonian in Eq.~\eqref{eq:D_four_band_H} by $H_{k_1,k_2}$. It satisfies $u_\bk H_{k_1,k_2}^* = -H_{k_1,k_2}u_\bk$. Define
\begin{align}
D_\bk
&=\frac{1+e^{-ik_3}}2\,1_4
  +\frac{1-e^{-ik_3}}2\,H_{k_1,k_2},\notag\\
\mathcal H_\bk
&=\begin{pmatrix}0&D_\bk^\dag \\D_\bk&0\end{pmatrix}.
\label{eq:alt64-finite-range-model}
\end{align}
Take $t>0$ and $t_2\ne0$. The symmetry in Eq.~\eqref{eq:D_four_band_C} gives
\begin{align}
u_\bk D_\bk^*v_\bk^\dag&=D_\bk,\notag\\
D_\bk^\dag D_\bk
&=\cos^2\frac{k_3}{2}\,1_4
  +\sin^2\frac{k_3}{2}\,H_{k_1,k_2}^2.
\end{align}
Thus $\mathcal H$ is fully gapped, with minimum positive energy $\min\{1,t,2|t_2|\}>0$.

We now show that the winding number of $D_\bk$ equals the Chern number of $H_{k_1,k_2}$.
Deform $H_{k_1,k_2}$ to $\operatorname{sgn}H_{k_1,k_2}$ without closing the gap. The matrix $D_\bk$ remains invertible, as follows from the expression for $D_\bk^\dag D_\bk$ above, so the winding number is unchanged. The positive- and negative-energy projectors $P_+,P_-$ of the flattened $H_{k_1,k_2}$ give
\begin{align}
q_\bk=P_+(k_1,k_2)+e^{-ik_3}P_-(k_1,k_2).
\label{eq:alt64-suspension-q}
\end{align}
Substituting this unitary into the winding-number formula and integrating over $k_3$ gives
\begin{align}
W_3[q]
&=\frac{i}{2\pi}\int_{T^2}\Tr(P_-dP_-dP_-)
={\rm ch}_1[H].
\label{eq:alt64-model-winding}
\end{align}

\begin{figure*}[t]
\centering
\includegraphics[width=0.6\textwidth]{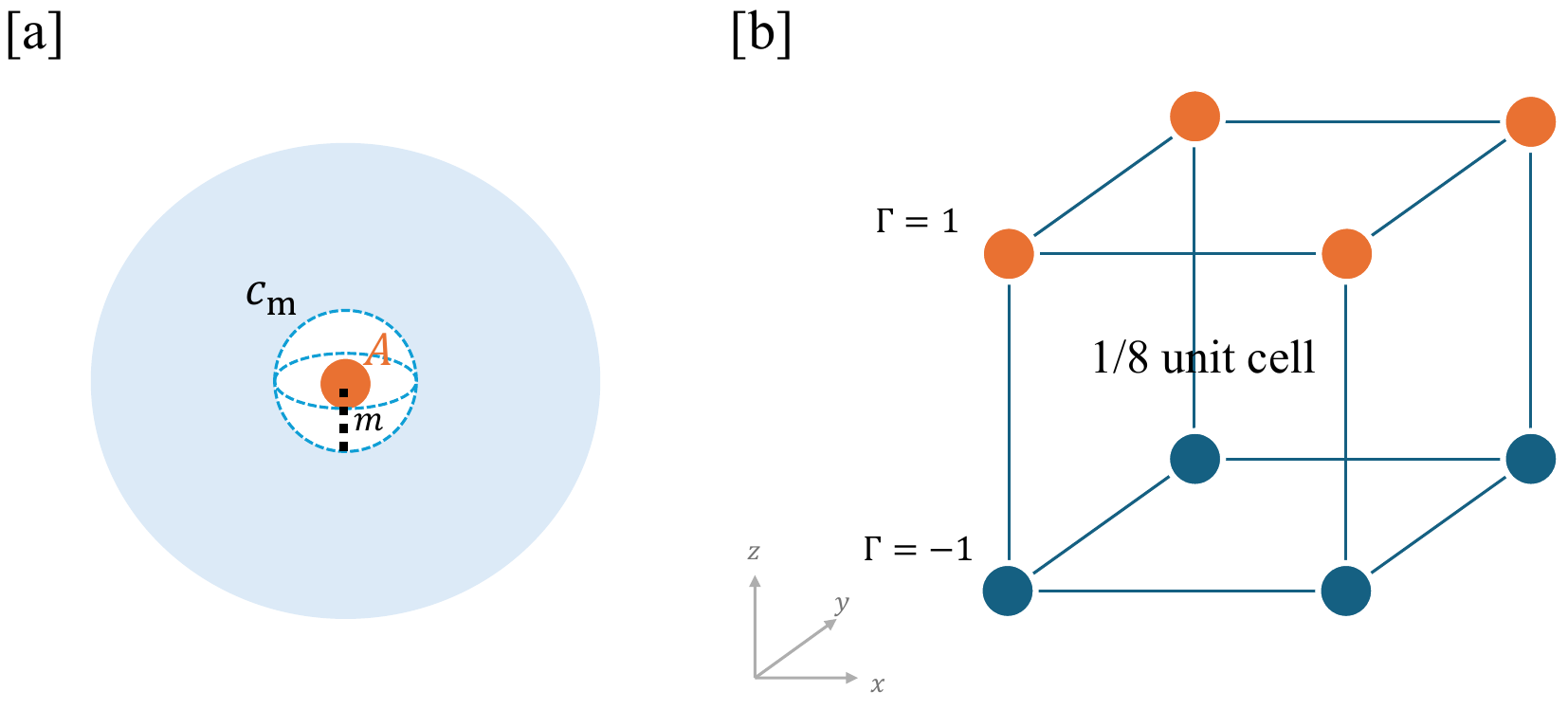}
\caption{Real-space SPT--LSM mechanism for three-dimensional class BDI.
[a] The shaded region represents a class AIII bulk with an odd winding number in a unit-monopole background.
The monopole binds a zero mode $c_{\rm m}$, indicated schematically by the blue dashed sphere.
The orange dot is a microscopic orbital $A$ at the inversion center whose chiral eigenvalue is opposite to that of the bound state.
The black dashed bond represents a symmetry-preserving hybridization with amplitude $m$ that gaps the orbital and the bound state.
[b] Periodic orbital arrangement for the symmetry matrices in Eq.~\eqref{eq:alt64-torus-symmetries}.
Each of the eight inequivalent inversion centers in a unit cell hosts one complex-fermion orbital.
Orange dots in the upper plane have $\Gamma=+1$, and blue dots in the lower plane have $\Gamma=-1$.
The displayed cube spans half a lattice period in each direction.
}
\label{fig:3d-BDI-SPT-LSM}
\end{figure*}

\subsubsection{Relation to SPT--LSM theorems}
\label{sec:3d-BDI-SPT-LSM}
The requirement of an odd winding number can also be interpreted as an SPT--LSM-type constraint on free fermions.
We develop this real-space interpretation using the number-conserving realization in Appendix~\ref{sec:summary-manybody}.

Taken together, the diagonal entries of $u,v$ in Eq.~\eqref{eq:alt64-torus-symmetries} contain each factor $e^{i\bm n\cdot\bk}$ with $\bm n\in\{0,1\}^3$ exactly once.
These diagonal symmetry matrices can be realized by antiunitary spatial inversion $PT$ acting on orbitals at positions $\bm n/2$ within the unit cell. The momentum factors record the cell displacements under inversion. The action is therefore non-cellwise in the fixed Bloch basis, even though it simply maps each orbital to its geometric inversion image.
The model can thus be viewed as placing one complex-fermion orbital at each of the eight inequivalent inversion centers in a unit cell, with the $n_3=0,1$ orbitals in opposite chiral subspaces.
This orbital arrangement is shown in Fig.~\ref{fig:3d-BDI-SPT-LSM}[b].
When each orbital is isolated, the chiral condition fixes its onsite energy to zero.
An isolated orbital can be gapped by coupling it to a zero mode with the opposite chiral eigenvalue, but no such partner is present at the same inversion center.
To explain how the surrounding three-dimensional phase supplies a partner, we now consider real space $\R^3$ with a single inversion center and impose no lattice translations.

For a three-dimensional class AIII insulator with cellwise chiral symmetry, the axion angle is $\theta\equiv\pi W_3[q]\pmod{2\pi}$~\cite{RyuSchnyderFurusakiLudwig2010}.
For an odd winding number, the Witten effect assigns to a unit monopole of the background $U(1)$ field a bound charge $Q/e\equiv1/2\pmod1$~\cite{RosenbergFranz2010,AokiFukayaKanKoshinoMatsuki2023}.
If chiral symmetry is preserved even at the monopole core, the positive- and negative-energy states pair with the same local density. Zero modes must then carry the charge deviation from half filling.
An odd number of zero modes therefore appears at a unit monopole.

We must also check compatibility with $(PT)^2=1$.
As a representative with winding number of absolute value one, we can choose a Dirac model of a conventional three-dimensional topological insulator with chiral and inversion symmetries.
Its antiunitary inversion acts as $(PT\psi)(\bx)=i\sigma_y \psi(-\bx)^*$ and satisfies $(PT)^2=-I$.
A monopole with magnetic charge $m\in\Z$ contributes an additional factor $(-1)^m$ through the gauge gluing under inversion, changing the square to $(PT)^2=(-1)^{m+1}I$.\footnote{Cover a sphere around the monopole with northern and southern gauge patches. The overlap relation is $\psi_N(\bx)=e^{im\phi}\psi_S(\bx)$.
Here $\phi$ is the azimuthal angle of $\bx$. The inversion image $-\bx$ has azimuthal angle $\phi+\pi$, so
$\psi_S(-\bx)^*=e^{im(\phi+\pi)}\psi_N(-\bx)^*$.
For the transformed wave function $\psi'=PT\psi$, define the northern component by $\psi'_N(\bx)=i\sigma_y \psi_S(-\bx)^*$.
Requiring $\psi'$ to obey the original gluing condition $\psi'_S=e^{-im\phi}\psi'_N$ fixes the southern component:
$\psi'_S(\bx)=e^{-im\phi}i\sigma_y \psi_S(-\bx)^*
 =e^{-im\phi}e^{im(\phi+\pi)}i\sigma_y \psi_N(-\bx)^*
=(-1)^m i\sigma_y \psi_N(-\bx)^*$.
Applying the antiunitary operation twice gives $((PT)^2\psi)_N=i\sigma_y [(-1)^m i\sigma_y \psi_N^*]^*=-(-1)^m\psi_N$.}
Thus a class AIII insulator with odd winding number is compatible with $(PT)^2=1$ in a unit-monopole background.
In Fig.~\ref{fig:3d-BDI-SPT-LSM}[a], the monopole-bound state is chosen to have the opposite chiral eigenvalue from the orbital at the center. The two modes can then be coupled and gapped while preserving the symmetry.

For a single inversion center, this cancellation of zero modes corresponds to a nontrivial third differential of the real-space AHSS,
\begin{align}
 d^3_{3,-3}:E^3_{3,-3}\cong\Z
 &\longrightarrow E^3_{0,-1}\cong\Z_2,\notag\\
 n&\longmapsto n\bmod2,
 \label{eq:3d-BDI-SPT-LSM-differential}
\end{align}
as in Ref.~\cite{Shiozaki-Xiong-Gomi}.
The $\Z$ on the left corresponds to the surrounding three-dimensional class AIII phase. The $\Z_2$ on the right records the parity of the zero modes left at the inversion center after cancellation by lower-dimensional cells.
For one zero mode at the center, the cancellation condition is
\begin{align}
 (d^3_{3,-3})^{-1}(1)=2\Z+1.
 \label{eq:3d-BDI-SPT-LSM-preimage}
\end{align}
An odd winding number is therefore required to gap the central zero mode by coupling it to the mode supplied by the surrounding three-dimensional phase.
Figure~\ref{fig:3d-BDI-SPT-LSM}[b] shows the periodic orbital arrangement, with the local parity condition applying at each inversion center. The AHSS in Eq.~\eqref{eq:3d-BDI-SPT-LSM-differential} concerns a single inversion center; it does not identify the $E^3$ groups of the full periodic system.

\section{\texorpdfstring{\NoCaseChange{Four dimensions}}{Four dimensions}}
\label{sec:c2f}
In this section, we study the parity of ${\rm ch}_2[H]$.
We assume that $H_\bk$ and the symmetry matrices are smooth and periodic on $T^4$, and that $H_\bk$ has a gap at zero energy. 
We denote the occupied bundle by $E$.
The fixed Bloch basis trivializes the full state bundle.
We use the orientation $dk_1\wedge dk_2\wedge dk_3\wedge dk_4$ for integration.

For classes C and D, we reduce the problem to a three-dimensional winding number and apply the symmetry invariants introduced above.
For classes AI and AII, evaluating ${\rm ch}_2[H]$ directly produces a Chern--Simons term ${\rm CS}_3$, whose dependence on the spin structure complicates the proof.
We therefore first compute the second Chern class $\langle c_2(E),[T^4]\rangle$ and then recover ${\rm ch}_2$ by including the contribution from the first Chern class.
The spin-structure dependence of the Chern--Simons term is discussed in Appendix~\ref{sec:c2f-spin-note}.

\subsection{Common preliminaries}
\label{sec:c2f-common}
\subsubsection{Chern character and second Chern class}
For a local orthonormal frame $\Phi$ of the occupied states and the projector $P=\Phi\Phi^\dag$, define
\begin{align}
    A=\Phi^\dag d\Phi,\qquad F(A)=dA+A^2.
\end{align}
The connection $A$ is anti-Hermitian. We use the following normalization and define the four-form $q_4(A)$ representing $c_2$:
\begin{align}
{\rm ch}_2[H]&=-\frac1{8\pi^2}\int_{T^4}\tr[F(A)^2],
    \label{eq:c2f-character-normalization}\\
    \begin{split}
    \langle c_2(E),[T^4]\rangle&=\int_{T^4}q_4(A),\\
    q_4(A)&=\frac1{8\pi^2}\{\tr[F(A)^2]-(\tr F(A))^2\}.
    \end{split}
    \label{eq:c2f-normalization}
\end{align}
Either quantity can be evaluated with any $U(N)$ connection on the occupied bundle.
The two quantities are related by
\begin{align}
    {\rm ch}_2[H]=\frac12\langle c_1(E)^2,[T^4]\rangle-\langle c_2(E),[T^4]\rangle.
    \label{eq:c2f-c1-rel}
\end{align}
For a matrix-valued $p$-form $X$ and a matrix-valued $q$-form $Y$, recall that
\begin{align}
    \tr(XY)=(-1)^{pq}\tr(YX),\qquad
    (XY)^\top=(-1)^{pq}Y^\top X^\top.
    \label{eq:c2f-graded}
\end{align}

\subsubsection{Time-reversal-symmetric connection}
\label{sec:c2f-tr-connection}
Let $-\tau$ denote the momentum action in Eq.~\eqref{eq:summary-momentum-action}, and write the class-AI and class-AII symmetries as $T_\bk H_\bk^*T_\bk^\dag=H_{-\tau\bk}$ and $T_\bk^\top=\pm T_{-\tau\bk}$. The occupied bundle satisfies
\begin{align}
    (-\tau)^*E\simeq\overline E,\qquad (-\tau)^*c_1(E)=-c_1(E).
    \label{eq:c2f-c1-symmetry}
\end{align}
In three or fewer dimensions, a complex bundle with $c_1=0$ is trivial. Whether the bundle admits a global frame on a half-domain $T^3\times I$ or $T^2\times I$ can therefore be determined from $c_1$ on $T^3$ or on the subspace $T^2$, respectively.

Using local frames $\Phi_i$, set
\begin{align}
    a_i=\Phi_i^\dag T^\top dT^*\Phi_i,\qquad
    B_i=A_i+\frac12a_i.
    \label{eq:c2f-tr-average}
\end{align}
The frame change $\Phi_j=\Phi_i g_{ij}$ gives $a_j=g_{ij}^\dag a_i g_{ij}$, so $B$ also defines a connection on the occupied bundle. Choosing $\Psi_{-\tau\bk}=T_\bk\Phi_\bk^*$ on the reflected patch yields
\begin{align}
\begin{gathered}
    (-\tau)^*A^\Psi=(A^\Phi+a^\Phi)^*,\qquad
    (-\tau)^*a^\Psi=-(a^\Phi)^*,\\
    (-\tau)^*B^\Psi=(B^\Phi)^*.
    \end{gathered}
    \label{eq:c2f-antiunitary}
\end{align}
Here, $(-\tau)^*$ includes the pullback of differential forms. The connection $B$ is thus preserved by time reversal, with $(-\tau)^*q_4(B)=q_4(B)$. The existence of equivariant connections and their construction by averaging were introduced in Ref.~\cite{DeNittisGomi2016}.

\subsubsection{Boundary formula for \texorpdfstring{$c_2$}{c2}}
Define the local three-form associated with $c_2$ by
\begin{align}
\begin{gathered}
    q_3(A)=\frac1{8\pi^2}\left\{\tr\left(A\,dA+\frac23A^3\right)
                              -\tr A\,d\tr A\right\},\\
    dq_3(A)=q_4(A).
    \end{gathered}
    \label{eq:c2f-local-CS}
\end{align}
We have $q_4(A^*)=q_4(A)$ and $q_3(A^*)=q_3(A)$. Extend the bundle and its connection from a closed oriented three-manifold $Y$ to an oriented four-manifold $Z$, and define
\begin{align}
    \mathcal Q_3[A]:=\int_Zq_4(\widetilde A)\pmod1,\qquad \partial Z=Y.
    \label{eq:c2f-CS}
\end{align}
\footnote{An extension exists because $\Omega^{SO}_3(BU(n))=0$. The values obtained from two extensions differ by the $c_2$ number of a closed four-manifold, which is an integer. No spin structure is needed.}
If a global frame exists, then $\mathcal Q_3[A]\equiv\int_Yq_3(A)\pmod1$.

With $\omega_3$ defined in Eq.~\eqref{eq:coh-w3}, the gauge transformation $A^g=g^{-1}Ag+g^{-1}dg$ gives
\begin{align}
    \alpha_2(A,g)&=\frac1{8\pi^2}
        \{\tr(A\,dg\,g^{-1})-\tr A\,d\log\det g\},\notag\\
    q_3(A^g)-q_3(A)&=-\omega_3(g)+d\alpha_2(A,g).
    \label{eq:c2f-gauge}
\end{align}
The one-form $d\log\det g$ is independent of the branch chosen for the logarithm.

Let $X_+=T^3\times[0,\pi]_{k_4}$, and orient both boundaries $Y_b=\{k_4=b\}$ ($b=0,\pi$) by $dk_1\wedge dk_2\wedge dk_3$. Since $\partial X_+=Y_0-Y_\pi$, we have
\begin{align}
    \int_{X_+}q_4(A)
       \equiv \mathcal Q_3[A_0]-\mathcal Q_3[A_\pi]\pmod1.
    \label{eq:c2f-boundary}
\end{align}

\subsubsection{Boundary quantity in two patches}
Divide $T^3$ into $V_N=T^2\times[0,\pi]_{k_3}$ and $V_S=T^2\times[-\pi,0]_{k_3}$.
Assume that the bundle admits a global frame on each patch $V_N,V_S$.
On the gluing surfaces $S_z=T^2 \times \{k_3=z \in \{0,\pi\}\}$, write the relation between the frames as $\Phi^S=\Phi^Ng_z$. Orient $S_z$ by $dk_1\wedge dk_2$.

Applying the gauge-transformation formula \eqref{eq:c2f-gauge} on each patch gives the decomposition
\begin{align}
    \mathcal Q_3[A]
       &\equiv\int_{V_N}q_3(A^N)+\int_{V_S}q_3(A^S)\notag\\
       &\quad+\left[\int_{S_z}\alpha_2(A^N,g_z)-{\rm WZ}[g_z]\right]_{z=0}^{\pi}
                    \pmod1.
    \label{eq:c2f-gluing}
\end{align}
Here, we use the notation $[u_z]_{z=0}^{\pi}=u_\pi-u_0$.
The signs of the two boundary contributions follow from $\partial V_N=S_\pi-S_0$.

\subsubsection{Particle-hole symmetry and reduction to three dimensions}
\label{sec:c2f-ph-transport}
For classes C and D, we express the four-dimensional Chern character as the winding number of a clutching matrix. The construction combines parallel transport of the occupied projector with gluing to its PHS-related inverse image~\cite{MonacoPeluso2023}. Let $y=(k_1,k_2,k_3)\in Y = T^3$ and $t=k_4$. Write the occupied projector as $P_t(y)$ and set $Q_0=1-P_0$. The symmetry conditions are
\begin{align}
    C_t(y)P_t(y)^*C_t(y)^\dag&=1-P_{-t}(cy),\notag\\
    C_t(y)^\top&=\varepsilon C_{-t}(cy).
    \label{eq:c2f-ph-projection-symmetry}
\end{align}
For class C, $cy=(k_1,-k_2,-k_3)$ and $\varepsilon=-1$; for class D, $c=\mathrm{id}$ and $\varepsilon=1$. If the full state space has dimension $2N$, the occupied and unoccupied subspaces each have dimension $N$.

On the positive half-domain, define parallel transport by solving~\cite{Kato1950}
\begin{align}
    K_t=[\partial_tP_t,P_t],\qquad
    \partial_tU_t^+=K_tU_t^+,\qquad U_0^+=1.
    \label{eq:c2f-ph-kato}
\end{align}
Because $P_t(y)$ is globally smooth in the fixed Bloch basis, uniqueness and smooth parameter dependence of solutions to the ordinary differential equation ensure that $U_t^+(y)$ is smooth and periodic on $Y$. The identities $K_t^\dag=-K_t$ and $[K_t,P_t]=\partial_tP_t$ then give
\begin{align}
    (U_t^+)^\dag U_t^+=1,\qquad P_t=U_t^+P_0(U_t^+)^\dag.
    \label{eq:c2f-ph-intertwining}
\end{align}
The construction requires no choice of frame for the occupied bundle and does not replace $P_0(y)$ with a constant projector.\footnote{Restricted to the occupied bundle, the parallel transport is the Wilson line of the Berry connection. In local frames, $w_t=\Phi_t^\dag U_t^+\Phi_0$ satisfies $\partial_tw_t=-(\Phi_t^\dag\partial_t\Phi_t)w_t$.}
Using PHS, we may choose the transport on the negative half-domain as
\begin{align}
\begin{gathered}
    U_t^-(y)=C_t(cy)U_t^+(cy)^*C_0(cy)^\dag,\\
    P_{-t}=U_t^-P_0(U_t^-)^\dag,\qquad U_0^-=1.
    \end{gathered}
    \label{eq:c2f-ph-negative}
\end{align}
Since $P_\pi=P_{-\pi}$, the matrices
\begin{align}
\begin{aligned}
    h&=U_\pi^+P_0+U_\pi^-Q_0,\\
    g&=(U_\pi^-)^\dag h\\
     &=(U_\pi^-)^\dag U_\pi^+P_0+Q_0
    \end{aligned}
    \label{eq:c2f-ph-hg}
\end{align}
are both unitary on $Y$. The maps $U_\pi^\pm$ send the initial occupied subspace to the same final occupied subspace, and likewise for the unoccupied subspace. Since $U_t^-$ is a homotopy from the identity,
\begin{align}
    W_3[h]=W_3[g].
    \label{eq:c2f-ph-winding}
\end{align}
Moreover, $C_\pi(y)U_\pi^\pm(y)^*=U_\pi^\mp(cy)C_0(y)$ and the exchange of the occupied and unoccupied subspaces by $C_0$ give
\begin{align}
    C_\pi(y)h(y)^*=h(cy)C_0(y).
    \label{eq:c2f-ph-sewing}
\end{align}
The resulting relation is the symmetry condition in one lower dimension.

Finally, we relate ${\rm ch}_2[H]$ to $W_3[h]$.
Extending the holonomy by the identity on the orthogonal complement gives a global unitary matrix~\cite{TradlerWilsonZeinalian2016}.
Even when the occupied bundle is nontrivial, the augmented system $H(y,t)\oplus[-H(y,0)]$ has a continuous global frame on the cut-open cylinder:
\begin{align}
    \begin{aligned}
    \widetilde\Phi_{\pm t}
       &=\begin{pmatrix}U_t^\pm P_0\\Q_0\end{pmatrix},\\
    \widetilde\Phi_{\pm t}^\dag\widetilde\Phi_{\pm t}&=1_{2N},\\
    \widetilde\Phi_\pi&=\widetilde\Phi_{-\pi}g.
    \end{aligned}
    \label{eq:c2f-ph-augmented-frame}
\end{align}
The frame can be smoothed near $t=0$ while keeping both ends fixed. The added subsystem depends only on $Y$ and contributes zero to ${\rm ch}_2$. Gluing the two ends therefore gives
\begin{align}
    {\rm ch}_2[H]=-W_3[g]=-W_3[h].
    \label{eq:c2f-ph-clutching-character}
\end{align}
It remains to evaluate the three-dimensional winding number $W_3[h]$ subject to the symmetry~\eqref{eq:c2f-ph-sewing}.

\subsection{\texorpdfstring{$(d,d_\tau)=(4,0)$}{(d,d\_tau)=(4,0)}, class AI}
\label{sec:c2f-ai40}
Let $-\tau\bk=-\bk$, and assume
\begin{align}
    T_\bk H_\bk^*T_\bk^\dag=H_{-\tau\bk},\qquad
    T_\bk^\top=T_{-\tau\bk}.
    \label{eq:c2f-ai40-symmetry}
\end{align}
We have $(-\tau)^*c_1(E)=-c_1(E)$, whereas $(-\tau)^*$ acts as the identity on $H^2(T^4;\Z)$. This cohomology group is torsion-free, so $c_1(E)=0$ and hence
\begin{align}
    {\rm ch}_2[H]=-\langle c_2(E),[T^4]\rangle.
    \label{eq:c2f-ai40-character}
\end{align}
It therefore suffices to show that $\langle c_2(E),[T^4]\rangle$ is even.

Since $c_1(E)=0$, we can choose a global frame $\Phi$ of occupied states on the half-domain $X_+=T^3\times[0,\pi]_{k_4}$. The symmetric connection $B=A+\frac12a$ of Sec.~\ref{sec:c2f-tr-connection} satisfies $(-\tau)^*q_4(B)=q_4(B)$. As $-\tau$ preserves the four-dimensional orientation, Stokes' theorem gives
\begin{align}
    \langle c_2(E),[T^4]\rangle
       =2\int_{X_+}q_4(B)
       =2\int_Y\{q_3(B_0)-q_3(B_\pi)\}.
    \label{eq:c2f-ai40-boundary}
\end{align}
We use the restrictions of the same frame $\Phi$ at both ends, so the boundary formula holds as an equality of real numbers.

Let $\iota\bk=-\bk$ denote inversion on the boundary $Y=T^3$, and define the sewing matrices at $s=0,\pi$ by
\begin{align}
    T_{\bk,s}\Phi_{\bk,s}^*=\Phi_{-\bk,s}g_s(\bk),\qquad
    g_s(\bk)^\top=g_s(-\bk).
    \label{eq:c2f-ai40-sewing}
\end{align}
The symmetric connection satisfies $B_s^*=(\iota^*B_s)^{g_s}$. The total derivative in Eq.~\eqref{eq:c2f-gauge} integrates to zero. Since $\iota$ reverses orientation,
\begin{align}
    \int_Yq_3(B_s)=-\int_Yq_3(B_s)-W_3[g_s].
\end{align}
Since $g_s$ satisfies the class-CI condition in Sec.~\ref{sec:3d-CI}, we have $W_3[g_s]\in2\Z$. Thus,
\begin{align}
    {\rm ch}_2[H]=-\langle c_2(E),[T^4]\rangle
       =W_3[g_0]-W_3[g_\pi]\in2\Z.
    \label{eq:c2f-ai40-final}
\end{align}

Thus, an odd Chern character ${\rm ch}_2[H]$ is impossible in the finite-dimensional framework considered here, even with a momentum-dependent time-reversal symmetry matrix.
The relation between this constraint and SPT--LSM theorems remains open: a formulation of SPT--LSM theorems for internal symmetries in terms of the real-space AHSS has yet to be developed.

\subsection{\texorpdfstring{$(d,d_\tau)=(4,1)$}{(d,d\_tau)=(4,1)}, class C}
\label{sec:c2f-c41}
Let $-\tau\bk=(k_1,-k_2,-k_3,-k_4)$, and assume
\begin{align}
    C_\bk H_\bk^*C_\bk^\dag=-H_{-\tau\bk},\qquad
    C_\bk^\top=-C_{-\tau\bk}.
    \label{eq:c2f-c41-symmetry}
\end{align}
Set $C_s=C|_{k_4=s}$ ($s=0,\pi$). The invariant $\mu$ defined in Eq.~\eqref{eq:alt62-nu} determines the parity through
\begin{align}
    {\rm ch}_2[H]\equiv\mu[C_0]-\mu[C_\pi]\pmod2.
    \label{eq:c2f-c41-result}
\end{align}

To derive the parity relation, take $c(k_1,k_2,k_3)=(k_1,-k_2,-k_3)$. The matrix $h$ constructed in Sec.~\ref{sec:c2f-ph-transport} satisfies
\begin{align}
    C_\pi h^*C_0^\dag=h\circ c,\qquad
    C_b^\top=-C_b\circ c.
    \label{eq:c2f-c41-CII-sewing}
\end{align}
These relations are the three-dimensional class-CII condition \eqref{eq:alt62-symmetry} with $(q,u,v)=(h,C_\pi,C_0)$. The three-dimensional parity formula therefore gives $W_3[h]\equiv\mu[C_\pi]-\mu[C_0]\pmod2$. Combining this congruence with ${\rm ch}_2[H]=-W_3[h]$ in Eq.~\eqref{eq:c2f-ph-clutching-character} proves the claimed parity relation.

\subsubsection{Locality and examples}
\label{sec:c2f-c41-locality}
The restrictions of a finite-range $C_\bk$ to $k_4=0,\pi$ are finite-range matrices $C_0,C_\pi$ satisfying the three-dimensional class-CII condition. Theorem~\ref{thm:CII-finite-range} then gives $\mu[C_0]=\mu[C_\pi]=0$, and Eq.~\eqref{eq:c2f-c41-result} implies that ${\rm ch}_2[H]$ is even.
Thus, finite-range symmetry operators cannot realize an odd Chern character.

As in Sec.~\ref{sec:2dclassC-model}, spectral flattening of the Hamiltonian provides a canonical construction of examples.
If a gapped Hamiltonian $H_\bk$ satisfies $H_\bk^\top = - H_{-\tau \bk}$, define its flattening by $Q_\bk = H_\bk (H_\bk^\dag H_\bk)^{-1/2}$ and set $C_\bk = Q_{-\tau\bk}$.
For example, choose the gamma matrices
\begin{align}
    \Gamma_0&=\sigma_x\otimes\sigma_y,&
    \Gamma_1&=\sigma_y\otimes1_2,\nonumber\\
    \Gamma_2&=\sigma_x\otimes\sigma_x,&
    \Gamma_3&=\sigma_x\otimes\sigma_z,\nonumber\\
    \Gamma_4&=\sigma_z\otimes1_2.
    \label{eq:c2f-4d-model-gamma}
\end{align}
The corresponding standard Dirac model is
\begin{align}
    H_\bk&=\sum_{a=0}^4d_a(\bk)\Gamma_a,\nonumber\\
    d_0(\bk)&=-3+\sum_{j=1}^4\cos k_j,\nonumber\\
    d_j(\bk)&=\sin k_j\quad(j=1,\ldots,4).
    \label{eq:c2f-4d-model-H}
\end{align}

\subsubsection{SPT--LSM constraints}
\label{sec:4d-C-SPT-LSM}
We now compare the parity and locality results with the SPT--LSM picture of anomaly cancellation on a mirror plane.
We do not impose lattice translations here. Consider the real-space reflection $\tau:(x_1,x_2,x_3,x_4)\mapsto(-x_1,x_2,x_3,x_4)$.
Place four-dimensional Chern insulators with the same ${\rm ch}_2=n$ on the two sides of the mirror plane. The boundaries of the two disconnected regions carry $|n|$ pairs of oppositely propagating three-dimensional Weyl states.
On the mirror plane, the spatial action is trivial, and each pair forms a three-dimensional class-C Dirac state with PHS squaring to $-1$.
A single pair admits no PHS-preserving mass term, whereas two pairs can be gapped.\footnote{Let $s_i$ be Pauli matrices distinguishing the Weyl partners. A representative boundary Hamiltonian and PHS are
$h_\partial=s_z\otimes(p_2\sigma_x+p_3\sigma_y+p_4\sigma_z)$ and $\mathcal C_\partial=(s_x\otimes\sigma_y)K$.
The Hamiltonian and PHS satisfy $\mathcal C_\partial^2=-1$ and $\mathcal C_\partial h_\partial(p)\mathcal C_\partial^{-1}=-h_\partial(-p)$.
The possible mass matrices $s_x\otimes1_2,s_y\otimes1_2$ are both forbidden because they do not change sign under $\mathcal C_\partial$.
For two pairs, $s_x\otimes1_2\otimes\rho_y$ is an allowed mass, where $\rho_i$ are Pauli matrices distinguishing the copies.}
The obstruction to gapping a single pair is the same anomaly as that of the boundary of the four-dimensional class-C $\Z_2$ phase with cellwise PHS~\cite{RyuSchnyderFurusakiLudwig2010}.

If the mirror plane initially supports a Dirac state with this nontrivial $\Z_2$ anomaly, the condition for gapping that state together with the $|n|$ pairs supplied by the surrounding bulk is $n+1\equiv0\pmod2$.
A symmetry-preserving gapped phase therefore requires an odd ${\rm ch}_2$.
The single Dirac state initially placed on the plane cannot be realized as a standalone three-dimensional lattice system with finitely many local orbitals and cellwise PHS. It must be supplied, for example, as the boundary of a higher-dimensional phase.
This boundary construction assumes an anomalous degree of freedom as input, so its assumptions differ from those of the finite-range no-go result in Sec.~\ref{sec:c2f-c41-locality}.

In our framework, a symmetry action satisfying $\mu[C_0]-\mu[C_\pi]=1$ in Eq.~\eqref{eq:c2f-c41-result} enforces an odd ${\rm ch}_2$ in a finite-dimensional Bloch model.
The model with $C_\bk=Q_{-\tau\bk}$ realizes odd parity through a symmetry action that is quasilocal but not of finite range.
The model's symmetry action differs from a crystalline symmetry acting only on the degrees of freedom at each point of the fixed plane.
The SPT--LSM picture of anomaly cancellation is therefore consistent both with the prohibition of odd values for finite-range actions and with their realization by exponentially decaying symmetry actions.
This comparison establishes consistency without identifying $\mu$ with an anomaly index in the real-space AHSS.

\subsection{\texorpdfstring{$(d,d_\tau)=(4,2)$}{(d,d\_tau)=(4,2)}, class AII}
\label{sec:c2f-aii42}
Let $-\tau\bk=(k_1,k_2,-k_3,-k_4)$, and assume
\begin{align}
    T_\bk H_\bk^*T_\bk^\dag=H_{-\tau\bk},\qquad
    T_\bk^\top=-T_{-\tau\bk}.
    \label{eq:c2f-aii42-symmetry}
\end{align}
As shown below, ${\rm ch}_2[H]$ is even.
Thus, an odd Chern character is impossible in the finite-dimensional framework, even with a momentum-dependent antiunitary symmetry matrix.

Since $c_1(E)$ need not vanish, we examine $\langle c_2(E),[T^4]\rangle$ and the contribution from the first Chern class separately.

\subsubsection{Evaluation of \texorpdfstring{$c_2$}{c2}}
We use the symmetric connection $B=A+\frac12a$ of Sec.~\ref{sec:c2f-tr-connection}. From $(-\tau)^*q_4(B)=q_4(B)$ and the boundary formula \eqref{eq:c2f-boundary}, we obtain
\begin{align}
    \begin{aligned}
    &\langle c_2(E),[T^4]\rangle=2\int_{X_+}q_4(B)\\
    &\quad\equiv2\{\mathcal Q_3[B_0]-\mathcal Q_3[B_\pi]\}\pmod2.
    \end{aligned}
    \label{eq:c2f-aii42-boundary}
\end{align}
We now prove that $\mathcal Q_3[B_s]\equiv0\pmod1$ at each boundary.

Suppress the index $s$ and write the boundary reflection as $m(k_1,k_2,k_3)=(k_1,k_2,-k_3)$. The relation $(-\tau)^*c_1=-c_1$ forces the first Chern number on the $k_1k_2$ plane to vanish, so we can choose a frame $\Phi$ on $V_N=T^2\times[0,\pi]_{k_3}$. On $V_S=T^2\times[-\pi,0]_{k_3}$, set $\Psi_{m\bk}=T_\bk\Phi_\bk^*$. The relation $m^*B^\Psi=(B^\Phi)^*$ and the orientation reversal by $m$ make the local three-form integrals cancel:
\begin{align}
    \int_{V_N}q_3(B^\Phi)+\int_{V_S}q_3(B^\Psi)=0.
    \label{eq:c2f-aii42-local-zero}
\end{align}

On the remaining gluing surfaces $S_z=\{k_3=z\}$ ($z=0,\pi$), the transition function $g=\Phi^\dag\Psi$ satisfies
\begin{align}
    g^\top=-g,\qquad B^g=B^*=-B^\top.
\end{align}
For brevity, write $B=B^\Phi|_{S_z}$. The preceding symmetry relations imply
\begin{align}
    dg\,g^{-1}&=-B-gB^\top g^{-1},\notag\\
    \tr(B\,dg\,g^{-1})&=0,\qquad
    2\tr B=-d\log\det g.
\end{align}
The first equality in the second line uses $\tr(B^2)=0$ and $\tr(BgB^\top g^{-1})=-\tr(BgB^\top g^{-1})=0$; the latter follows from $g^\top=-g$. The two-form in Eq.~\eqref{eq:c2f-gauge} therefore also vanishes:
\begin{align}
    \alpha_2(B,g)
       =\frac1{8\pi^2}\{\tr(B\,dg\,g^{-1})
                            -\tr B\,d\log\det g\}=0.
\end{align}
Furthermore, Lemma~\ref{lem:alt63-skew-WZ} for antisymmetric unitary matrices gives ${\rm WZ}[g]\equiv0\pmod1$.
Every term in the gluing formula \eqref{eq:c2f-gluing} therefore vanishes, yielding
\begin{align}
    \begin{gathered}
    \mathcal Q_3[B_s]\equiv0\pmod1\quad(s=0,\pi),\\
    \langle c_2(E),[T^4]\rangle\in2\Z.
    \end{gathered}
    \label{eq:c2f-aii42-c2-even}
\end{align}

\subsubsection{Contribution from the first Chern class}
Write $e_j=[dk_j/(2\pi)]$. From $(-\tau)^*c_1=-c_1$, we have
\begin{align}
    c_1(E)=n_{13}e_1e_3+n_{14}e_1e_4
                  +n_{23}e_2e_3+n_{24}e_2e_4.
    \label{eq:c2f-aii42-c1}
\end{align}
Each coefficient $n_{ij}$ is even. To see this, restrict to $T^2_{ij}$ and fix the other two coordinates at $0$. The remaining symmetry is two-dimensional class AII with one reversed coordinate, so Eq.~\eqref{eq:AII-mirror-even} applies. It follows that
\begin{align}
    \frac12\langle c_1(E)^2,[T^4]\rangle
       =n_{14}n_{23}-n_{13}n_{24}\in4\Z.
    \label{eq:c2f-aii42-square}
\end{align}
Combining the contribution from the first Chern class with $\langle c_2(E),[T^4]\rangle\in2\Z$ gives
\begin{align}
    {\rm ch}_2[H]
       =\frac12\langle c_1(E)^2,[T^4]\rangle-\langle c_2(E),[T^4]\rangle\in2\Z.
    \label{eq:c2f-aii42-final}
\end{align}

\subsubsection{\texorpdfstring{Gapless states bound to a $\pi$ flux and SPT--LSM constraints}{Gapless states bound to a pi flux and SPT--LSM constraints}}
\label{sec:4d-AII-SPT-LSM}

The even-parity constraint in Eq.~\eqref{eq:c2f-aii42-final} is consistent with the real-space AHSS picture. For simplicity, we do not impose lattice translations in the following construction.
Let $C_2:(x_1,x_2,x_3,x_4)\mapsto(-x_1,-x_2,x_3,x_4)$ be the real-space twofold rotation, and insert a $\pi$ flux along its fixed plane $x_1=x_2=0$.

First, consider a model of a four-dimensional Chern insulator without flux, with antiunitary twofold rotation symmetry $\mathcal T_0=C_2T$ satisfying $\mathcal T_0^2=1$.
Let $\phi$ be the polar angle in the normal plane. Although $\mathcal T_0$ maps the flux $\pi$ to $-\pi$, the gauge transformation $G(\phi)=e^{i\phi}$ restores the original flux background outside the core.
The symmetry preserving the flux background is therefore $\mathcal T_\pi=G\mathcal T_0$. Since $C_2$ maps $\phi$ to $\phi+\pi$, we have
\begin{align}
 \mathcal T_\pi^2
 =e^{i\phi}e^{-i(\phi+\pi)}\mathcal T_0^2=-1.
 \label{eq:4d-AII-flux-square}
\end{align}
The spatial action is trivial in the effective theory on the flux plane, where $\mathcal T_\pi$ therefore acts as class-AII time reversal.

To obtain the flux-bound states, we add the kinetic term in the fourth direction to the cylindrical-surface calculation for a three-dimensional topological insulator~\cite{RosenbergGuoFranz2010}.\footnote{For a general construction of anomalous states on a $\pi$ flux by dimensional extension, see Ref.~\cite{GeierFulgaLau2021}.}
Specifically, take a representative model with $n:={\rm ch}_2[H_{\rm bulk}]=1$ and regularize the flux core as a cylindrical hole of radius $R$, as shown in Fig.~\ref{fig:4d-AII-SPT-LSM}[a].
The inner wall $S^1_R\times\R^2_{x_3,x_4}$ carries a single three-dimensional Weyl state~\cite{QiHughesZhang2008}.
Take $R$ much larger than the bulk correlation length and use the low-energy theory on the inner wall. Setting $\hbar=e=1$ and writing the flux as $\Phi=2\pi\alpha$, we can choose the basis and coordinate orientation so that
\begin{align}
 h_\ell(k_3,k_4)
 &=v\left[k_3\sigma_x+k_4\sigma_y
       +\frac{\ell+\frac12-\alpha}{R}\sigma_z\right],
 \qquad \ell\in\Z.
 \label{eq:4d-AII-flux-cylinder}
\end{align}
At $k_4=0$, the Hamiltonian reduces to the cylindrical-surface Hamiltonian of Rosenberg et al.; the additional term $v k_4\sigma_y$ anticommutes with the other terms. The shift by $1/2$ comes from the spinor Berry phase $\pi$ around the cylinder.
For a $\pi$ flux, namely $\alpha=1/2$, the angular gap of the $\ell=0$ mode vanishes, giving a two-dimensional gapless Weyl state described by
\begin{align}
 h_0(k_3,k_4)&=v(k_3\sigma_x+k_4\sigma_y).
 \label{eq:4d-AII-flux-Dirac}
\end{align}
The gapless state is localized on the inner wall and propagates in the $x_3,x_4$ directions.
On the lowest mode, we may take $\mathcal T_\pi=i\sigma_yK$, which forbids the mass term $m\sigma_z$.

To obtain an SPT--LSM constraint, consider a single gapless Weyl state initially present on the two-dimensional fixed plane of $C_2T$.
As shown in Fig.~\ref{fig:4d-AII-SPT-LSM}[a], coupling the initially present state to the Weyl state bound to the $\pi$ flux can open a gap while preserving time-reversal symmetry.
Stacking bulk phases yields general $n$, for which the parity of the number of Weyl states on the flux plane is $n\bmod2$. Canceling the single state on the fixed plane and opening a gap therefore requires an odd ${\rm ch}_2[H_{\rm bulk}]$.
In the real-space AHSS interpretation of Ref.~\cite{Shiozaki-Xiong-Gomi}, this anomaly cancellation is described by a nontrivial second differential $d^2_{4,-4}:E^2_{4,-4}\to E^2_{2,-3}$. An invertible phase on a four-cell cancels the anomaly on the two-dimensional fixed plane.

The single class-AII Weyl state initially placed on the fixed plane has the same anomaly as the surface of a three-dimensional topological insulator and cannot be realized in a standalone two-dimensional lattice system preserving cellwise time-reversal symmetry.
A realization supplied by a higher-dimensional bulk falls outside the setting of finite-dimensional Bloch Hamiltonians and symmetry matrices considered here. The SPT--LSM odd-parity constraint therefore does not contradict the even-parity constraint in Eq.~\eqref{eq:c2f-aii42-final}.

\begin{figure*}[t]
\centering
\includegraphics[width=0.7\textwidth]{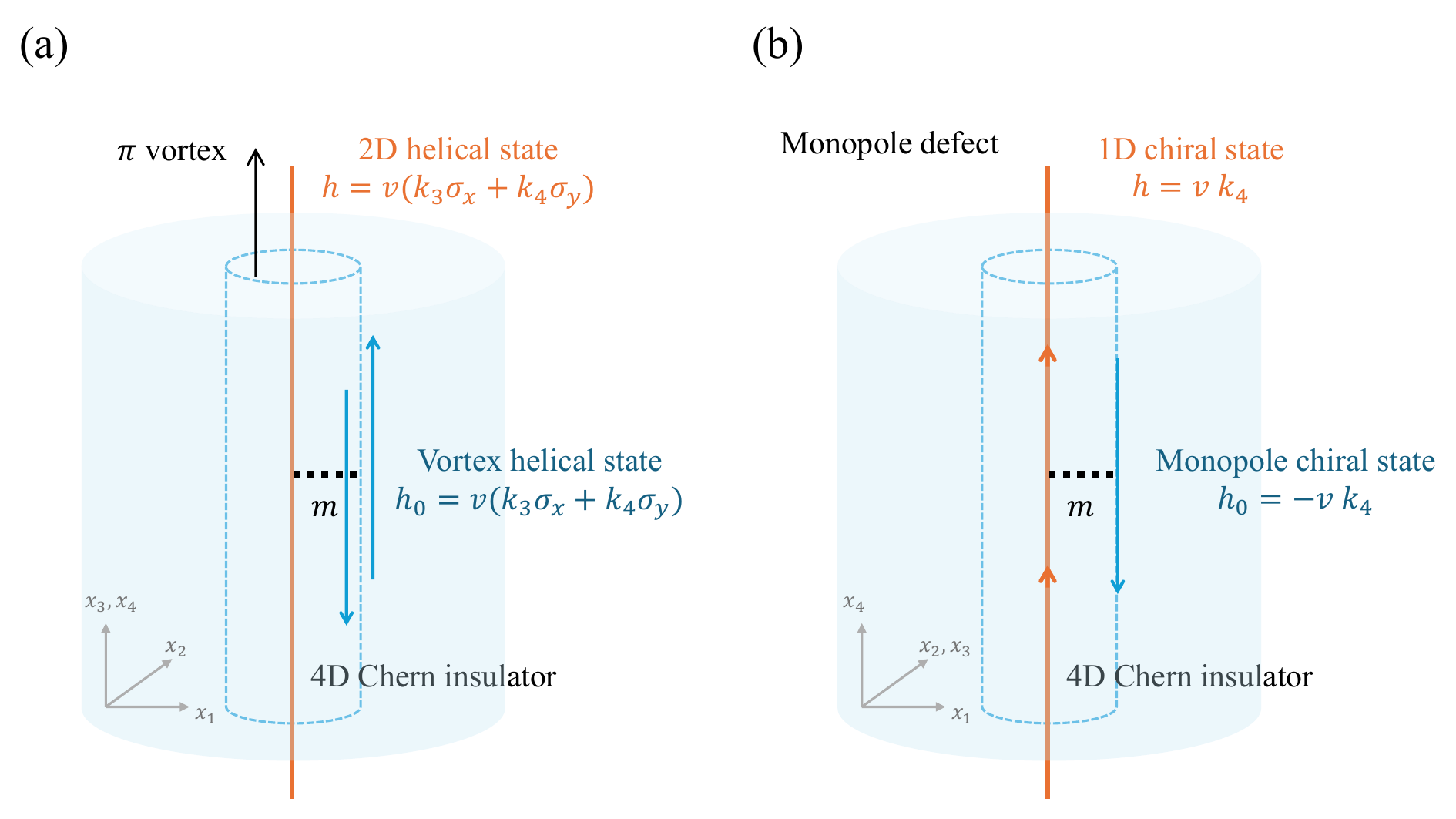}
\caption{Real-space SPT--LSM mechanisms in four dimensions.
The shaded regions represent four-dimensional Chern insulators with $|{\rm ch}_2|=1$.
[a] Class AII with $d_\tau=2$.
A $\pi$ flux pierces the $x_1x_2$ plane and extends along $x_3,x_4$, represented schematically by a single vertical axis.
The blue dashed cylinder marks the regularized flux core. The blue arrows indicate the bound two-dimensional helical Weyl state of Eq.~\eqref{eq:4d-AII-flux-Dirac}.
The orange line denotes the original helical state on the rotation-fixed plane.
The coupling $m$ gaps the two states while preserving time reversal.
[b] Class AI with $d_\tau=4$.
A unit monopole line extends along $x_4$, with transverse coordinates $x_1,x_2,x_3$; the dashed tube schematically marks its core.
The chiral channel $h=vk_4$ (orange) and the monopole-bound channel $h_0=-vk_4$ (blue) propagate in opposite directions.
The coupling $m$ gaps the two channels while preserving $PT$ with $(PT)^2=1$.
}
\label{fig:4d-AII-SPT-LSM}
\end{figure*}

\subsection{\texorpdfstring{$(d,d_\tau)=(4,3)$}{(d,d\_tau)=(4,3)}, class D}
\label{sec:c2f-d43}
We impose the symmetry
\begin{align}
\begin{gathered}
    -\tau\bk=(k_1,k_2,k_3,-k_4),\\
    C_\bk H_\bk^*C_\bk^\dag=-H_{-\tau\bk},\qquad
    C_\bk^\top=C_{-\tau\bk}.
    \end{gathered}
    \label{eq:c2f-d43-symmetry}
\end{align}
On the boundaries $Y_s=\{k_4=s\}$ ($s=0,\pi$), the matrices satisfy $C_s^\top=C_s$. We can therefore use the class-CI $\Z_2$ invariant $\zeta[C_s;\mathfrak s]$ introduced in Eq.~\eqref{eq:alt64-nu-definition_1}. Identify both boundaries with $Y=T^3$ and choose a common spin structure $\mathfrak s$. The parity is then determined entirely by the two boundary invariants:
\begin{align}
    \quad {\rm ch}_2[H]
       \equiv\zeta[C_0;\mathfrak s]-\zeta[C_\pi;\mathfrak s]\pmod2.
    \label{eq:c2f-d43-final}
\end{align}
We now prove the parity relation and show that the difference of boundary invariants is independent of the spin structure.

Take the matrix $h$ constructed by parallel transport in Sec.~\ref{sec:c2f-ph-transport}. Equations~\eqref{eq:c2f-ph-sewing} and \eqref{eq:c2f-ph-clutching-character} give
\begin{align}
    C_\pi&=hC_0h^\top,
    \label{eq:c2f-d43-h}\\
    {\rm ch}_2[H]&=-W_3[h].
    \label{eq:c2f-d43-winding}
\end{align}
The first relation is the three-dimensional class-BDI symmetry condition with $(q,u,v)=(h,C_\pi,C_0)$. Substituting $X=C_0$ and $g=h$ into the transformation law \eqref{eq:alt64-nu-transformation} for $\zeta$ gives
\begin{align}
    \zeta[C_\pi;\mathfrak s]-\zeta[C_0;\mathfrak s]
       \equiv W_3[h]\pmod2.
    \label{eq:c2f-d43-zeta-transformation}
\end{align}
Combining this congruence with Eq.~\eqref{eq:c2f-d43-winding} proves the parity relation.

The transformation law holds for any common spin structure $\mathfrak s$, while its right-hand side $W_3[h]$ is independent of $\mathfrak s$. Thus, under our standing assumption that a gapped $H$ exists, the difference in Eq.~\eqref{eq:c2f-d43-final} is independent of the spin structure and depends only on $C_0,C_\pi$. The argument requires neither a trivial occupied bundle nor a global Takagi factorization. In particular, if $C_0=C_\pi$, then ${\rm ch}_2[H]$ is even.

\subsubsection{Example and locality}
\label{sec:c2f-d43-locality}
Finite range imposes the following constraint on the symmetry matrix $C_\bk$.
\begin{thm}
    \label{thm:finite_range_d43}
    As throughout this section, assume that a gapped $H$ exists.
    If the unitary matrix $C_\bk$ satisfies the symmetry \eqref{eq:c2f-d43-symmetry} and has finite range, meaning that it can be written for a finite set ${\cal F} \subset\Z^4$ as
    \begin{align}
        C_\bk=\sum_{\bm{n}\in{\cal F}}C_{\bm{n}} e^{-i \bm{n} \cdot \bk},
    \end{align}
    then
    \begin{align}
        \zeta[C_0;\mathfrak s]-\zeta[C_\pi;\mathfrak s] \equiv 0 \quad \pmod2.
        \label{eq:c2f-d43-finite-range}
\end{align}
\end{thm}
The proof is given in Appendix~\ref{app:d4-finite-range}.

An example with ${\rm ch}_2[H]=1$ uses the Hamiltonian in Eq.~\eqref{eq:c2f-4d-model-H} with only the symmetry matrix changed.
Let $Q_\bk$ be the flattening of the Hamiltonian in Eq.~\eqref{eq:c2f-4d-model-H}, and set $C_\bk = Q_{k_1,k_2,k_3,-k_4} (1_2 \otimes \sigma_y)$.

\subsubsection{SPT--LSM constraints}
\label{sec:4d-D-SPT-LSM}
The real-space action is $\tau:(x_1,x_2,x_3,x_4)\mapsto(-x_1,-x_2,-x_3,x_4)$, whose fixed set is the $x_4$ axis.
One can consider an SPT--LSM construction in which an oppositely propagating state bound to a monopole line in a four-dimensional Chern insulator cancels a chiral state on the fixed axis.
Below, we use a particle-number-conserving realization and do not impose lattice translations.

First, the model in Eq.~\eqref{eq:c2f-4d-model-H} with ${\rm ch}_2=1$ satisfies, for $J=1_2\otimes\sigma_y$,
\begin{align}
 JH_\bk^*J^\dag=-H_{-\tau\bk},\qquad JJ^*=-I.
 \label{eq:4d-D-monopole-parent}
\end{align}
Thus, PHS squares to $-1$ in the model without defects.
Insert a unit monopole into the three-dimensional transverse space $(x_1,x_2,x_3)$ and extend it along the $x_4$ direction.
Outside the flux core, applying inversion and complex conjugation twice produces an additional $-1$ from the monopole transition function, as in the patch calculation of Sec.~\ref{sec:3d-BDI-SPT-LSM}.
The PHS that preserves the monopole background therefore squares to $+1$ and acts as class-D PHS on the states bound to the monopole line.

The states bound to the monopole line can also be obtained by dimensional extension of the monopole zero modes in a three-dimensional class-AIII insulator~\cite{AokiFukayaKanKoshinoMatsuki2023}.
Let $H_\perp$ be the transverse Hamiltonian and $\G$ its chiral operator.
For a zero mode satisfying $H_\perp\psi_0=0$ and $\G\psi_0=s\psi_0$ ($s=\pm1$), adding motion in the fourth direction gives, at low energy,
\begin{align}
 (H_\perp+vk_4\G)\psi_0=svk_4\psi_0.
 \label{eq:4d-D-monopole-line}
\end{align}
The zero mode thus becomes a state propagating in one direction along the monopole line.
If a single oppositely propagating chiral state is initially placed on the axis, the two states can be gapped while preserving class-D PHS.
With a suitable choice of basis and velocities, the theory of the two states can be written as
\begin{align}
 \begin{aligned}
 h_{\rm line}(k_4)&=vk_4\sigma_z+m\sigma_y,\\
 h_{\rm line}(k_4)^*&=-h_{\rm line}(-k_4).
 \end{aligned}
 \label{eq:4d-D-monopole-gap}
\end{align}
A gap opens for $m\ne0$.
The gapped pair provides a concrete example of a nontrivial four-dimensional bulk canceling the anomaly on the fixed axis.

However, a standalone unidirectional state cannot be realized in a local one-dimensional lattice system with finitely many bands.
The construction assumes this anomalous degree of freedom and the monopole background as inputs. It therefore does not contradict Theorem~\ref{thm:finite_range_d43}, whose setting consists of periodic finite-dimensional Bloch matrices.
Equation~\eqref{eq:c2f-d43-final} directly shows that a gapped phase compatible with a symmetry action satisfying $\zeta[C_0;\mathfrak s]-\zeta[C_\pi;\mathfrak s]=1$ must have odd ${\rm ch}_2$.
The model above realizes odd parity with the symmetry action $C_\bk=Q_{-\tau\bk}J$. The action is quasilocal but not of finite range, consistent with the finite-range no-go result.
As in class C, we do not identify the model's symmetry action with a crystalline symmetry that simply maps each orbital to its geometric image. A general identification of the difference of $\zeta$ with real-space defect anomalies remains open.

\subsection{\texorpdfstring{$(d,d_\tau)=(4,4)$}{(d,d\_tau)=(4,4)}, class AI}
\label{sec:c2f-ai44}
Assume the momentum-fixing symmetry
\begin{align}
    T_\bk H_\bk^*T_\bk^\dag=H_\bk,\qquad T_\bk^\top=T_\bk.
    \label{eq:c2f-ai44-symmetry}
\end{align}
We show using local frames that the occupied bundle is the complexification of a real bundle.

Let $H$ be an $N\times N$ matrix with $n$ occupied states, and flatten it so that $H^2=1_N$. Cover $T^4$ by sufficiently small patches $O_i$ and take a Takagi factorization on each patch:
\begin{align}
\begin{gathered}
    T=U_iU_i^\top,\qquad U_j=U_iS_{ij},\\
    S_{ij}\in O(N),\qquad S_{ij}S_{jl}S_{li}=1_N.
    \end{gathered}
    \label{eq:c2f-ai44-local-Takagi}
\end{align}
The transition matrix $S_{ij}=U_i^\dag U_j$ is real orthogonal: it is unitary and satisfies $S_{ij}S_{ij}^\top=1_N$.

In the local Takagi basis, $\widetilde H_i=U_i^\dag H U_i$ is real symmetric and therefore admits a real orthonormal frame of negative-energy states:
\begin{align}
    \begin{gathered}
    \widetilde H_i\widetilde\Phi_i=-\widetilde\Phi_i,\\
    \widetilde\Phi_i\in\Mat_{N\times n}(\R),\\
    \widetilde\Phi_i^\top\widetilde\Phi_i=1_n.
    \end{gathered}
    \label{eq:c2f-ai44-local-real-frame}
\end{align}
On an overlap, $S_{ij}\widetilde\Phi_j$ and $\widetilde\Phi_i$ are bases of the same negative-energy subspace, so
\begin{align}
    S_{ij}\widetilde\Phi_j=\widetilde\Phi_iV_{ij},\qquad
    V_{ij}=\widetilde\Phi_i^\top S_{ij}\widetilde\Phi_j\in O(n).
    \label{eq:c2f-ai44-local-transition}
\end{align}
Since $S_{ij}S_{jl}S_{li}=1_N$, we have
\begin{align}
    V_{ij}V_{jl}V_{li}=1_n.
    \label{eq:c2f-ai44-local-cocycle}
\end{align}
The matrices $V_{ij}$ therefore define the transition functions of a real bundle $V$. In the original basis, the occupied frames $\Phi_i=U_i\widetilde\Phi_i$ satisfy
\begin{align}
    \Phi_j=U_iS_{ij}\widetilde\Phi_j=\Phi_iV_{ij}.
    \label{eq:c2f-ai44-local-occupied-frame}
\end{align}
The occupied frames are glued by the same transition functions as the real bundle, so $E\simeq V\otimes_\R\C$.

For the complexification of a real bundle, the characteristic classes obey~\cite{Milnor-Stasheff}
\begin{align}
    c_2(E)=-p_1(V),\qquad 2c_1(E)=0.
    \label{eq:c2f-ai44-classes}
\end{align}
Since $H^2(T^4;\Z)$ is torsion-free, $c_1(E)=0$. Moreover, the cohomology of $T^4$ with $\Z_2$ coefficients is generated by degree-one classes whose squares vanish, so the square of any degree-two class also vanishes. Using $p_1(V)\bmod2=w_2(V)^2$, we obtain
\begin{align}
    \langle c_2(E),[T^4]\rangle
       &=-\langle p_1(V),[T^4]\rangle
       \label{eq:c2f-ai44-local-pontryagin}\\
       &\equiv\langle w_2(V)^2,[T^4]\rangle=0\pmod2.
\end{align}
Consequently,
\begin{align}
    {\rm ch}_2[H]=-\langle c_2(E),[T^4]\rangle\in2\Z.
    \label{eq:c2f-ai44-final}
\end{align}

Thus, an odd Chern character is impossible in the finite-dimensional framework, even with a momentum-dependent antiunitary symmetry matrix.

\subsubsection{Real-space AHSS}
The even-parity constraint can also be understood through the real-space AHSS. Consider a local model without translation symmetry and with an antiunitary symmetry acting as $\bx\mapsto-\bx$ in real space. The inversion center is a class-AI zero-cell with no protected zero-energy states. A one-cell has symmetry class A and can carry an integer number of chiral states.
A Chern insulator on a two-cell, together with its symmetry image, changes the number of chiral states on a one-cell by an even integer. The corresponding differential is
\begin{align}
    d^1_{2,-2}:\Z\longrightarrow\Z,\qquad n\longmapsto2n.
\end{align}
The remaining invariant on the one-cell is the parity $E^2_{1,-2}=\coker d^1_{2,-2} = \Z/2\Z$. With no contribution from three-cells, this parity survives to the $E^3$ page.

The four-dimensional even-parity constraint means that the odd part of the Chern number on a four-cell corresponds to the remaining boundary anomaly on the one-cell. Equivalently, the following map is nontrivial:
\begin{align}
\begin{gathered}
    d^3_{4,-4}:E^3_{4,-4}=\Z\longrightarrow E^3_{1,-2}=\Z/2\Z,\\
    n\longmapsto n\bmod2.
    \end{gathered}
    \label{eq:c2f-ai44-AHSS}
\end{align}
We use states bound to a monopole line to verify that the map is nontrivial.
A monopole in four-dimensional space is a line defect of codimension three.
For a four-dimensional Chern insulator with $|{\rm ch}_2|=1$ and no defects, one may take $(PT)^2=-1$. In a unit-monopole background, the same patch gluing as in Sec.~\ref{sec:3d-BDI-SPT-LSM} gives $(PT)^2=1$.

The states bound to the line can be obtained by extending the monopole zero modes of a three-dimensional chiral insulator into the fourth direction.
(States bound to monopole lines in four-dimensional lattice models have also been studied~\cite{TynerJuricic2024}.)
For a model with winding number of magnitude one, a unit monopole supports a zero mode satisfying $H_\perp\psi_0=0$ and $\G\psi_0=s\psi_0$ ($s=\pm1$)~\cite{AokiFukayaKanKoshinoMatsuki2023}.
Here, $H_\perp$ is the three-dimensional transverse Hamiltonian, and $\G$ is its chiral operator.
In the dimensional extension to a four-dimensional phase with $|{\rm ch}_2|=1$, adding motion in the fourth direction gives, at low energy,
\begin{align}
 H_4(k_4)=H_\perp+vk_4\G,\qquad
 H_4(k_4)\psi_0=svk_4\psi_0.
 \label{eq:c2f-ai44-monopole-chiral}
\end{align}
The zero mode thus becomes a state propagating in one direction along the monopole line.
As in Fig.~\ref{fig:4d-AII-SPT-LSM}[b], coupling the monopole-bound state to an oppositely propagating chiral state on the one-cell can open a gap while preserving $PT$ symmetry.
Two-cells contribute an even number of states. Canceling an odd number of the original states therefore requires an odd ${\rm ch}_2$, giving the SPT--LSM constraint in Eq.~\eqref{eq:c2f-ai44-AHSS}.

\section{\texorpdfstring{\NoCaseChange{Boundary actions induced by cellwise symmetries and ASPT}}{Boundary actions induced by cellwise symmetries and ASPT}}
\label{sec:oddmodel-construction}
\label{sec:aspt-realization}

Sections~\ref{sec:1d}--\ref{sec:c2f} established how the topology and locality of a symmetry action constrain the Chern or winding number of a gapped Hamiltonian. We now examine how the momentum-dependent actions that allow odd values can arise at the boundary of a higher-dimensional bulk from its cellwise symmetries. 

We construct the models with odd Chern or winding numbers considered in this paper as gapped boundaries of finite-range bulk models in one higher dimension.
When a cellwise bulk symmetry is restricted to the boundary subspace, the momentum dependence of the boundary states induces a non-cellwise symmetry action. Figure~\ref{fig:aspt-boundary-overview}(a) illustrates the construction.
For internal symmetries, deriving a non-onsite boundary action from an onsite bulk action is also used in SPT boundary theories~\cite{ChenLiuWen2011,ElseNayak2014}.
For a crystalline operation $\tau$, we likewise start from a cellwise symmetry that preserves the half-space: the symmetry does not reverse the added direction normal to the boundary.

Each construction reproduces the original symmetry action and Hamiltonian exactly.
In each case, we also show that the bulk can be deformed to an atomic limit while preserving the symmetries.
With the induced boundary symmetry action fixed, the allowed Chern or winding numbers have odd parity, in place of the even parity required by a cellwise action.
These realizations correspond to boundary constructions of ASPT phases in the following sense: a trivial bulk supports a boundary invertible phase that the conventional symmetry action would forbid~\cite{WangQiGu2019,WangGu2020}.
The constructions establish quasilocal boundary realizations for free fermions; we do not identify many-body anomaly indices in this paper.

We first treat classes D and C, then the chiral classes BDI, DIII, CII, and CI.
For the chiral classes, we restrict the construction to symmetry actions satisfying the condition in Eq.~\eqref{eq:oddmodel-scalar-energy}.
We have not obtained boundary realizations for classes AI and AII, consistent with the ``No (finite dim.)'' entries in Table~\ref{tab:result} and the prohibition of odd topological numbers.

\begin{figure*}[t]
\centering
\includegraphics[width=0.6\textwidth]{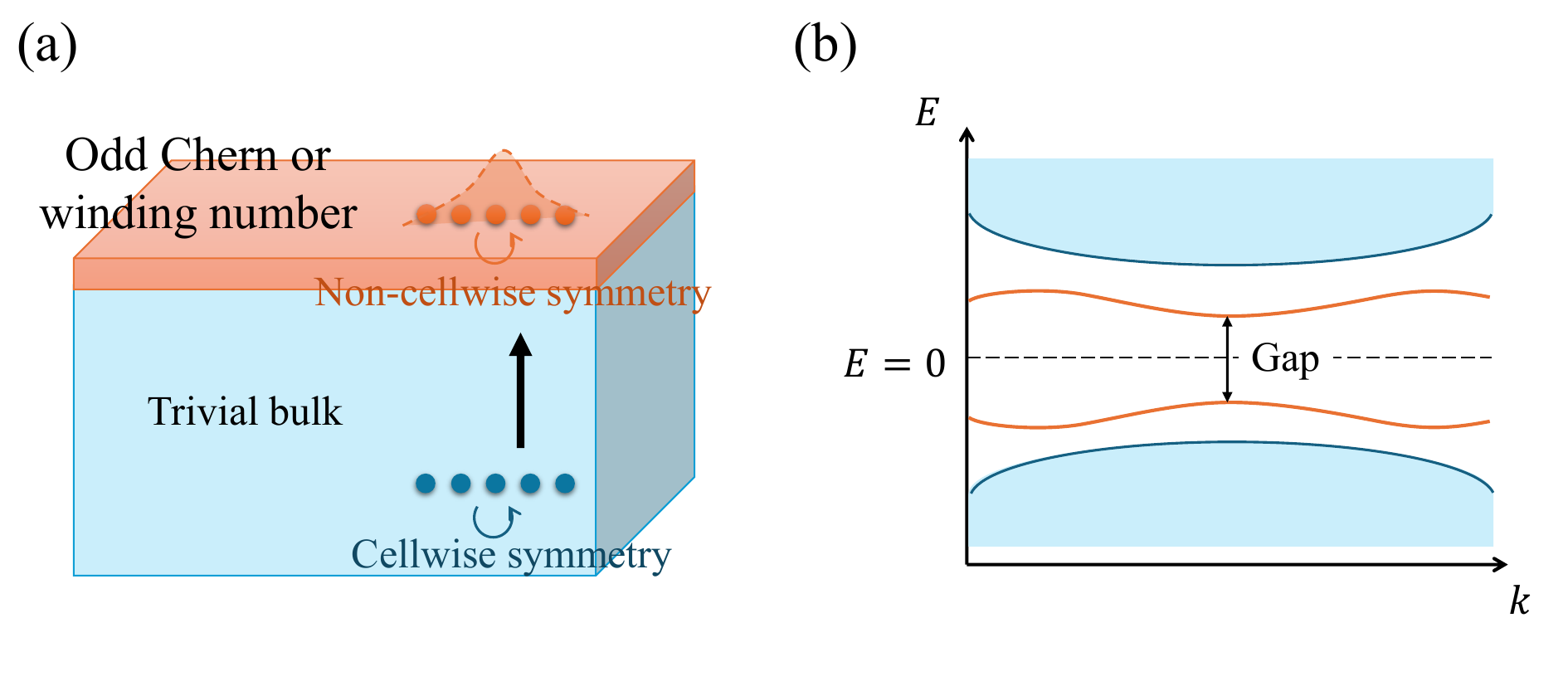}
\caption{Boundary realization of anomalous SPT phases.
(a) Restricting a cellwise symmetry of a finite-range bulk to the boundary subspace induces a non-cellwise boundary action.
The orange boundary supports a gapped phase with an odd Chern or winding number. The blue bulk can be deformed to an atomic limit while preserving the retained symmetries.
(b) Schematic half-space spectrum after the boundary is gapped.
For sufficiently small $\delta>0$, the orange boundary bands, described by $H^{\rm bdy}_{\bk}=\delta H_\bk$, remain separated from the blue bulk continua and are gapped at $E=0$ (dashed line).
The horizontal axis denotes momentum parallel to the boundary.}
\label{fig:aspt-boundary-overview}
\end{figure*}

\subsection{Boundary realizations for classes D and C}
\label{sec:oddmodel-CD}
\label{sec:general-aspt-boundary}

\subsubsection{Canonical construction of the models to be realized at the boundary}

Let $d$ be even.
Consider a gapped finite-range Hamiltonian $H_\bk$ with a cellwise PHS represented by a unitary matrix $C_0$ satisfying
\begin{align}
 C_0H_\bk^*C_0^\dagger&=-H_{-\tau\bk},
 &C_0^\top&=\epsilon_0 C_0.
 \label{eq:CD-seed-symmetry}
\end{align}
The spectrally flattened Hamiltonian satisfies
\begin{align}
 &Q_\bk:=\operatorname{sgn}H_\bk,\quad
 Q_\bk^2=I,\quad [Q_\bk,H_\bk]=0,\notag\\
 &C_0Q_\bk^*C_0^\dagger=-Q_{-\tau\bk}.
 \label{eq:standard-construction-Q}
\end{align}
Thus, defining
\begin{align}
 C_\bk&=-iQ_{-\tau\bk}C_0,\notag\\
 C_\bk H_\bk^*C_\bk^\dagger&=-H_{-\tau\bk},
 &C_\bk^\top&=\epsilon_C C_{-\tau\bk},\notag\\
 \epsilon_C&=-\epsilon_0
 \label{eq:oddmodel-CD}
\end{align}
reverses the square of the PHS. Class D with $\epsilon_0=1$ therefore becomes class C with $\epsilon_C=-1$, while class C with $\epsilon_0=-1$ becomes class D with $\epsilon_C=1$. Since the construction leaves $H$ unchanged, it preserves the Chern number.
For two-dimensional class C with internal symmetry, choose $C_0$ as the redundancy matrix of the BdG representation. The construction then agrees with $C_\bk=g_{-\bk}C_0$ derived from the non-onsite $\Z_4^F$ action $g_\bk=-iQ_\bk$ in Ref.~\cite{KobayashiInamuraShiozaki2026}.

The matrix $C$ constructed this way generally does not have finite range. Some symmetry classes nevertheless admit finite-range realizations, as shown in Table~\ref{tab:result}.
Examples are given in the preceding sections devoted to the individual classes.

\subsubsection{Parent model and particle-hole symmetry}

Using $H,C_0$ from Eq.~\eqref{eq:CD-seed-symmetry}, we apply the generalized SSH chain construction of Ref.~\cite{ShiozakiDetached2025} to define the $(d+1)$-dimensional bulk model
\begin{align}
 H_{d+1}(\bk,k_{d+1})
 &=\cos k_{d+1}\,\rho_x\otimes I
   -\sin k_{d+1}\,\rho_y\otimes H_\bk.
 \label{eq:aspt-parent-H}
\end{align}
Here,
$\rho_{x},\rho_{y},\rho_{z}$ are Pauli matrices acting on the two added internal orbitals.
If $H$ has finite range, the parent model has finite range in every direction.
The bulk is gapped because
\begin{align}
 H_{d+1}^2
 =\cos^2 k_{d+1}\,I
  +\sin^2 k_{d+1}\,(I_\rho\otimes H_\bk^2)>0.
\end{align}
The parent model has the PHS
\begin{align}
 C^{\rm bulk}&=\rho_y\otimes C_0,\notag\\
 C^{\rm bulk}H_{d+1}(\bk,k_{d+1})^*
 (C^{\rm bulk})^\dagger
 &=-H_{d+1}(-\tau\bk,-k_{d+1}),\notag\\
 C^{\rm bulk}(C^{\rm bulk})^*
 &=-\epsilon_0 I=\epsilon_C I.
 \label{eq:aspt-parent-symmetry}
\end{align}
The parent PHS $C^{\rm bulk}$ therefore exchanges the symmetry classes according to ${\rm D} \leftrightarrow {\rm C}$.
The PHS $C^{\rm bulk}$ is cellwise and does not reverse the added real-space direction $x_{d+1}$, even when a crystalline operation is involved.

The parent model also has an auxiliary chiral symmetry,
\begin{align}
 \Gamma^{\rm bulk}=\rho_z\otimes I,\qquad
 \{\Gamma^{\rm bulk},H_{d+1}\}=0.
 \label{eq:aspt-parent-chiral}
\end{align}
To gap the boundary, we break the auxiliary chiral symmetry while retaining the desired PHS.

\subsubsection{Momentum-dependent symmetry induced at the boundary}

For simplicity, we present the boundary-state construction only for
\begin{align}
 H_\bk^2=E_\bk^2I,\qquad E_\bk>0.
\end{align}
\footnote{For a general gapped $H_\bk$, replace the profiles below with functions of the positive-definite Hermitian matrix $|H_\bk|=(H_\bk^2)^{1/2}$. The eigenvalues $h,-h$ exchanged by PHS have the same decay profile, so the expression for the boundary symmetry action remains unchanged.}
Consider the half-infinite system $x_{d+1}\le0$, and let $P_\bk^{(\pm)}=(I\pm Q_\bk)/2$.
In the model of Eq.~\eqref{eq:aspt-parent-H}, each state with energy $+E_\bk$ gives a boundary zero mode in $\ket{\rho+}$, and each state with energy $-E_\bk$ gives one in $\ket{\rho-}$.
The two types of zero modes share a decay profile, described by a real scalar function $\eta_\bk(x_{d+1})$.
The function $\eta_\bk$ decays exponentially away from the boundary and is normalized by $\sum_{x_{d+1}\le0}\eta_\bk(x_{d+1})^2=1$.
The isometry into the boundary zero-mode subspace is
\begin{align}
 \begin{split}
  V_\bk v
  &=\sum_{x_{d+1}\le0}\eta_\bk(x_{d+1})\ket{x_{d+1}}\\
  &\quad\otimes\bigl(\ket{\rho+}\otimes P_\bk^{(+)}
       +\ket{\rho-}\otimes P_\bk^{(-)}\bigr)v.
 \end{split}
 \label{eq:aspt-boundary-isometry}
\end{align}
Denoting the half-space Hamiltonian by $\widehat H_{d+1,\bk}$, we obtain
\begin{align}
 V_\bk^\dagger V_\bk=I,\qquad
 \widehat H_{d+1,\bk}V_\bk=0.
 \label{eq:aspt-boundary-zero}
\end{align}
No local frame is needed to define the isometry, and $V_\bk$ is defined throughout the Brillouin zone.
Exactly flat boundary bands are a special feature of this model.

Since PHS gives $E_{-\tau\bk}=E_\bk$, we can choose identical profiles at $\bk$ and $-\tau\bk$. Consequently,
\begin{align}
 C_0P_\bk^{(\pm)*}&=P_{-\tau\bk}^{(\mp)}C_0,\notag\\
 \eta_{-\tau\bk}(x_{d+1})&=\eta_\bk(x_{d+1}).
 \label{eq:aspt-boundary-profile-phs}
\end{align}
The matrix $\rho_y$ exchanges $\ket{\rho+}$ and $\ket{\rho-}$. At the same time, $C_0K$ exchanges the positive- and negative-energy sectors while leaving the profile unchanged. Hence,
\begin{align}
 C^{\rm bulk}V_\bk^*
 =V_{-\tau\bk}\bigl(-iQ_{-\tau\bk}C_0\bigr),
 \label{eq:aspt-boundary-intertwining}
\end{align}
so PHS maps the boundary zero-mode subspace to itself. Restricting PHS and the auxiliary chiral symmetry to this subspace gives
\begin{align}
 C_\bk^{\rm bdy}
 &=V_{-\tau\bk}^\dagger C^{\rm bulk}V_\bk^*
 =-iQ_{-\tau\bk}C_0=C_\bk,\notag\\
 \Gamma_\bk^{\rm bdy}
 &=V_\bk^\dagger\Gamma^{\rm bulk}V_\bk=Q_\bk.
 \label{eq:aspt-induced-symmetry}
\end{align}
Restricting the cellwise bulk PHS to the boundary therefore reproduces the momentum-dependent PHS in Eq.~\eqref{eq:oddmodel-CD}.

\subsubsection{Detached topological boundary state}
If the Hamiltonian $H_\bk$ has a nonzero Chern number ${\rm ch}_{d/2}[H]$, a uniform gap cannot be opened in the flat boundary bands while preserving the auxiliary chiral symmetry $\Gamma^{\rm bulk}$.
To see why, let $q^{\rm bulk}$ denote the lower-left block of the flattened parent Hamiltonian in a basis that diagonalizes the auxiliary chiral symmetry $\Gamma^{\rm bulk}$. The relation $|W_{d+1}[q^{\rm bulk}]|=2|{\rm ch}_{d/2}[H]|$ shows that the boundary states are topologically nontrivial and cannot have a uniform gap.
Topologically nontrivial gapless boundary states that are spectrally separated from the bulk are called detached boundary states
~\cite{AltlandFragility2024,NakamuraDetached2025,ShiozakiDetached2025}.
For the model in Eq.~\eqref{eq:aspt-parent-H}, the obstruction to a uniform boundary gap can also be shown directly.

The occupied and unoccupied bands of $Q_\bk$ have Chern numbers ${\rm ch}_{d/2}[H]$ and $-{\rm ch}_{d/2}[H]$, respectively.
Let $H^{\rm bdy}_\bk$ denote the boundary Hamiltonian. Suppose, for contradiction, that $H^{\rm bdy}_\bk$ has a uniform gap. We can then define its flattening $Q^{\rm bdy}_{\bk} = \operatorname{sgn} H^{\rm bdy}_\bk$.
Chiral symmetry then gives $\{Q_\bk,Q^{\rm bdy}_\bk\}=0$, so $Q^{\rm bdy}_\bk$ exchanges the occupied and unoccupied bands of $Q_\bk$.
The action of $Q^{\rm bdy}_\bk$ preserves the Chern number, which would require ${\rm ch}_{d/2}[H] = -{\rm ch}_{d/2}[H]$. This contradicts the assumption that ${\rm ch}_{d/2}[H]$ is nonzero.

\subsubsection{Gapped boundary and trivialization of the bulk}
Breaking the auxiliary chiral symmetry $\Gamma^{\rm bulk}$ removes the obstruction to gapping the flat boundary bands.
Add to $H_{d+1}$ a perturbation proportional to $H_\bk$ with uniform strength along $x_{d+1}$:
\begin{align}
 H_{d+1,\delta}&=H_{d+1}+\delta \, I_\rho \otimes H_\bk,
 \qquad \delta>0.
 \label{eq:aspt-boundary-perturbation}
\end{align}
The added term breaks $\Gamma^{\rm bulk}$ but preserves the PHS represented by $C^{\rm bulk}$.
Because $(I_\rho\otimes H_\bk)V_\bk=V_\bk H_\bk$, the boundary Hamiltonian is exactly
\begin{align}
 \widehat H_{d+1,\delta,\bk}V_\bk=V_\bk(\delta H_\bk),
 \qquad H_\bk^{\rm bdy}=\delta H_\bk.
 \label{eq:aspt-gapped-boundary}
\end{align}
Each eigenvalue $h$ of $H_\bk$ produces two $(d+1)$-dimensional bulk energies, $\delta h \pm \sqrt{\cos^2 k_{d+1} + h^2 \sin^2 k_{d+1}}$. Choosing $\delta$ to satisfy $0<\delta\max_\bk\|H_\bk\|<\tfrac12\min\{1,\min_\bk\|H_\bk^{-1}\|^{-1}\}$\footnote{We use the operator norm induced by the Euclidean vector norm $\|\cdot\|_2$, namely $\|A\|:=\sup_{\|v\|_2=1}\|Av\|_2$. For a Hermitian matrix $H_\bk$, this gives $\|H_\bk\|=\max_{h\in{\rm Spec}H_\bk}|h|$ and, when the matrix is invertible, $\|H_\bk^{-1}\|^{-1}=\min_{h\in{\rm Spec}H_\bk}|h|$.} therefore gaps the boundary while keeping it separated from the bulk bands, producing a ``detached gapped boundary state.'' Figure~\ref{fig:aspt-boundary-overview}(b) shows the resulting spectrum schematically.
Multiplication by a positive constant leaves the occupied bands unchanged, so
\begin{align}
 {\rm ch}_{d/2}[H^{\rm bdy}]={\rm ch}_{d/2}[H].
 \label{eq:aspt-CD-chern}
\end{align}
The boundary thus has the desired momentum-dependent PHS $C^{\rm bdy}_\bk$ and odd Chern number ${\rm ch}_{d/2}[H]$.

The bulk can also be deformed to an atomic limit while retaining the desired PHS, provided that the auxiliary chiral symmetry is not imposed. To construct the deformation, first return the perturbation coefficient $\delta$ to zero without closing the bulk gap.
The matrix $\Gamma^{\rm bulk}$ anticommutes with the bulk Hamiltonian $H_{d+1}$ and satisfies the PHS relation
\begin{align}
C^{\rm bulk} (\Gamma^{\rm bulk})^* = -\Gamma^{\rm bulk} C^{\rm bulk},
\end{align}
so the bulk Hamiltonian can be continuously deformed to the trivial atomic limit $m \Gamma^{\rm bulk}$.
The deformation is not required to preserve either the boundary gap or the boundary subspace that is spectrally separated from the bulk.
Thus, whenever cellwise symmetries give a $2\Z$ classification in even-dimensional class D or C, an odd Chern number and a momentum-dependent PHS can be realized at the boundary of a trivial bulk.

\subsection{Boundary realizations for the chiral classes}
\label{sec:oddmodel-chiral}

In this subsection, we describe boundary realizations for classes BDI, DIII, CII, and CI.

\subsubsection{Models to be realized at the boundary}

Let $d$ be odd.
We start from an invertible matrix $D_\bk$ with odd winding number $W_d[D]\in2\Z+1$ and construct symmetry matrices $u,v$ satisfying Eq.~\eqref{eq:summary-chiral-symmetries}.
Choose $D_\bk$ to have an odd winding number and a finite Fourier expansion, and impose
\begin{align}
 D_\bk^\dagger D_\bk=D_\bk D_\bk^\dagger
 &=E_\bk^2I,\notag\\
 E_\bk>0,\qquad E_{-\tau\bk}&=E_\bk.
 \label{eq:oddmodel-scalar-energy}
\end{align}
\footnote{Boundary realizations of models that do not satisfy this condition are outside the scope of this paper.}
Here, $E_\bk$ is real and positive. For $d=1,3$, concrete choices are
\begin{align}
 d=1:\quad D_k&=\begin{pmatrix}e^{ik}&0\\0&1\end{pmatrix},\notag\\
 d=3:\quad D_\bk&=\left(-2+\sum_{j=1}^3\cos k_j\right)I_2
       +i\sum_{j=1}^3\sin k_j\,\sigma_j.
 \label{eq:oddmodel-3d}
\end{align}
Both choices have $W_d[D]=1$ and satisfy Eq.~\eqref{eq:oddmodel-scalar-energy}. In one dimension, adding the trivial block $1$ makes the matrix size even and allows the model to accommodate $\epsilon_T=-1$ as well.
Let $q_\bk$ be the flattening of $D_\bk$:
\begin{align}
    q_\bk=D_\bk/E_\bk.
\end{align}

Let $J$ be a constant unitary matrix satisfying
\begin{align}
 \begin{aligned}
  J^\top&=\epsilon_TJ,\\
  J&=I\quad (\epsilon_T=1),\\
  J&=\bigoplus i\sigma_y\quad (\epsilon_T=-1).
 \end{aligned}
 \label{eq:oddmodel-J}
\end{align}
When $\epsilon_T=-1$, we take $q$ to have even matrix size.
The condition for DIII and CI is $(q_\bk u_\bk)^\top=\epsilon_Tq_{-\tau\bk}u_{-\tau\bk}$, so it suffices to set
\begin{align}
 u_\bk=q_\bk^\dagger J
 \quad\Longrightarrow\quad q_\bk u_\bk=J.
 \label{eq:oddmodel-transpose}
\end{align}
For BDI and CII, we instead set
\begin{align}
 v_\bk=J,\qquad u_\bk=q_{-\tau\bk}Jq_\bk^\top.
 \label{eq:oddmodel-diagonal}
\end{align}
For DIII and CI, the construction gives $D_\bk u_\bk=E_\bk J$. For BDI and CII, it gives $u_\bk D_\bk^*J^\dagger=D_{-\tau\bk}$ and $u_\bk^\top=\epsilon_Tu_{-\tau\bk}$. These identities establish Eq.~\eqref{eq:summary-chiral-symmetries}.

We define the Hamiltonian, chiral symmetry, and TRS by
\begin{align}
 &H_\bk=\begin{pmatrix}0&D_\bk^\dagger\\D_\bk&0\end{pmatrix},
 \qquad \G=\begin{pmatrix}I&0\\0&-I\end{pmatrix},
 \notag\\
 &T_\bk = \begin{cases}
    \begin{pmatrix}J&0\\0&q_{-\tau\bk} J q_\bk^\top\end{pmatrix} & (\text{BDI,CII})\\
    \epsilon_T\begin{pmatrix}0&q_{-\tau\bk}^\dag J\\J q_\bk^*&0\end{pmatrix} & (\text{DIII,CI})
 \end{cases}.
 \label{eq:oddmodel-chiral-H}
\end{align}

\subsubsection{Finite-range parent model}
\label{sec:oddmodel-parent}

We now realize $H_\bk,T_\bk,\G$ from the preceding construction at the boundary of a finite-range bulk Hamiltonian.
The sign $c \in \{\pm 1\}, T_\bk\G^*=c\G T_\bk$ distinguishes BDI and CII from DIII and CI.

For the symmetries constructed above, multiplying $T_\bk$ by a suitable positive scalar function $f_\bk$ yields a finite-range matrix $F_\bk=f_\bk T_\bk$:
\begin{align}
 \begin{aligned}
  c=1:\quad F_\bk&=f_\bk T_\bk\\
  &=\begin{pmatrix}
  E_\bk^2J&0\\0&D_{-\tau\bk}JD_\bk^\top
  \end{pmatrix},\qquad f_\bk=E_\bk^2,\\
  c=-1:\quad F_\bk&=f_\bk T_\bk\\
  &=\epsilon_T\begin{pmatrix}
  0&D_{-\tau\bk}^\dagger J\\
  J D_\bk^*&0
  \end{pmatrix},\qquad f_\bk=E_\bk.
 \end{aligned}
 \label{eq:oddmodel-F-examples}
\end{align}
Because $E_\bk^2I=D_\bk^\dagger D_\bk$ is also a finite Fourier series, $F_\bk$ has finite range in both cases. In addition, $F_\bk$ obeys
\begin{align}
 F_\bk^\top&=\epsilon_T F_{-\tau\bk},\notag\\
 F_\bk H_\bk^*&=H_{-\tau\bk}F_\bk,
 &F_\bk\G^*&=c\G F_\bk.
 \label{eq:oddmodel-F-properties}
\end{align}

Introduce two auxiliary two-level systems, with Pauli matrices $\rho_\mu,\kappa_\mu$ ($\mu=x,y,z$), respectively. In $d+1$ dimensions, let $k_{d+1}$ denote the momentum along the added direction and define the generalized SSH model
\begin{align}
 K_\bk&=\begin{pmatrix}0&F_\bk^\top\\
                         F_\bk^*&0\end{pmatrix}_\kappa,
 &K_\bk^2&=f_\bk^2I,\notag\\
 B_{\bk,k_{d+1}}
 &=\frac{I-K_\bk}{2}+e^{ik_{d+1}}\frac{I+K_\bk}{2},\notag\\
 H_{d+1}(\bk,k_{d+1})
 &=\begin{pmatrix}0&B_{\bk,k_{d+1}}\\
                    B_{\bk,k_{d+1}}^\dagger&0\end{pmatrix}_{\rho}.
 \label{eq:oddmodel-parent}
\end{align}
The parent Hamiltonian depends only on $F$ and therefore has finite range in every direction. Its square gives
\begin{align}
 H_{d+1}(\bk,k_{d+1})^2
 &=\left(\cos^2\frac{k_{d+1}}2
     +f_\bk^2\sin^2\frac{k_{d+1}}2\right)I,\notag\\
 \Delta_{\rm bulk}&:=\min\{1,\min_\bk f_\bk\}>0,
 \label{eq:oddmodel-parent-gap}
\end{align}
so the bulk spectrum is separated from zero by at least $\Delta_{\rm bulk}$.

We use the momentum-independent matrices
\begin{align}
 \tilde T&=\begin{pmatrix}0&\epsilon_T I\\I&0\end{pmatrix}_\kappa,
 &\tilde \Gamma&=\begin{pmatrix}
    \G&0\\
    0&c\G^*
 \end{pmatrix}_\kappa,\notag\\
 T^{\rm bulk}&=I_\rho\otimes \tilde T,
 &\G^{\rm bulk}&=\rho_z\otimes \tilde \Gamma.
 \label{eq:oddmodel-bulk-symmetry}
\end{align}
Transposing $F_\bk\G^*=c\G F_\bk$ and using $\G^\dagger=\G$ yields $\G F_\bk^\top=cF_\bk^\top\G^*$. The resulting symmetry relations are
\begin{align*}
 \tilde TK_\bk^*\tilde T^\dagger&=K_{-\tau\bk},\qquad [\tilde \Gamma,K_\bk]=0,\\
 \tilde T\tilde \Gamma^*&=c\tilde \Gamma \tilde T,
\end{align*}
and $H_{d+1}$ has the symmetries
\begin{align}
 T^{\rm bulk}H_{d+1}(\bk,k_{d+1})^*
 (T^{\rm bulk})^\dagger
 &=H_{d+1}(-\tau\bk,-k_{d+1}),\notag\\
 \G^{\rm bulk} H_{d+1}(\bk,k_{d+1}) (\G^{\rm bulk})^\dag &= - H_{d+1}(\bk,k_{d+1}),\notag\\
 T^{\rm bulk}(T^{\rm bulk})^*&=\epsilon_T I,\notag\\
 T^{\rm bulk}(\G^{\rm bulk})^*&=c\G^{\rm bulk}T^{\rm bulk}.
 \label{eq:oddmodel-bulk-algebra}
\end{align}
Both symmetry matrices $T^{\rm bulk},\G^{\rm bulk}$ are momentum independent, so the parent model has cellwise symmetries.

\subsubsection{Actions induced at the boundary}

Restrict the layers to the half-space $x_{d+1}=n\leq0$ and replace $e^{\mp ik_{d+1}}$ by the shift operators $\ket{x_{d+1}}\mapsto\ket{x_{d+1}\pm1}$. Remove hopping terms that cross the boundary; in particular, the shift corresponding to $e^{-ik_{d+1}}$ sends $\ket0$ to zero.

First consider $f_\bk=1$. For each original degree of freedom, only one of the four auxiliary degrees of freedom survives at the boundary. To see this, write $\rho_z\ket{\rho\pm}=\pm\ket{\rho\pm}$. Each eigenvalue $\pm1$ of $K_\bk$ defines an SSH chain with the two degrees of freedom in the $\rho$ space. In the $K_\bk=+1$ sector, $B=e^{ik_{d+1}}$ couples $\ket{n,\rho-}$ to $\ket{n-1,\rho+}$, leaving only $\ket{0,\rho+}$ unpaired at the boundary. In the $K_\bk=-1$ sector, $B=I$ pairs $\rho+$ with $\rho-$ within each layer, so no unpaired degree of freedom remains.
For general $f_\bk>0$, the intercell coupling $(1+f_\bk)/2$ exceeds the intracell coupling $(1-f_\bk)/2$ in magnitude in the $K_\bk=+f_\bk$ sector. The inequality is reversed in the $K_\bk=-f_\bk$ sector. Boundary zero modes therefore occur only on the $\rho+$ side of the sector with the stronger intercell coupling. Their number equals the dimension of the original model before the auxiliary spaces were introduced.

An explicit expression for the zero modes follows from
\begin{align}
 \tilde V_\bk&=\frac{(i)^{\frac{1-\epsilon_T}{2}}}{\sqrt2}\begin{pmatrix}I\\T_\bk^*\end{pmatrix},
 \notag\\
 \tilde V_\bk^\dagger \tilde V_\bk&=I,\qquad K_\bk \tilde V_\bk=f_\bk \tilde V_\bk.
\end{align}
For $\epsilon_T=-1$, we choose the constant phase $i$ so that the induced TRS below is exactly $T_\bk$.
Let $\eta_\bk(x_{d+1})$ be a real, normalized decay function. The isometry into the boundary zero-mode subspace is
\begin{align}
 V_\bk v
 &=\sum_{x_{d+1} \leq0}\eta_\bk(x_{d+1})
   \ket{x_{d+1}}\otimes\ket{\rho+}\otimes \tilde V_\bk v,\quad V_\bk^\dagger V_\bk=I.
 \label{eq:oddmodel-boundary-frame}
\end{align}
For the half-space Hamiltonian $\widehat{H}_\bk$, the identity $\widehat{H}_\bk V_\bk=0$
shows that the boundary bands are flat at zero energy.
The rest of the half-space spectrum is separated from zero by at least $\Delta_{\rm bulk}$.

The identities $\tilde T\tilde V_\bk^*=\tilde V_{-\tau\bk}T_\bk$, $\tilde\Gamma\tilde V_\bk=\tilde V_\bk\G$, and $f_{-\tau\bk}=f_\bk$ determine the action of the bulk symmetries on the boundary subspace:
\begin{align}
 T^{\rm bulk}V_\bk^*
 &=V_{-\tau\bk}T_\bk,\notag\\
 \G^{\rm bulk}V_\bk&=V_\bk\G.
 \label{eq:oddmodel-induced}
\end{align}
The boundary therefore realizes the prescribed $T_\bk$ and constant $\G$ without modification.

\subsubsection{Boundary gap and trivialization of the bulk}
\label{sec:oddmodel-ASPT}

The finite-range Hamiltonian
\begin{align}
 \tilde H_\bk=\begin{pmatrix}
    H_\bk&0\\
    0&H_{-\tau\bk}^*
 \end{pmatrix}_\kappa
 \label{eq:oddmodel-mass}
\end{align}
obeys the following identities as a consequence of Eq.~\eqref{eq:oddmodel-F-properties}, the anticommutation relation $\{\G,H_\bk\}=0$, and the complex conjugate of that relation:
\begin{align}
 [K_\bk,\tilde H_\bk]&=0,
 &\tilde H_\bk \tilde V_\bk&=\tilde V_\bk H_\bk,\notag\\
 \tilde T \tilde H_\bk^*\tilde T^\dagger&=\tilde H_{-\tau\bk},
 &\{\tilde \Gamma,\tilde H_\bk\}&=0.
\end{align}
These identities imply that the perturbation $I_\rho\otimes\tilde H_\bk$ commutes with $H_{d+1}$ and preserves $T^{\rm bulk},\G^{\rm bulk}$. Adding the perturbation with equal strength in every layer gives
\begin{align}
 \widehat{H}_{\delta,\bk}
 &=\widehat{H}_\bk+\delta\, I_\rho\otimes \tilde H_\bk,\notag\\
      \widehat{H}_{\delta,\bk} V_\bk
 &=V_\bk H^{\rm bdy}_\bk, \quad H^{\rm bdy}_\bk =\delta H_\bk.
 \label{eq:oddmodel-gapped-boundary}
\end{align}
The perturbation $\delta I_\rho\otimes\tilde H_\bk$ breaks the auxiliary anticommutation relation with $\rho_z\otimes I$ but retains the physical chiral symmetry $\G^{\rm bulk}$ and TRS.

Let $M_H:=\max_\bk\|H_\bk\|$.
For $0<\delta M_H<\Delta_{\rm bulk}/2$, the boundary bands are gapped and separated from the rest of the half-space spectrum [Fig.~\ref{fig:aspt-boundary-overview}(b)]. The original $d$-dimensional Hamiltonian is thus reproduced as the positive rescaling $\delta H_\bk$, with the same symmetries $T_\bk,\G$ and winding number.

The bulk also admits a deformation to the atomic limit $\rho_x\otimes I$ that preserves TRS and the physical chiral symmetry. First return $\delta$ to zero, then follow the path
\begin{align*}
 H_{d+1}\longrightarrow\rho_z\otimes\tilde H
 \longrightarrow\rho_x\otimes I.
\end{align*}
For each arrow, interpolate between the adjacent Hamiltonians with weights $\cos\theta$ and $\sin\theta$, respectively ($0\leq\theta\leq\pi/2$). The Hamiltonians in the first pair anticommute because $[K,\tilde H]=0$; those in the second pair also anticommute. Each Hamiltonian in the displayed sequence is gapped and preserves $T^{\rm bulk},\G^{\rm bulk}$, so both interpolation paths have finite range and remain gapped. The boundary gap is not required to remain open along these paths.

Models that have odd winding numbers and satisfy Eq.~\eqref{eq:oddmodel-scalar-energy} can therefore be realized at the boundary of a trivial finite-range bulk. Both the parent Hamiltonian and the perturbation that gaps its boundary have finite range. The induced boundary symmetry $T_\bk=F_\bk/f_\bk$, however, is generally a quasilocal action with exponentially decaying matrix elements.

\section{\texorpdfstring{\NoCaseChange{Conclusions and outlook}}{Conclusions and outlook}}
\label{sec:conclusion}

We have related the parity of Chern and winding numbers of finite-dimensional Bloch Hamiltonians to the topology of momentum-dependent antiunitary symmetry actions. The relations cover all 15 cases in up to four spatial dimensions with a strong $2\Z$ classification for cellwise symmetries, and were derived with a cellwise chiral operator whenever chiral symmetry is present (Table~\ref{tab:result}). Whenever a fixed symmetry action enforces odd parity, any symmetric gapped phase must have nontrivial topology. For chiral classes in one and three dimensions, we also defined a $\Z_2$ invariant ${\rm WP}_d[H,\G]$ for momentum-dependent $\G_\bk$ in the fixed original Bloch basis. The parity invariant is independent of the eigenframes of $\G$ and obstructs a simultaneous deformation of $H$ and $\G$ to momentum-independent operators that preserves the gap and the chiral condition.

Explicit models and no-go theorems distinguish cases admitting finite-range symmetry actions with odd invariants from cases that admit quasilocal actions but no finite-range ones. In particular, we proved finite-range no-go theorems for the crystalline symmetries of three-dimensional CII and four-dimensional D. For the two-dimensional AII and four-dimensional AI and AII settings considered here, odd values are impossible for finite-dimensional Bloch Hamiltonians and symmetry matrices, even with arbitrary continuous momentum dependence.

Extending the results to quantum many-body systems remains an open problem.
Interactions or bosonic degrees of freedom may permit odd values excluded by our results.
For four-dimensional class AI with internal symmetry, one possible approach is to extend methods that extract indices from symmetries restricted to finite regions~\cite{KawagoeShirley2025,ShirleyZhangJiLevin2026,ChavdaKobayashi2026} to antiunitary actions.
For a system with onsite charge conservation, the question is whether such an index would require a four-dimensional quantum Hall response corresponding to an odd value of the free-fermion invariant ${\rm ch}_2$~\cite{QiHughesZhang2008}.
Constructing the index and realizing an odd response are both open problems.

A further extension would relax our assumptions on the symmetry algebra.
In recent lattice models with non-onsite $U(1)$ symmetries, vector and axial charges generate an Onsager algebra, and products of the corresponding transformations include actions that move degrees of freedom on the lattice.
The resulting lattice construction realizes a chiral anomaly in the continuum limit~\cite{ChatterjeePaceShao2025}.
The present analysis keeps the square and commutation relations of each AZ class fixed and excludes realizations in which products of internal symmetries generate additional nontrivial lattice translations.
Bloch phases arising from orbital positions remain allowed when they preserve the original algebraic relations.

Another direction is to study symmetry actions that enforce nontrivial phases in $\Z_2$ classifications, rather than constrain the parity of integer invariants. 
A two-dimensional DIII example~\cite{Lu2024MagneticLSM} has an odd number of Majorana Kramers pairs per primitive unit cell. 
TRS squares to fermion parity, and the commutator of magnetic translations is also fermion parity. 
Every symmetry-preserving gapped quadratic Hamiltonian then belongs to the nontrivial $\Z_2$ phase. 
A band-theoretic formulation of SPT--LSM constraints could determine when the topology and locality of a symmetry action enforce a nontrivial phase.

Physical applications require concrete realizations of non-cellwise symmetry actions with nontrivial invariants and compatible gapped phases, either in standalone systems or as ASPT phases at boundaries of higher-dimensional bulks. Relating the invariants to boundary transport and to states bound to defects or hinges is another important task. In our ASPT constructions, the bulk can be deformed to an atomic limit, yet its boundary carries a nontrivial symmetry action and supports a topologically nontrivial gapped phase. These constructions offer a setting in which boundary properties extend beyond the information contained in the bulk classification alone.

\begin{acknowledgments}
The author thanks Mikio Furuta for pointing out that momentum-dependent $C_\bk$ can allow an odd Chern number and for identifying the relation between the parity of the Chern number and the second Stiefel--Whitney class of the real bundle defined by $C_\bk$, as discussed in Sec.~\ref{sec:2d_classD}. Those observations motivated this work.
The author used GPT-6 Astra (OpenAI) to assist with proof development, literature searches, manuscript drafting and language editing, and the preparation of preliminary schematic figures.
Proposals generated by the model contributed specifically to the boundary construction for chiral classes in Sec.~\ref{sec:oddmodel-chiral}, the parallel-transport argument in Sec.~\ref{sec:c2f-ph-transport}, and the finite-range no-go proof in Appendix~\ref{app:d4-finite-range}.
The model also drew attention to the spin-structure dependence of the Chern--Simons action discussed in Appendix~\ref{sec:c2f-spin-note}.
The author verified and revised the AI-assisted material and takes full responsibility for the manuscript.
This work was supported by JSPS KAKENHI
Grant Nos. JP22H05118, JP26H01305, and JP26K00629.
\end{acknowledgments}

\appendix

\section{Symmetries in fermionic many-body systems}
\label{sec:summary-manybody}

This appendix summarizes how the single-particle symmetry conditions used in the main text can be realized in fermionic many-body systems.

In the many-body literature, an internal symmetry is called onsite if it is a product of actions on individual local degrees of freedom~\cite{ChenLiuWen2011,Seifnashri2024}. A non-onsite internal symmetry can also admit a symmetric gapped invertible phase~\cite{ShirleyZhangJiLevin2026}. We use cellwise for both internal and crystalline symmetries whose matrices are momentum independent in the fixed unit-cell orbital basis. The term onsite retains its many-body meaning and is not extended to crystalline symmetries that move degrees of freedom in space.
Related work on interacting fermions has studied ASPT states, which are symmetric gapped states realized on the boundary of a topologically trivial bulk~\cite{WangQiGu2019,WangGu2020}. Section~\ref{sec:aspt-realization} discusses their relation to the boundary realizations of momentum-dependent symmetries considered here.

An AZ class imposes conditions on the matrix $H_\bk$, but the same conditions can arise from different many-body symmetry actions~\cite{Wen2012}.
In particular, particle-hole symmetry of complex fermions with $U(1)$ symmetry must be distinguished from the particle-hole redundancy of Nambu fermions (see, e.g., Ref.~\cite{ChiuTeoSchnyderRyu2016}).
We first construct a particle-number-conserving realization for every class considered here, then discuss more general fermionic realizations.

\subsection{Particle-number-conserving realizations}

Consider a finite periodic system in which each component of $\hat c_\bk$ annihilates an independent complex fermion. Define the particle-number-conserving free-fermion Hamiltonian by
\begin{align}
 \widehat H=\sum_\bk\left(
 \hat c_\bk^\dag H_\bk\hat c_\bk-\frac12\tr H_\bk\right),
 \qquad \widehat P_f=(-1)^{\widehat F},
 \label{eq:summary-fock-H}
\end{align}
where $\widehat F = \sum_\bk \hat c_\bk^\dag \hat c_\bk$ is the total particle number. The constant term ensures the same energy reference before and after particle-hole exchange. With $\epsilon_T,\epsilon_C\in\{+1,-1\}$ as in Table~\ref{tab:AZ}, Eqs.~\eqref{eq:summary-TRS} and \eqref{eq:summary-PHS} are realized, respectively, by
\begin{align}
 \widehat{\mathcal T}\hat c_\bk\widehat{\mathcal T}^{-1}
 &=T_\bk^\dag\hat c_{-\tau\bk},
 &\widehat{\mathcal T}i\widehat{\mathcal T}^{-1}&=-i,\notag\\
 \widehat{\mathcal C}\hat c_\bk\widehat{\mathcal C}^{-1}
 &=C_\bk^\top(\hat c_{-\tau\bk}^\dag)^\top,
 &\widehat{\mathcal C}i\widehat{\mathcal C}^{-1}&=i,
 \label{eq:summary-fock-actions}
\end{align}
and both many-body symmetry operators commute with $\widehat H$. The operators $\widehat{\mathcal T}$ and $\widehat{\mathcal C}$ implement antiunitary time reversal and unitary particle-hole exchange, respectively.
Their squares act on the fermion operators as
\begin{align}
 \widehat{\mathcal T}^2\hat c_\bk\widehat{\mathcal T}^{-2}
 &=\epsilon_T\hat c_\bk,\qquad
 \widehat{\mathcal C}^2\hat c_\bk\widehat{\mathcal C}^{-2}
 =\epsilon_C\hat c_\bk.
 \label{eq:summary-fock-squares}
\end{align}
Thus, $+1$ in Table~\ref{tab:AZ} denotes the identity action, while $-1$ denotes the action of fermion parity $\widehat P_f$.
With these many-body actions, classes AI and AII can be realized as particle-number-conserving systems with time reversal squaring to $1$ and $\widehat P_f$, respectively. Classes D and C can be realized with particle-hole exchange squaring to $1$ and $\widehat P_f$, respectively. Classes BDI, DIII, CII, and CI have both symmetries, whose squares are $(1,1)$, $(\widehat P_f,1)$, $(\widehat P_f,\widehat P_f)$, and $(1,\widehat P_f)$, respectively.

Chiral symmetry in a particle-number-conserving system can be realized as antiunitary particle-hole exchange.
For the momentum-dependent condition~\eqref{eq:summary-chiral}, take
\begin{align}
 \widehat{\mathcal S}\hat c_\bk\widehat{\mathcal S}^{-1}
 &=\G_\bk^\top(\hat c_\bk^\dag)^\top,\qquad
 \widehat{\mathcal S}i\widehat{\mathcal S}^{-1}=-i.
 \label{eq:summary-fock-chiral}
\end{align}
Combining complex conjugation with the reordering of creation and annihilation operators transforms the single-particle matrix in Eq.~\eqref{eq:summary-fock-H} as $H_\bk\mapsto-\G_\bk H_\bk\G_\bk$. Invariance of the many-body Hamiltonian is therefore equivalent to the chiral condition. Moreover, $\G_\bk^2=I$ implies that $\widehat{\mathcal S}^2$ acts trivially on fermion operators. The combination of antiunitarity and particle-hole exchange leaves the momentum label unchanged.
A momentum-independent $\G_\bk$ gives a cellwise action confined to each unit cell. In general, the action mixes creation operators in different cells.

The geometric cell motion factored out in Eq.~\eqref{eq:summary-fock-actions} is $\bm R\mapsto\tau\bm R$ in both cases.
The case $d_\tau=0$ corresponds to an internal action, the case $d_\tau=1$ involves a reflection, and the cases $d_\tau=2,3$ involve reversal of two and three coordinates, respectively.
The two-dimensional AII and D settings in this paper can be realized, respectively, by time reversal combined with a reflection and by unitary particle-hole exchange combined with spatial inversion.
If $T_\bk,C_\bk$ are momentum independent, the actions are cellwise: orbitals in cell $\bm R$ are mapped only to orbitals in cell $\tau\bm R$. Momentum dependence describes additional motion and mixing between cells. Finite range means that the support stays within a finite distance of $\tau\bm R$, independent of system size; for a quasilocal action, the matrix elements decay exponentially with distance from the mapped cell. Even ordinary crystalline symmetries can produce additional cell displacements, depending on the orbital positions and the chosen unit cell. The locality conditions in Sec.~\ref{sec:summary-locality} concern the single-particle action and are distinct from realizability by a finite-depth local circuit.

\subsection{BdG Hamiltonians and additional symmetries}
\label{sec:summary-bdg-symmetries}

The Nambu spinor
$\hat\Psi_\bk=(\hat c_\bk,(\hat c_{-\bk}^\dag)^\top)^\top$
contains redundant particle and hole components. A BdG Hamiltonian therefore satisfies the identity
\begin{align}
 \Sigma_x H_\bk^*\Sigma_x=-H_{-\bk},\qquad
 \Sigma_x=\begin{pmatrix}0&I\\I&0\end{pmatrix}.
 \label{eq:summary-bdg-redundancy}
\end{align}
The identity expresses the Nambu redundancy, not a physical symmetry~\cite{ChiuTeoSchnyderRyu2016}.
A superconductor subject only to this redundancy belongs to class D. Adding physical time reversal without a spatial action gives class BDI if its square acts as $1$ on fermion operators, and class DIII if its square acts as $\widehat P_f$. The time-reversal matrix may depend on momentum.

Using the Nambu redundancy, one can express an additional PHS as a unitary symmetry and an additional TRS as a unitary matrix action that reverses the sign of the energy. Both correspondences are invertible.
Indeed, Eqs.~\eqref{eq:summary-PHS} and \eqref{eq:summary-TRS} give
\begin{align}
 g_\bk:=C_{-\bk}\Sigma_x,\qquad
 &g_\bk H_\bk g_\bk^\dag=H_{\tau\bk},
 \label{eq:summary-bdg-spatial}\\
 \widetilde\G_\bk:=T_{-\bk}\Sigma_x,\qquad
 &\widetilde\G_\bk H_\bk\widetilde\G_\bk^\dag=-H_{\tau\bk}.
 \label{eq:summary-bdg-trs-product}
\end{align}
Conversely, unitary matrices $g_\bk,\widetilde\G_\bk$ satisfying the respective Hamiltonian relations define PHS and TRS through
\begin{align}
 C_\bk=g_{-\bk}\Sigma_x,\qquad
 T_\bk=\widetilde\G_{-\bk}\Sigma_x.
 \label{eq:summary-bdg-partner-inverse}
\end{align}
Since $\Sigma_x$ is constant, these correspondences preserve the cellwise, finite-range, or quasilocal character of the matrices.

A many-body realization must also respect the Nambu constraint, which follows from $\Sigma_x(\hat\Psi_\bk^\dag)^\top = \hat\Psi_{-\bk}$. We therefore also impose
\begin{align}
 \Sigma_xX_\bk^*\Sigma_x=X_{-\bk},\qquad
 X=C,T,g,\widetilde\G.
 \label{eq:summary-bdg-trs-compatibility}
\end{align}
The correspondences between PHS or TRS and unitary matrix actions preserve Nambu compatibility. Multiplying $C_\bk$ by a constant phase leaves the AZ PHS condition and the sign of its square unchanged, but can affect Nambu compatibility. We therefore choose a compatible phase when interpreting the matrices in BdG systems.
On the many-body Hilbert space, $g_\bk$ corresponds to a unitary symmetry $\widehat{\mathcal U}$ and $T_\bk$ to an antiunitary symmetry $\widehat{\mathcal T}$, acting as
\begin{align}
 \widehat{\mathcal U}\hat\Psi_\bk\widehat{\mathcal U}^{-1}
 &=g_\bk^\dag\hat\Psi_{\tau\bk},
 &\widehat{\mathcal U}i\widehat{\mathcal U}^{-1}&=i,\notag\\
 \widehat{\mathcal T}\hat\Psi_\bk\widehat{\mathcal T}^{-1}
 &=T_\bk^\dag\hat\Psi_{-\tau\bk},
 &\widehat{\mathcal T}i\widehat{\mathcal T}^{-1}&=-i.
 \label{eq:summary-bdg-trs-action}
\end{align}
Both many-body operations preserve the Hamiltonian. By contrast, $\widetilde\G_\bk$ exchanges the positive and negative energies of the BdG matrix; the corresponding symmetry of the many-body Hamiltonian is the antiunitary operator $\widehat{\mathcal T}$.

The correspondences also preserve the signs of the squares. If $C_\bk^\top=\epsilon_C C_{-\tau\bk}$ and $T_\bk^\top=\epsilon_T T_{-\tau\bk}$, the compatibility condition gives
\begin{align}
 g_{\tau\bk}g_\bk=\epsilon_C I,\qquad
 \widetilde\G_{\tau\bk}\widetilde\G_\bk=\epsilon_T I.
 \label{eq:summary-bdg-partner-squares}
\end{align}
Thus, an additional unitary $\Z_2$ symmetry with $g_{\tau\bk}g_\bk=I$ gives a PHS squaring to $+1$. If $g_{\tau\bk}g_\bk=-I$, the square acts as fermion parity on fermion operators. This is a $\Z_4^F$-type action and corresponds to a PHS squaring to $-1$.

For $d_\tau=0$, $\widetilde\G_\bk$ is a chiral operator up to a phase. Setting $\G_\bk=z\widetilde\G_\bk$, with $z=1$ for $\epsilon_T=+1$ and $z=i$ for $\epsilon_T=-1$, gives
\begin{align}
 \G_\bk^\dag=\G_\bk,\qquad
 \G_\bk^2=I,\qquad
 \G_\bk H_\bk\G_\bk=-H_\bk.
 \label{eq:summary-bdg-chiral-normalization}
\end{align}
The normalized operator obeys $\Sigma_x\G_\bk^*\Sigma_x=\epsilon_T\G_{-\bk}$. Conversely, a chiral operator obeying this Nambu relation defines TRS through $T_\bk=z^{-1}\G_{-\bk}\Sigma_x$.
The resulting $\G_\bk$ generally depends on momentum. Section~\ref{sec:summary-chiral-general-frame} develops the framework for such operators: for gapped systems in $d=1,3$, it constructs frames that define an off-diagonal block $q_\bk$ and introduces the parity invariant ${\rm WP}_d[H,\G]$ in the fixed original orbital basis. For $d_\tau>0$, Eq.~\eqref{eq:summary-bdg-trs-product} also involves the spatial action $\tau$ and must be distinguished from a chiral operator that anticommutes with $H_\bk$ at the same momentum.

The same correspondences express additional unitary symmetries and crystalline chiral symmetries as momentum-dependent PHS or TRS. Although we do not classify systems subject to both Nambu redundancy and an additional symmetry, the parity constraints derived from the additional PHS or TRS alone remain valid when the redundancy is also imposed.
For example, Ref.~\cite{KobayashiInamuraShiozaki2026} constructs a non-onsite $\Z_4^F$ action $g_\bk=-iQ_\bk$ from the flattened Hamiltonian $Q_\bk=H_\bk(H_\bk^2)^{-1/2}$. The corresponding PHS matrix is $C_\bk=g_{-\bk}\Sigma_x=-iQ_{-\bk}\Sigma_x$, which satisfies the Nambu compatibility condition in this phase convention. Details are given in Sec.~\ref{sec:2dclassC-model}.

\section{Chern--Simons terms and spin structures}
\label{sec:c2f-spin-note}
On a general closed oriented four-manifold, $\int_X{\rm ch}_2(E)=\tfrac12\langle c_1(E)^2,[X]\rangle-\langle c_2(E),[X]\rangle$ can be half-integral. An extension to an oriented four-manifold therefore does not, by itself, define an absolute CS invariant for ${\rm ch}_2$ $\bmod1$. On a spin manifold, Wu's formula ensures that the integral of $c_1^2$ is even, allowing a CS term defined $\bmod1$.
We therefore fix a spin structure $\mathfrak s$ on $Y$, extend it together with the bundle and connection to a spin four-manifold $Z$, and define
\begin{align}
\begin{gathered}
    {\rm CS}_3[A;\mathfrak s]
       :=-\frac1{8\pi^2}\int_Z\tr F(\widetilde A)^2\pmod1,\\
    \partial(Z,\widetilde{\mathfrak s})=(Y,\mathfrak s).
    \end{gathered}
    \label{eq:c2f-ch2-spin-definition}
\end{align}
An extension exists because $\Omega^{\rm Spin}_3(BU(n))=0$. The resulting CS invariant is independent of the chosen extension.
When a global frame exists, the CS invariant agrees $\bmod1$ with the integral of the usual local form
\begin{align}
    {\rm cs}_3(A)=-\frac1{8\pi^2}\tr\left(A\,dA+\frac23A^3\right).
\end{align}

However, ${\rm CS}_3[A;\mathfrak s]$ may depend on the spin structure $\mathfrak s$.
Changing the spin structure by $\alpha\in H^1(Y;\Z_2)$ gives the transformation law~\cite{MooreChernSimons2019}
\begin{align}
    &{\rm CS}_3[A;\mathfrak s+\alpha]-{\rm CS}_3[A;\mathfrak s]\notag \\
       &\qquad\equiv\frac12\langle\alpha\cup(c_1(E)\bmod2),[Y]\rangle
       \pmod1.
    \label{eq:c2f-ch2-spin-change}
\end{align}
Here, the $\Z_2$ values $0,1$ on the right-hand side represent $0,1/2\in\R/\Z$, respectively.

We derive the transformation law directly from Eq.~\eqref{eq:c2f-ch2-spin-definition}.
Let $Z_0,Z_1$ be spin extensions of $(Y,\mathfrak s,E)$ and $(Y,\mathfrak s+\alpha,E)$, respectively. Forget the spin structures and glue the extensions as $X=Z_1\cup_Y(-Z_0)$, identifying the same bundle and connection on the boundary. Denote the glued bundle by $E_X$.
The boundary spin structures need not agree, so $X$ need not be spin.
The difference of the two CS values equals $\int_X{\rm ch}_2(E_X)$ $\bmod1$. Writing $\bar c_1:=c_1(E_X)\bmod2$, the integrality of $c_2$ and Wu's formula give
\begin{align}
 \int_X{\rm ch}_2(E_X)
 &\equiv\frac12\langle c_1(E_X)^2,[X]\rangle\notag\\
 &\equiv\frac12\langle w_2(TX)\cup\bar c_1,[X]\rangle\pmod1.
\end{align}

We construct a two-cocycle representing $w_2(TX)$ in terms of $\alpha$ using local spin lifts of the transition functions.
All cochains below have $\Z_2$ coefficients.
For the transition functions $g_{ij}\in SO(3)$ on $Y$, fix lifts defining one spin structure,
$\tilde g_{ij}\in{\rm Spin}(3)$, satisfying
$\tilde g_{ij}\tilde g_{jk}\tilde g_{ki}=1$.
Specify the second spin structure relative to the first by
$\tilde h_{ij}=\tilde g_{ij}(-1)^{a_{ij}}$,
$a\in Z^1(Y;\Z_2)$,
$[a]=\alpha$.

Represent the extensions of these spin structures to $Z_0,Z_1$ by the transition functions
$\tilde g^{(Z_0)}_{ij},\tilde h^{(Z_1)}_{ij}\in{\rm Spin}(4)$.
The extended transition functions satisfy $\tilde g^{(Z_0)}_{ij}\tilde g^{(Z_0)}_{jk}
 \tilde g^{(Z_0)}_{ki}=1$, $\tilde h^{(Z_1)}_{ij}\tilde h^{(Z_1)}_{jk}
 \tilde h^{(Z_1)}_{ki}=1$, respectively.
To compensate for the difference $a$ between the boundary spin structures, extend $a$ to $Z_1$
as a one-cochain,
$\tilde a\in C^1(Z_1;\Z_2)$, with $\tilde a|_Y=a$.
(Since $\Omega^{\rm SO}_3(B\Z_2)=\Z_2$, an extension as a one-cocycle need not exist,
but an extension as a one-cochain always does.)

On $Z_1$, introduce the sign-corrected lifts
\begin{align}
 \tilde h^{\prime(Z_1)}_{ij}
 :=\tilde h^{(Z_1)}_{ij}(-1)^{\tilde a_{ij}}.
\end{align}
The corrected lifts satisfy $\tilde h^{\prime(Z_1)}_{ij}
 \tilde h^{\prime(Z_1)}_{jk}
 \tilde h^{\prime(Z_1)}_{ki}
 =(-1)^{(d\tilde a)_{ijk}}$.
Although the corrected lifts are not generally transition functions of a spin structure,
they provide local spin lifts of the $SO(4)$ transition functions.
On the boundary,
$\tilde h^{\prime(Z_1)}|_Y
=\tilde h(-1)^a=\tilde g$,
so the corrected lifts glue to $\tilde g^{(Z_0)}$.
We thus obtain local spin lifts of the transition functions over all of $X$.

The obstruction to the cocycle condition is zero on $Z_0$
and $d\tilde a$ on $Z_1$.
Since $d\tilde a|_Y=da=0$, the two local obstructions glue into the two-cocycle
\begin{align}
 z_X|_{Z_0}=0,\qquad z_X|_{Z_1}=d\tilde a.
\end{align}
By the definition of the obstruction class for a spin lift,
$[z_X]=w_2(TX)$.

Using $d\bar c_1=0$, we obtain, with $\Z_2$ coefficients,
\begin{align}
 \langle w_2(TX)\cup\bar c_1,[X]\rangle
 &=\braket{d\tilde a\cup\bar c_1,[Z_1]}\notag\\
 &=\braket{a\cup\bar c_1,[Y]}.\notag
\end{align}
This proves Eq.~\eqref{eq:c2f-ch2-spin-change}.

In the boundary calculations for classes AI and AII, the main text uses $\mathcal Q_3$ associated with $c_2$ to avoid the complications caused by the spin-structure dependence of the ${\rm CS}$ term.

\section{Proof of Theorem~\ref{thm:finite_range_d43}}
\label{app:d4-finite-range}
Following the proof of Theorem~\ref{thm:CII-finite-range}, we enlarge the unit cell in the $x_1$ direction and multiply the symmetry matrix by $e^{i n k_1}$ so that we may assume
\begin{align}
    C_\bk = A(k_2,k_3,k_4) + B(k_2,k_3,k_4) e^{ik_1}.
\end{align}
The transformation $C_\bk\mapsto e^{ink_1}C_\bk$ preserves the symmetry \eqref{eq:c2f-d43-symmetry}. When Eq.~\eqref{eq:alt64-zeta-phase-shift} is applied at both boundaries, the trivial-bundle corrections cancel in the difference because $C_0,C_\pi$ have the same matrix size. Hence,
\begin{align}
 \begin{aligned}
 &\zeta[e^{ink_1}C_0;\mathfrak s]-\zeta[e^{ink_1}C_\pi;\mathfrak s]\\
 &\quad\equiv\zeta[C_0;\mathfrak s+n\ell_1]
                  -\zeta[C_\pi;\mathfrak s+n\ell_1]\\
 &\quad\equiv\zeta[C_0;\mathfrak s]-\zeta[C_\pi;\mathfrak s]\pmod2.
 \end{aligned}
 \label{eq:app-d4fr-phase-invariance}
\end{align}
The last congruence holds because, when a gapped $H$ exists, the winding number $W_3[h]$ on the right-hand side of Eq.~\eqref{eq:c2f-d43-zeta-transformation} is independent of the spin structure.
The second Chern number ${\rm ch}_2$ is also invariant under unit-cell enlargement. The right-hand side of the desired $\bmod2$ constraint \eqref{eq:c2f-d43-final} is therefore unchanged by unit-cell enlargement.

Write $p=(k_2,k_3)$.
Unitarity of $C_\bk$ and the symmetry $C_\bk^\top = C_{-\tau\bk}$ give
\begin{align}
 AB^\dagger&=A^\dagger B=0,\notag\\
 AA^\dagger+BB^\dagger
 &=A^\dagger A+B^\dagger B=I,\notag\\
 A(p,k_4)^\top&=A(p,-k_4),\qquad B(p,k_4)^\top=B(p,-k_4).
 \label{eq:app-d4fr-C-partial}
\end{align}
Thus, $P_A = AA^\dagger, P_B = BB^\dagger$ are complementary orthogonal projectors, and $A,B$ have constant ranks.
The image bundles carry the antiunitary actions
$v\mapsto A(p,-k_4)\bar v$ and $v\mapsto B(p,-k_4)\bar v$, respectively.
These actions relate the fibers at $k_4$ and $-k_4$, and each squares to the identity.
On the fixed planes $k_4=b\in\{0,\pi\}$, the actions therefore define the real bundles
\begin{align}
 (E_{A,b})_p&=\{v\in\operatorname{Ran}A(p,b):A(p,b)\bar v=v\},\notag\\
 (E_{B,b})_p&=\{v\in\operatorname{Ran}B(p,b):B(p,b)\bar v=v\}.
\end{align}
With these antiunitary actions, the image bundles of $A,B$ can be viewed as three-dimensional $C_2T$-symmetric bundles.

For a three-dimensional $C_2T$-symmetric insulator with vanishing first Chern class, the difference of the second SW numbers on the two fixed planes $b=0,\pi$ is known to give a strong invariant. The nontrivial phase is a second-order phase with chiral hinge states~\cite{AhnYang2019}.
This known strong invariant suggests that the quantity $\zeta[C_0;\mathfrak s]-\zeta[C_\pi;\mathfrak s]$ likewise characterizes the second-order phase of the $C_2T$-symmetric systems defined by $A,B$.
Since $A,B$ define complementary image bundles, only one of them supplies independent data.
We now establish a direct relation between the two invariants without assuming a vanishing first Chern class.

\begin{lem}
\label{lem:app-d4fr-SW-boundary}
If a gapped $H$ exists, then, for any spin structure $\mathfrak s$,
\begin{align}
 \begin{aligned}
 &\zeta[C_0;\mathfrak s]-\zeta[C_\pi;\mathfrak s]\\
 &\quad\equiv\langle w_2(E_{B,0})- w_2(E_{B,\pi}),[T^2]\rangle\pmod2.
 \end{aligned}
 \label{eq:app-d4fr-SW-boundary}
\end{align}
\end{lem}
\begin{proof}
By Eq.~\eqref{eq:c2f-d43-zeta-transformation}, the difference on the left-hand side is independent of the spin structure,
so we use $\mathfrak s_{\rm NS}$, which is antiperiodic in all directions.
On each fixed plane $b=0,\pi$, the bundles $\operatorname{Ran}A_b$ and $\operatorname{Ran}B_b$ have real structures, so their first Chern classes vanish.
We may therefore choose global complex orthonormal frames $V_{A,b},V_{B,b}$ for the two bundles.
Setting $V_b=(V_{A,b},V_{B,b})$, we can write
\begin{align}
 X_{A,b}&=V_{A,b}^\dagger A_b V_{A,b}^*,\qquad
 X_{B,b}=V_{B,b}^\dagger B_b V_{B,b}^*,\notag\\
 C_b(k_1,p)&=V_b(p)\operatorname{diag}
       (X_{A,b}(p),e^{ik_1}X_{B,b}(p))V_b(p)^\top.
 \label{eq:app-d4fr-C-block}
\end{align}
The matrices $X_{A,b},X_{B,b}$ are symmetric and unitary,
and the second matrix defines a real bundle isomorphic to $E_{B,b}$.
Since $V_b$ depends only on $p$, we have $W_3[V_b]=0$.
Equation~\eqref{eq:alt64-nu-transformation} and additivity under direct sums therefore give
\begin{align}
 \zeta[C_b;\mathfrak s_{\rm NS}]
 &\equiv\zeta[X_{A,b};\mathfrak s_{\rm NS}]
       +\zeta[e^{ik_1}X_{B,b};\mathfrak s_{\rm NS}]\pmod2.
 \label{eq:app-d4fr-zeta-split}
\end{align}
Since $X_{A,b},X_{B,b}$ are independent of $k_1$,
Eq.~\eqref{eq:alt64-zeta-two-variable} gives
$\zeta[X_{A,b};\mathfrak s_{\rm NS}]
=\zeta[X_{B,b};\mathfrak s_{\rm NS}]=0$.

Write the first SW numbers of $E_{B,b}$ as
\begin{align}
 b_{j,b}:=\langle w_1(E_{B,b}),[\gamma_j]\rangle,
 \qquad j=2,3.
 \label{eq:app-d4fr-first-SW}
\end{align}
Here, $\gamma_j$ is the circle in the $j$th direction of $T^2_{k_2,k_3}$.
Pulling $E_{B,b}$ back to $T^3$ gives $w_1(E_{B,b})=b_{2,b}\ell_2+b_{3,b}\ell_3$.
Equation~\eqref{eq:alt64-trivial-xi-spin} gives
$\xi_{\rm triv}[\mathfrak s_{\rm NS}]
=\xi_{\rm triv}[\mathfrak s_{\rm NS}+\ell_1]=0$,
so the trivial-bundle correction in the phase-factor transformation law \eqref{eq:alt64-zeta-phase-shift} vanishes.
Furthermore, the spin-structure transformation law \eqref{eq:alt64-zeta-spin-change} gives
\begin{align}
 \zeta[e^{ik_1}X_{B,b};\mathfrak s_{\rm NS}]
 &\equiv\zeta[X_{B,b};\mathfrak s_{\rm NS}+\ell_1]\notag\\
 &\equiv\langle w_2(E_{B,b}),[T^2]\rangle
       +b_{2,b}b_{3,b}\pmod2.
 \label{eq:app-d4fr-spin-correction}
\end{align}
The product correction comes from
$\xi_{\rm triv}[\mathfrak s_{\rm NS}+\ell_1+b_{2,b}\ell_2+b_{3,b}\ell_3]
=b_{2,b}b_{3,b}$. Thus,
\begin{align}
 \zeta[C_b;\mathfrak s_{\rm NS}]
 &\equiv\langle w_2(E_{B,b}),[T^2]\rangle
       +b_{2,b}b_{3,b}\pmod2.
 \label{eq:app-d4fr-zeta-reduction}
\end{align}

We next show that the correction terms are equal at the two boundaries, using the isomorphism of the corresponding real bundles.
The existence of a gapped $H$ implies the relation in Eq.~\eqref{eq:c2f-d43-h},
$C_\pi=hC_0h^\top$,
so the real bundles associated with $\mathcal P_{C_0}$ and $\mathcal P_{C_\pi}$ are isomorphic.
Let $L_1$ be the real line bundle with a sign reversal after one circuit of the $k_1$ circle.
By Eq.~\eqref{eq:app-d4fr-C-block}, the real bundle associated with $\mathcal P_{C_b}$ is isomorphic to
$E_{A,b}\oplus(L_1\otimes E_{B,b})$.
We leave the pullbacks of $E_{A,b},E_{B,b}$ to $T^3$ implicit.

Let $r_B=\operatorname{rank}B$.
Using $w_1(L_1)=\ell_1$, $\ell_1^2=0$, and the Whitney sum formula,
we obtain, on the mixed plane $T^2_{k_1,k_j}$,
\begin{align}
 b_{j,b}
 &\equiv\left\langle w_2(\mathcal P_{C_b}),[T^2_{k_1,k_j}]\right\rangle\notag\\
 &\quad+r_B\left\langle w_1(\mathcal P_{C_b}),[\gamma_j]\right\rangle
 \pmod2,\qquad j=2,3.
 \label{eq:app-d4fr-first-SW-bundle}
\end{align}
Indeed, the contribution to the second SW number containing $\ell_1$ is
$r_B\langle w_1(E_{A,b}),[\gamma_j]\rangle+(r_B-1)b_{j,b}$.
Adding the term involving the first SW number leaves $b_{j,b}$.
The right-hand side of Eq.~\eqref{eq:app-d4fr-first-SW-bundle} is the same at
$b=0,\pi$ because the boundary bundles are isomorphic. Therefore,
\begin{align}
 b_{j,0}&=b_{j,\pi}\quad(j=2,3),\notag\\
 b_{2,0}b_{3,0}&=b_{2,\pi}b_{3,\pi}\pmod2.
 \label{eq:app-d4fr-first-SW-match}
\end{align}
Subtracting the two boundary values in Eq.~\eqref{eq:app-d4fr-zeta-reduction}
cancels the correction terms and proves the claim.
\end{proof}

The problem is now to evaluate the strong index of the $C_2T$ insulator defined by $\operatorname{Ran} B$,
\begin{align}
    \Delta[E_B]:= \langle w_2(E_{B,0})- w_2(E_{B,\pi}),[T^2]\rangle.
\end{align}
Equations~\eqref{eq:app-d4fr-SW-boundary} and \eqref{eq:c2f-d43-final} give $\Delta[E_B]\equiv{\rm ch}_2[H]\pmod2$. Invariance of the second Chern number under unit-cell enlargement therefore implies that $\Delta[E_B]$ is unchanged by enlargement in the $x_4$ direction.
Choose the unit cell large enough that $B$ and both products $BB^\dagger,B^\dagger B$
have range one in the $x_4$ direction.
Consider the orthogonal projector
\begin{align}
 \begin{aligned}
 P&:=U^\dagger\Pi U,\\
 \Pi&=\frac12
 \begin{pmatrix}
 BB^\dagger&B\\
 B^\dagger&B^\dagger B
 \end{pmatrix},\\
 U&=\frac1{\sqrt2}
 \begin{pmatrix}
 I&iI\\ I&-iI
 \end{pmatrix}.
 \end{aligned}
\end{align}
The projector property of $BB^\dag$ and the transpose symmetry imply
\begin{align}
 P^2=P=P^\dagger,\qquad
 P(p,k_4)^*=P(p,-k_4).
\end{align}
The map $v\mapsto U^\dagger(v,B^\dagger v)^\top/\sqrt2$
transforms the antiunitary action $B$ on $\operatorname{Ran}B$ into complex conjugation.
Thus, the real bundle $F_b$ defined by the image of $P(p,b)$ is isomorphic to $E_{B,b}$ ($b=0,\pi$).
With the chosen unit cell, the projector has the expansion
\begin{align}
 \begin{aligned}
 &P(p,k_4)=M(p)+N(p)e^{ik_4}+N(p)^\top e^{-ik_4},\\
 &M^\top=M=M^*, \quad N^*=N.
 \end{aligned}
 \label{eq:app-d4fr-P-partial}
\end{align}
The equality of the boundary SW numbers follows from the next lemma.

\begin{lem}
Assume finite range in the $p$ directions. Let $F_b$ be the real bundles defined by the images of the orthogonal projectors $P(p,b=0,\pi)$ in Eq.~\eqref{eq:app-d4fr-P-partial}. Then $\langle w_2(F_0),[T^2]\rangle=\langle w_2(F_\pi),[T^2]\rangle$.
\end{lem}

\begin{proof}
The projector condition gives
\begin{align}
 N^2=O,\quad MN+NM=N.
\end{align}
If $N(p) \equiv O$, then $P(p,0)=P(p,\pi)$, and the claim follows.
Assume $r_N=\max_{p\in T^2}\operatorname{rank}N(p)>0$.
At a point $p \in T^2$ where the rank is maximal, choose $r_N$ independent columns spanning $\operatorname{Ran} N(p)$. Let $V(p)$ be the matrix formed by these columns.
The function $d(p):= \det (V(p)^\top V(p))$ is a finite Laurent polynomial by the finite-range assumption, and the choice of columns ensures that it is not identically zero:
\begin{align}
    d(p) = \sum_m d_m e^{i m \cdot p}.
\end{align}
Choose $a\in\mathbb R^2$ so that the values $a\cdot m$ are all distinct. Extend the domain to complex momenta and shift $p\mapsto p-i\lambda a$. For sufficiently large $\lambda$,
the term $d_m e^{\lambda m \cdot a}$ with maximal $a\cdot m$ uniformly dominates the sum of the remaining terms.
Thus, for some $\lambda_*>0$,
\begin{align}
 d(p-i\lambda_*a)\ne0\qquad(p\in T^2).
\end{align}
The minors of $N$ of order $(r_N+1)$ vanish identically as Laurent polynomials,
so $\operatorname{rank}N\le r_N$ also holds at complex momenta.
At the endpoint $\lambda=\lambda_*$, however, $d\ne0$ implies that
the $r_N$ columns of $V$ are independent.
The endpoint matrix $N$ therefore has rank $r_N$,
and the columns of $V$ span its image.

The deformation to complex momentum preserves the isomorphism classes of the real bundles on the fixed planes $k_4 = 0,\pi$. To see this, set
\begin{align}
 Q_{b,\lambda}(p)
 :=2P(p-i\lambda a,b)-I.
\end{align}
The transpose symmetry remains valid at complex momentum $p$, so the Laurent-polynomial identities
$Q_{b,\lambda}^\top=Q_{b,\lambda}$ and
$Q_{b,\lambda}^2=I$ hold.
The real and imaginary parts of $Q_{b,\lambda}$ are therefore real symmetric, and
\begin{align}
 (\re Q_{b,\lambda})^2=I+(\im Q_{b,\lambda})^2\ge I.
\end{align}
Thus, $\re Q_{b,\lambda}$ has no eigenvalues in $(-1,1)$. Its positive-eigenvalue bundle remains isomorphic to the bundle at $\lambda=0$, namely $F_b$.

For the rest of the proof, evaluate $M,N,V$ at the complex momentum $p-i\lambda_*a$. Define the projector onto $\operatorname{Ran} N$ by
\begin{align}
 R=V(V^\top V)^{-1}V^\top, \quad R^2=R, \quad R^\top = R. 
\end{align}
The definition gives $RN=N$. Every vector $v$ in $\operatorname{Ran} N$ has the form $v=Nu$, and $Nv=N^2u=O$, so $NR=O$.
Moreover, $MN=N(I-M)$ implies that $M$ preserves $\operatorname{Ran}N$, giving $MR=RMR$. Taking the transpose also gives $RM=RMR$.
Together, these relations give
\begin{align}
    RN=N,\qquad NR=0,\qquad [R,M]=0. 
\end{align}
Now introduce
\begin{align}
    S:= I-2 R.
\end{align}
Then $S^\top=S, S^2=I$, and
\begin{align}
    \begin{aligned}
    SMS^{-1}&=M,\\
    SNS^{-1}&=-N,\\
    S N^\top S^{-1}&=-N^\top.
    \end{aligned}
\end{align}
Thus, $S$ is complex orthogonal and relates the matrices on the two fixed planes $k_4=0,\pi$ by a similarity transformation:
\begin{align}
    SQ_{0,\lambda_*}S^{-1} = Q_{\pi,\lambda_*}.
\end{align}

It remains to compare the real bundles on the two fixed planes.
Because $S$ is generally not real, the similarity transformation does not directly define an isomorphism of real bundles.
The polar decomposition provides a path to a real orthogonal matrix. We use this path to deform the associated real symmetric matrices.

Let $S=UH, H = \sqrt{S^\dagger S}$ be the polar decomposition, where $U$ is unitary and $H$ is positive-definite Hermitian.
Since $S^\top S=I$ implies $H^\top=H^{-1}$ and $UU^\top=I$, the matrix $U$ is real orthogonal.
Moreover, the family
\begin{align}
 S_u:=UH^u,\qquad 0\le u\le1
\end{align}
consists of complex orthogonal matrices, with $S_0=U$ and $S_1=S$.
Along this path, define
\begin{align}
 \widetilde Q_u:=S_uQ_{0,\lambda_*}S_u^{-1}.
\end{align}
Then $\widetilde Q_u^\top=\widetilde Q_u$ and $\widetilde Q_u^2=I$. Hence,
$\operatorname{Re}\widetilde Q_u,\operatorname{Im}\widetilde Q_u$ are real symmetric and satisfy $(\operatorname{Re}\widetilde Q_u)^2=I+(\operatorname{Im}\widetilde Q_u)^2\ge I$.
The real symmetric matrix $\operatorname{Re}\widetilde Q_u$ therefore stays gapped at zero throughout the deformation,
and its positive-eigenvalue real bundle has a fixed isomorphism class as $u$ varies.
At the endpoints,
\begin{align}
 \re \widetilde Q_0 &=U(\operatorname{Re}Q_{0,\lambda_*})U^\top,\notag\\
 \re \widetilde Q_1&=\operatorname{Re}Q_{\pi,\lambda_*}.\notag
\end{align}
Since $U$ is real orthogonal, the real bundle at the start of this path is isomorphic to
the positive-eigenvalue real bundle of $\operatorname{Re}Q_{0,\lambda_*}$.
The deformation $\operatorname{Re}\widetilde Q_u$ preserves the isomorphism class between the two endpoints. Hence, the positive-eigenvalue real bundles of $\operatorname{Re}Q_{0,\lambda_*}$ and $\operatorname{Re}Q_{\pi,\lambda_*}$ are isomorphic.

Returning from complex to real momenta then gives $F_0\simeq F_\pi$.
Consequently,
\begin{align}
 \Delta[E_B]
 &=\langle w_2(E_{B,0})- w_2(E_{B,\pi}),[T^2]\rangle
 \equiv0\pmod2. \qedhere
\end{align}

\end{proof}

\bibliography{PT_evenodd_cellwise}


\end{document}